\documentclass{article}

\usepackage[utf8]{inputenc}
\usepackage[english]{babel}
\usepackage{csquotes}

\usepackage[letterpaper, portrait, margin=1in]{geometry}

\usepackage{enumitem}

\usepackage{graphicx}
\usepackage{wrapfig}
\usepackage{booktabs}
\usepackage{multirow}
\usepackage{multicol}
\usepackage{float}

\usepackage{hyperref}
\hypersetup{
  colorlinks=true,
  linkcolor=blue,
  filecolor=magenta,
  urlcolor=cyan,
  citecolor=blue
}

\usepackage[authordate,backend=biber,natbib]{biblatex-chicago}

\usepackage{amsmath}
\usepackage{amsfonts}
\usepackage{amssymb}
\usepackage{amsthm}

\newcommand{\indep}{\raisebox{0.05em}{\rotatebox[origin=c]{90}{$\models$}}}

\numberwithin{equation}{section}

\theoremstyle{definition}
\newtheorem*{definition*}{Definition}
\newtheorem*{intuition*}{Intuition}
\newtheorem*{assumption*}{Assumption}
\newtheorem*{remark*}{Remark}
\newtheorem*{example*}{Example}
\newtheorem*{exercise*}{Exercise}

\theoremstyle{plain}
\newtheorem*{theorem*}{Theorem}
\newtheorem*{proposition*}{Proposition}
\newtheorem*{claim*}{Claim}
\newtheorem*{conjecture*}{Conjecture}

\theoremstyle{definition}

\newtheorem{asm}{Assumption}

\theoremstyle{plain}

\theoremstyle{definition}
\newtheorem{definition}{Definition}[section]

\newtheorem{assumption}{Assumption}[section]
\newtheorem{remark}{Remark}[section]
\newtheorem{example}{Example}[section]

\theoremstyle{plain}
\newtheorem{theorem}{Theorem}[section]

\newtheorem{lemma}{Lemma}[section]
\newtheorem{corollary}[theorem]{Corollary}

\usepackage[reftex]{theoremref}
\usepackage{bm}
\usepackage{mathrsfs}

\title{Efficient Semiparametric Estimation of Average Treatment Effects Under
Covariate Adaptive Randomization}
\author{Ahnaf Rafi
\footnote{The author would like to thank Ivan Canay, Joel
Horowitz, Eric Auerbach and Federico Bugni. Their advice, suggestions and
comments have been invaluable in the development and writing of the
paper. Additional thanks go to Yong Cai, Deborah Kim, Filip Obradovic,
Amilcar Velez as well as the participants of the Spring 2021 and Spring 2022
Econometrics Reading Groups at Northwestern University for their suggestions and
comments.}
\\
Department of Economics \\
Northwestern University \\
\url{ahnafrafi2023@northwestern.edu}
}
\date{\today}

\begin{document}

\maketitle

\begin{abstract}
Experiments that use covariate adaptive randomization (CAR) are commonplace in
applied economics and other fields. In such experiments, the experimenter first
stratifies the sample according to observed baseline covariates and then assigns
treatment randomly within these strata so as to achieve balance according to
pre-specified stratum-specific target assignment proportions. In this paper, we
compute the semiparametric efficiency bound for estimating the average treatment
effect (ATE) in such experiments with binary treatments allowing for the class
of CAR procedures considered in
\citet{2018bugniInferenceCovariateAdaptiveRandomization,
2019bugniInferenceCovariateAdaptiveRandomization}. This is a broad class of
procedures and is motivated by those used in practice. The stratum-specific
target proportions play the role of the propensity score conditional on all
baseline covariates (and not just the strata) in these experiments. Thus, the
efficiency bound is a special case of the bound in
\citet{1998hahnRolePropensityScore}, but conditional on all baseline
covariates. Additionally, this efficiency bound is shown to be achievable under
the same conditions as those used to derive the bound by using a cross-fitted
Nadaraya-Watson kernel estimator to form nonparametric regression adjustments.
\end{abstract}

\noindent%

\hfill

\noindent%

\newpage

% ! TEX root = ../2023_semipar_eff_car.tex

\section{Introduction}

Experiments that use covariate adaptive randomization (CAR) are commonplace in
applied economics and other fields. In such experiments, the experimenter first
groups sample units according to observed baseline covariates (stratification)
and then assigns treatment randomly to achieve \emph{balance} within these
groups (strata). The term balance here means that the experimenter additionally
specifies stratum-specific target treatment proportions and assigns treatment so
that the proportion of units assigned to treatment reaches the corresponding
target as the sample size grows across strata. A textbook treatment of covariate
adaptive and stratified randomization in clinical trials can be found in
\citet{2015rosenbergerRandomizationClinicalTrials}. For review articles on their
use in development economics, see \citet{2007dufloChapter61Using},
\citet{2009bruhnPursuitBalanceRandomization} and
\citet{2017atheyChapterEconometricsRandomized}. In experiments comparing
outcomes from binary treatments, the quantity of interest is often the average
treatment effect (ATE). In this paper, we are concerned with efficient
semiparametric estimation of the ATE while allowing for a broad class of CAR
procedures motivated by those used in practice. We have two main questions. The
first is: is there a well-defined semiparametric efficiency bound
(SPEB) for this class of procedures? The second is: does there exist a
semiparametric estimator that achieves this bound and under what conditions will
this happen? Our answers to both are affirmative. For the first, we show that
a version of the bound in \citet{1998hahnRolePropensityScore} is valid. For the
second, we show that under the same (weak) conditions used for derivation of the
bound, there is a feasible estimator that achieves the bound asymptotically.

In randomized experiments, correctly implemented randomization ensures that in
expectation, confounding factors are equally distributed across treatment arms
so that differences in outcomes are solely due to the differences in
treatment. Stratified randomization additionally aims to ensure that along
observable dimesions, this also remains true in practice. In CAR experiments,
both discrete and continuously distributed covariates are combined to form
strata. As shown in \citet[Section 7.2]{2017atheyChapterEconometricsRandomized},
the main statistical benefit of stratified randomization for ATE estimation is
improved precision. The standard recommendation for ATE estimation in stratified
experiments is to regress the observed outcomes on indicators for strata and
their interactions with treatment status in a linear regression equation (a
fully saturated linear regression model). The coefficients on the interaction
terms from this regression are then combined with sample stratum proportions to
construct the ATE estimate. The resulting estimator of the ATE is analogous to
the Horvitz-Thompson estimator
(\citet{1952horvitzGeneralizationSamplingReplacement}) and does not use
information beyond the strata. However, experimenters concerned with estimation
precision may want to use information in the baseline covariates not contained
in the strata.

An additional complication in the CAR context comes from the choice of treatment
assignment mechanism by the experimenter. Many popular treatment assignment
mechanisms used in practice aim for faster targeting of the target assignment
proportions than simple independent and identically distributed (i.i.d.)
assignment and in doing so, induce dependence in the observed outcomes through
dependence in the treatment assignments. One such example is stratified permuted
block randomization (SPBR,
\th\ref{eg--sbr}). \citet{2018bugniInferenceCovariateAdaptiveRandomization}
provide additional examples. The dependence induced by CAR designs can affect
the behavior of ATE estimators in surprising ways. Analyzing the case where
target assignment proportions are constant across strata,
\citet{2018bugniInferenceCovariateAdaptiveRandomization} show that the standard
difference in treatment and control group means can have a limit variance that
depends explicitly on the choice of treatment assignment mechanism. For example,
all else held equal, the same estimator exhibits a larger limit variance under
i.i.d. treatment assignments than when treatments are assigned according to
SPBR, even when both mechanisms have the same target
proportions. \citet{2018bugniInferenceCovariateAdaptiveRandomization} also show
that the same phenomenon holds true of the ``stratum fixed effects'' estimator
which is recommended by \citet{2009bruhnPursuitBalanceRandomization}.
\citet{2019bugniInferenceCovariateAdaptiveRandomization}
extends this work and consider both multiple treatments and target assignment
proportions that vary by strata. They show that the estimator of the ATE from a
fully saturated regression has a limit variance that depends on the target
proportions (among other things), but not on the particular choice of assignment
mechanism. They also show that for the stratum fixed effects estimator however,
the limit variance is still affected explicitly by the choice of assignment
mechanism. \citet{2018bugniInferenceCovariateAdaptiveRandomization,
2019bugniInferenceCovariateAdaptiveRandomization} do not consider efficiency
questions.

There is a literature that considers efficiency gains in experiments from using
baseline covariates via linear regression adjustments. However, whether there
are gains at all depends on the linear regression specification.
In the case without stratification,
\citet{2008freedmanRegressionAdjustmentsExperimental}
shows that linear regression adjustments (without interactions between treatment
status and baseline covariates) can hurt asymptotic precision of the ATE
estimates though the estimates remain consistent. However, work by
\citet{2001yangEfficiencyStudyEstimators} and
\citet{2013linAgnosticNotesRegression} show that appropriately formed regression
adjustments (i.e. with the correct interaction terms) cannot hurt (and indeed
can improve) asymptotic precision in ATE estimation. For CAR experiments,
\citet{2022maRegressionAnalysisCovariate} build on the results of
\citet{2018bugniInferenceCovariateAdaptiveRandomization,
2019bugniInferenceCovariateAdaptiveRandomization} and show that
correctly formed linear regression adjustments cannot hurt (and can improve)
asymptotic precision under CAR. These works do not treat semiparametrically
efficient estimation.

In this paper, we are concerned with the semiparametrically efficient estimation
of the ATE in the CAR framework under the minimal set of assumptions laid out by
\citet{2019bugniInferenceCovariateAdaptiveRandomization}. There is a large
statistics literature around semiparametric efficiency, starting with the
seminal work of \citet{1956steinEfficientNonparametricTesting}. Most of these
are developed assuming i.i.d. data.
\citet{1998bickelEfficientAdaptiveEstimation} offers a comprehensive textbook
treatment and \citet{1990neweySemiparametricEfficiencyBounds} provides an
approachable review. Notable examples with non-i.i.d. data can be found in
\citet{2001bickelInferenceSemiparametricModels},
\citet{2004greenwoodIntroductionEfficientEstimation} and
\citet{2010komunjerSemiparametricEfficiencyBound} as well in references
therein. For treatment effects, \citet{1998hahnRolePropensityScore} derives the
SPEB for the ATE in observational studies with i.i.d. data when treatment
assignment is ignorable conditional on observable covariates (ignorable as in
\citet[Section 1.3]{1983rosenbaumCentralRolePropensity}). The
\citet{1998hahnRolePropensityScore} bound does not apply immediately to our
setting since the treatment assignment rule is allowed to depend on the entire
profile of sample strata. For instance, it is not clear \emph{a priori} if the
choice of assignment mechanism will affect the efficiency bound since it can
clearly influence limit behavior of estimators as in
\citet{2018bugniInferenceCovariateAdaptiveRandomization,
2019bugniInferenceCovariateAdaptiveRandomization}. Furthermore, even when
treatment assignments are i.i.d. in a CAR context, the covariates being
conditioned on during assignment are the strata, so it is again not immediate
from the bound in \citet{1998hahnRolePropensityScore} what role the additional
baseline covariates can play in providing efficiency gains. This paper shows
that a version of the \citet{1998hahnRolePropensityScore} bound accounting for
both stratum-specific target proportions and all baseline covariates is the SPEB
across all CAR experimental procedures. The target proportions play the role of
the propensity score conditional on baseline covariates. Additionally, the
choice of assignment mechanism does not affect this bound, only the target
proportions do. We derive this by using the partial sums arguments of
\citet{2018bugniInferenceCovariateAdaptiveRandomization,
2019bugniInferenceCovariateAdaptiveRandomization} to show that the log
likelihood ratios under \(1 / \sqrt{n}\) local alternatives has the local
asymptotic normality (LAN) property of
\citet{1960lecamLocallyAsymptoticallyNormal}.

In concurrent work, \citet{2022armstrongAsymptoticEfficiencyBounds} proves the
LAN property under a larger class of experimental designs using martingale
methods. The experiments considered there include the CAR framework and
additionally allows for arbitrary dependence of the treatment rule on the
covariates as well as observations of past outcomes (to account for sequential
assignment). \citet{2022armstrongAsymptoticEfficiencyBounds}'s main goal is to
show that an optimized form of the \citet{1998hahnRolePropensityScore} bound is
a lower bound on limit variance of the ATE across all experimental designs in
that class. This is motivated by recent papers on optimal design of experiments
by using past waves to either allocate treatment sequentially
(e.g. \citet{2011hahnAdaptiveExperimentalDesign}) or to form optimal strata in a
main experiment (e.g. \citet{2022tabord-meehanStratificationTreesAdaptive},
\citet{2022baiOptimalityMatchedPairDesigns} and
\citet{2021cytrynbaumDesigningRepresentativeBalanced}). The optimized bound is
analogous to implementing a ``conditional on covariates'' Neyman allocation
(\citet{1934neymanTwoDifferentAspects}). We differ from
\citet{2022armstrongAsymptoticEfficiencyBounds} work on two counts. First, while
their lower bound holds across the designs the aforementioned class, they do not
provide efficiency bounds in specific instances within the class. We provide
efficiency bounds for a given fixed stratification scheme and we do
not consider the question of optimizing the bound. As a result our efficiency
bound as a variance lower bound is higher than the optimized one in
\citet{2022armstrongAsymptoticEfficiencyBounds} and hence sharper for a given
fixed stratification scheme. Second,
\citet{2022armstrongAsymptoticEfficiencyBounds} does not consider the question
of when the bounds are achievable. We do this for the CAR framework explicitly
as explained in the subsequent two paragraphs.

Once an efficiency bound is established, the question of its sharpness remains,
in the sense of whether or not it is achievable. A well known phenomenon in the
literature is that conditions under which a SPEB is achievable are typically
much stronger than those required to derive the
bound. \citet{1990ritovAchievingInformationBounds} provide counterexamples where
a finite and non-singular SPEB exists but in certain non-trivial submodels, even
consistent (let alone efficient) estimation of the parameter of interest is
impossible. One of their examples is the partially linear model
(\citet{1986engleSemiparametricEstimatesRelation},
\citet{1988robinsonRootNConsistent}) which is commonly used in applied
economics. A well known condition in the literature is that \(1 /
\sqrt{n}\)-consistent and asymptotically normal (\(1 / \sqrt{n}\)-CAN) two-step
semiparametric estimators require any first step infinite-dimensional
(henceforth nonparametric) nuisance parameters to be estimated at a rate faster
than \(n^{- \frac{1}{4}}\) where \(n\) is the sample size (see for instance
\citet{1994neweyAsymptoticVarianceSemiparametric} and
\citet{2003chenEstimationSemiparametricModels}). This is mainly because
nonparametric estimators exhibit considerable bias and this condition limits the
effect of this bias on the second estimation step. The \(n^{- \frac{1}{4}}\)
rate condition however is especially restrictive when the covariates are of
higher dimension, due to the curse of dimensionality. Achieving this rate
condition typically requires imposing smoothness and/or Donsker conditions
(often dimension dependent) on unknown nuisance parameters and the estimator in
question. For the ATE, \citet{1998hahnRolePropensityScore} shows that
semiparametrically efficient estimation can be done under regularity conditions
through nonparametric regression adjustments or imputation. In their first step,
series estimators of conditional means and the propensity score are
used. Estimation by kernel methods can also be done, see
\citet{2004imbensNonparametricEstimationAverage} for a review. In all of these,
the use of smoothness restrictions on nonparametric population unknowns is
ubiquitous.

An additional contribution of this paper is to show that the SPEB derived for
the ATE is achievable under the same conditions used for its derivation. We do
this by leveraging the \emph{efficient influence function} to form an estimating
equation (or moment condition) for the ATE. The resulting estimator is the
familiar (and famed) augmented inverse probability weighted (AIPW) estimators
due to \citet{1994robinsEstimationRegressionCoefficients},
\citet{1995robinsAnalysisSemiparametricRegression}, and
\citet{1999scharfsteinAdjustingNonignorableDrop}. We further note that since the
propensity score is known in these experiments, the estimating equation is
linear in the unknown nonparametric component and has the local robustness
property of \citet{2018chernozhukovDoubleDebiasedMachine} and
\citet{2022chernozhukovLocallyRobustSemiparametric}. This reduces
the first order impact of bias in the first stage nonparametric estimates on the
resulting ATE estimator considerably and allows for efficient estimation under
much weaker conditions. The nonparametric unknowns here are conditional means of
the potential outcomes given baseline covariates. We show first that for a
generic estimator of conditional means, a weak \(L_{2}\) consistency
requirement combined with cross-fitting as in
\citet{2018chernozhukovDoubleDebiasedMachine} is sufficient for efficient
estimation of the ATE under CAR. The conditions under which the particular
choice of estimator achieves this \(L_{2}\) consistency are first left abstract
- different conditions apply to kernel estimators, random forests, series
estimators and neural networks for instance. For a given choice of nonparametric
estimator, smoothness conditions or functional form restrictions may be
unavoidable in achieving the \(L_{2}\) consistency property. Next, we establish
that if the particular nonparametric estimator is a cross-fitted Nadaraya-Watson
kernel regression estimator, then efficient estimation is possible under no
additional restrictions on the semiparametric model. Our motivation for using
the cross-fitted Nadaraya-Watson estimator is two-fold. First, the results of
\citet{1980devroyeDistributionFreeConsistencyResults} and
\citet{1980spiegelmanConsistentWindowEstimation} show that the Nadaraya-Watson
estimator is universally \(L_{2}\) consistent if the outcome in the regression
has a finite second moment. Second, the algebra of the cross-fitted
Nadaraya-Watson estimator is convenient in showing negligiblility of remainder
terms.

We provide simulation evidence of the finite sample performance of the feasible
efficient estimator. Our simulations compare this estimator to an infeasible
efficient estimator as well as an estimator that uses information only from the
stratum labels but not the additional baseline covariates. A comparison is also
provided with an imputation estimator of
\citet{1994chengNonparametricEstimationMean} and
\citet{1998hahnRolePropensityScore}. Across all simulations designs, we find
that using the feasible efficient estimator produces an efficiency gain of
13\% in terms of mean squared error reduction in comparison to the estimator
that discards information from baseline covariates beyond strata.
The maximal gain in this comparison over all simulations is a 40\% reduction in
mean squared error. Compared to the imputation estimator, our estimator exhibits
considerably less bias across most simulation designs.

The remainder of the paper is organized as follows. Section
\ref{sec--preliminaries} sets up the assumptions about the underlying population
from which the sample is drawn as well as the sampling process and the treatment
assignment scheme. Section \ref{sec--speb} presents the main result concerning
the semiparametric efficiency bound (\th\ref{thm--speb-car}) after a discussion
about parametric submodels in CAR experiment context. Section
\ref{sec--efficient-estimation} shows that the efficiency bound is achievable
under the same assumptions as used for its derivation via a two-step
semiparametric estimator using the efficient influence function. Section
\ref{sec--sims} provides Monte Carlo evidence of the performance of the
efficient estimator in finite samples. Section \ref{sec--conclusion} concludes
the paper.

%%% Local Variables:
%%% mode: latex
%%% TeX-master: "../2023_semipar_eff_car"
%%% End:
% LocalWords:  SSRA SPBR ATT

% ! TEX root = ../2023_semipar_eff_car.tex

\section{Preliminaries}
\label{sec--preliminaries}

This section describes assumptions on the populations of interest as well as the
sampling process which produces the observed data. This is done in two
subsections. Subsection \ref{subsec--CAR-pop} first describes the unobserved
study population of interest through a Neyman-Rubin causal model. Then, a target
\emph{observed population} is described for all CAR experiments considered in
this paper. This target observed population is useful for deriving the
efficiency bound in Section \ref{sec--speb}. Next, Subsection
\ref{subsec--CAR-sample} describes sampling assumptions and provides the two
main examples of sampling schemes in this paper.

Throughout this paper, all random variables and vectors will be defined on a
sufficiently rich underlying probability space \((\Omega, \mathscr{F},
\Pr)\). Independence of random elements is denoted by \(\indep\) and
expectations computed against the probability measure \(\Pr\) are denoted by
\(\mathbb{E} [\cdot]\). We will seldom refer to the underlying probability space
but define it nonetheless for clarity. Furthermore, absent any subscripts,
\(\mathbb{E} [\cdot]\) will also denote expected values assuming the ``true''
distribution of underlying data to be defined later on. The \(d\)-dimensional
multivariate normal distribution with mean vector \(\mathbf{m} \in
\mathbb{R}^{d}\) and covariance matrix \(\Sigma \in \mathbb{R}^{d \times d}\) is
denoted \(\mathcal{N} \left( \mathbf{m}, \Sigma \right)\). We denote the
matrix/vector transpose operation by \((\cdot)^{\prime}\). The natural numbers
are denoted by \(\mathbb{N} = \{1, 2, 3, 4, \dots\}\) and for a given
\(\mathcal{S} \in \mathbb{N}\), denote \(\mathbb{N}_{\mathcal{S}} = \{1, \dots,
\mathcal{S}\} = \{s \in \mathbb{N} : 1 \leq s \leq \mathcal{S}\}\).

\subsection{A population framework for CAR experiments.}
\label{subsec--CAR-pop}

We employ a standard binary treatment potential outcomes framework and assume
that samples are drawn from an infinite super-population. The population of
interest is described by a random vector \(W\) taking values in \(\mathbb{R}^{2
+ k}\) for some \(k \in \mathbb{N}\), with \(W^{\prime} = \left(Y (0), Y (1),
Z^{\prime} \right)\). For each \(a \in \{0, 1\}\), \(Y (a)\) is a scalar random
variable representing the potential outcome from receiving treatment \(a\). The
treatment \(a = 1\) can be interpreted as an ``innovation'' whereas \(a = 0\) is
a ``status quo'' or ``control''. \(Z\) is a random \(\mathbb{R}^{k}\)-vector of
baseline covariates. We denote the true distribution of \(W\) by
\(Q_{0}\). Prior to the assignment of treatment, only baseline covariates are
observable. Furthermore, after treatment has been assigned, only the outcome
associated with received treatment is observable so that \((Y (0), Y (1))\) is
not jointly observable. We maintain the following assumptions about the
population distribution \(Q_{0}\).

\begin{assumption}
\th\label{asm--Q}
Let \(\mu_{Z}\) be a \(\sigma\)-finite measure on the Borel sets of
\(\mathbb{R}^{k}\) and for \(a \in \{0, 1\}\), let \(\mu_{a}\) be a
\(\sigma\)-finite measure on the Borel sets of \(\mathbb{R}\). Furthermore, let
\(\mu\) denote the product measure on the Borel sets of \(\mathbb{R}^{2 + k}\)
constructed from \(\mu_{0}\), \(\mu_{1}\) and \(\mu_{Z}\). The true population
distribution, \(Q_{0}\), belongs to a family \(\mathbf{Q}\) of distributions on
\(\mathbb{R}^{2 + k}\) such that for each \(Q \in \mathbf{Q}\),
\begin{enumerate}[label = (\alph*)]
\item \label{asm--dominance}
  \(Q\) is dominated by \(\mu\) with Radon-Nikodym density \(q \left( \cdot; Q
  \right) = \mathrm{d} Q / \mathrm{d} \mu\).
\item \label{asm--nontriv-var}
  The potential outcomes have finite second moments under \(Q\),
  i.e. \(\mathbb{E}_{Q} \left[ Y {(a)}^{2} \right] < \infty\) for each \(a \in
  \{0, 1\}\), where \(\mathbb{E}_{Q}\) denotes the expected value assuming data
  are distributed according to \(Q\).
\end{enumerate}
\end{assumption}

The dominance assumption is made for mathematical convenience and is standard in
the literature on semiparametric efficiency. Additionally, the choices of
dominating measures \(\mu_{0}, \mu_{1}\) and \(\mu_{Z}\) are irrelevant and do
not affect the derivation of the efficiency bound. The finite second moment
condition is used for deriving Gaussian limiting distributions in later
sections. The family of distributions \(\mathbf{Q}\) is a nonparametric family
since it is infinite dimensional. The parameter of interest is the average
treatment effect (ATE), \(\beta_{\ast} : \mathbf{Q} \to \mathbb{R}\) defined by
\begin{equation}
  \beta_{\ast} (Q) = \mathbb{E}_{Q} [Y (1) - Y (0)] = \int \left( y_{1} - y_{0}
  \right) \; Q (\mathrm{d} y_{0}, \mathrm{d} y_{1}, \mathrm{d} z).
  \label{eqn--ate-Q}
\end{equation}
The true value of the ATE will be denoted \(\beta_{0} = \beta_{\ast} \left(
Q_{0} \right)\).

In both observational and experimental studies on the average effect of
treatment, observed outcomes result from the receipt of treatment. That is,
observed outcomes are given by \(Y\) defined by
\begin{equation*}
  Y = Y (1) \cdot A +  Y (0) \cdot (1 - A)
\end{equation*}
where \(A\) is a Bernoulli random variable denoting the treatment received. In
experimental settings, the conditional distribution of \(A\) given the baseline
covariates \(Z\) is assumed to be fully known to the experimenter. The target
observed population in the special case of a CAR experiment is the distribution
of the vector \(X^{\prime} = \left( Y, A, Z^{\prime} \right)\) which satisfies
\th\ref{asm--car-pop-ssra} below.

\begin{assumption}[Simple Stratified Randomization Assignment]
\th\label{asm--car-pop-ssra}
There exists a \(\mathcal{S} \in \mathbb{N}\), a measurable function
\(\mathbb{S} : \mathbb{R}^{k} \to \mathbb{N}_{\mathcal{S}}\) and a vector \(\pi
= (\pi (1), \dots, \pi (\mathcal{S}))\) with \(\pi (s) \in (0, 1)\) for every
\(s \in \mathbb{N}_{\mathcal{S}}\), which are all known to the experimenter. The
observable random vector \(X^{\prime} = (Y, A, Z^{\prime})\) is constructed from
\(W\) and its components satisfy
\begin{align}
  & Y = Y (1) \cdot A +  Y (0) \cdot (1 - A),
  \label{eqn--pop-obsoutproc} \\
  & [(Y (0), Y (1), Z) \indep A] | \mathbb{S} (Z),
  \label{eqn--car-pop-A-exog} \\
  & [A | \mathbb{S} (Z) = s] \sim \mathrm{Bernoulli} (\pi (s)).
  \label{eqn--car-pop-A-cond-Bern}
\end{align}
\end{assumption}

The function \(\mathbb{S}\) is used to construct a (measurable) finite partition
of the support of the covariates \(Z\) and thus, \(\mathcal{S}\) is an upper
bound on the number of strata. The condition in \eqref{eqn--car-pop-A-exog}
requires that assignment to treatment be exogenous to both potential outcomes
and the remaining variation in covariates given stratum labels. The condition in
\eqref{eqn--car-pop-A-cond-Bern} requires that treatment assignment
proportions correspond to the pre-specified target probabilities in \(\pi\). In
finite samples, the covariate adaptive randomization schemes that will be
described in the next subsection all provide different ways to target the
experiment in \th\ref{asm--car-pop-ssra}.

\begin{remark}
\th\label{rem--identification-ate-P}
Let \(\mathbf{P}\) denote the family of distributions for the random vector
\(X\) that is determined by \th\ref{asm--Q} and \th\ref{asm--car-pop-ssra}. Let
\(P_{0}\) denote the distribution in \(\mathbf{P}\) that corresponds to the true
population distribution \(Q_{0} \in \mathbf{Q}\). It is straightforward to show
that the ATE in \eqref{eqn--ate-Q} is nonparametrically identifiable (see
\citet[Definition 3.2]{2007matzkinNonparametricIdentification}) via the map
\(\beta : \mathbf{P} \to \mathbb{R}\)
\begin{equation}
  \beta (P) = \mathbb{E}_{P} \left[ \frac{Y \cdot A}{\pi (\mathbb{S} (Z))} -
  \frac{Y \cdot (1 - A)}{1 - \pi (\mathbb{S} (Z))} \right].
  \label{eqn--ate-P}
\end{equation}
That is, for a given \(Q \in \mathbf{Q}\), if \(P (Q)\) denotes the distribution
in \(\mathbf{P}\) formed from \(Q\) and \th\ref{asm--car-pop-ssra}, then
\(\beta_{\ast} (Q) = \beta (P (Q))\). In particular, the true value of the
average treatment effect can be written as
\begin{equation}
  \beta_{0} = \beta_{\ast} \left( Q_{0} \right) = \beta \left( P_{0} \right).
  \label{eqn--true-ATE}
\end{equation}
\end{remark}

\subsection{Sampling framework for CAR experiments.}
\label{subsec--CAR-sample}

In this subsection, we describe the assumptions maintained for the sampling
process that produces observed data in a CAR experiment.

\begin{assumption}
\th\label{asm--iid-Q}
\(\mathbf{W} = \left\{ W_{i} : i \in \mathbb{N} \right\}\) is a sequence of
independent and identically distributed (i.i.d.) random \(\mathbb{R}^{2 +
k}\)-vectors with \(W_{i} \sim Q_{0}\) for every \(i \in \mathbb{N}\). The first
\(n \in \mathbb{N}\) elements of \(\mathbf{W}\), denoted by
\(\mathbf{W}_{n}^{\prime} = \left( W_{1}, \dots, W_{n} \right)\), are the
potential outcome and covariate values for observations in the sample.
\end{assumption}

\th\ref{asm--iid-Q} states that potential outcome and covariate values for
observations in the sample are drawn at random from the population distribution
\(Q_{0}\). That is, if both potential outcomes and covariates were all
observable, the experimenter would have a representative sample from the
underlying population. This assumption is maintained in the recent literature on
CAR experiments as well as optimal experimental designs, for instance in
\citet{2018bugniInferenceCovariateAdaptiveRandomization},
\citet{2019bugniInferenceCovariateAdaptiveRandomization},
\citet{2022tabord-meehanStratificationTreesAdaptive},
\citet{2022baiOptimalityMatchedPairDesigns} and
\citet{2021cytrynbaumDesigningRepresentativeBalanced}. Given the baseline
covariates and corresponding strata, the experimenter chooses a vector of
treatment assignments \(\mathbf{A}_{n}^{\prime} = \left( A_{n 1}, \dots, A_{nn}
\right)\) which is a random vector supported in \({\{0, 1\}}^{n}\). The
experimenter has full control over the distribution of treatment assignments and
treatment assignments do not have to be i.i.d. We introduce the following
notation for convenience. For each \(s \in \mathbb{N}_{\mathcal{S}}\) and \(a
\in \{0, 1\}\) the stratum size and stratum treatment group size are
respectively
\begin{equation}
  N_{n} (s) = \sum_{i = 1}^{n} \mathbb{I} \left\{ \mathbb{S} \left( Z_{i}
  \right) = s \right\} \quad \text{and} \quad N_{n} (a, s) = \sum_{i = 1}^{n}
  \mathbb{I} \left\{ A_{n i} = a, \mathbb{S} \left( Z_{i} \right) = s \right\}.
  \label{eqn--strat-treat-size}
\end{equation}
We maintain the following assumptions about the observed data in a covariate
adaptive randomized experiment.

\begin{assumption}
\th\label{asm--treat-strat}
Let \(\mathcal{S}\), \(\mathbb{S}\) and \(\pi\) be as in
\th\ref{asm--car-pop-ssra}. For a sample of size \(n \in \mathbb{N}\), the
observed data are \(\mathbf{X}_{n}^{\prime} = \left( X_{n1}, \dots, X_{nn}
\right)\) with individual observations \(X_{n i} = \left( Y_{n i}, A_{n i},
Z_{i} \right)\) for each \(i \in \mathbb{N}_{n}\). The observed outcomes and
treatment assignment mechanism satisfy the following.
\begin{enumerate}[label = (\alph*)]
\item \label{asm--outcomes}
  The observed outcomes are given by
  \begin{equation}
    Y_{n i} = Y_{i} (1) A_{n i} + Y_{i} (0) \left( 1 - A_{n i} \right).
    \label{eqn--obsoutproc}
  \end{equation}
\item \label{asm--treat-exog}
  For every \(n \in \mathbb{N}\), \(\left[\mathbf{W}_{n} \indep \mathbf{A}_{n}
  \right] | \mathbf{S}_{n}\) where \(\mathbf{S}_{n}^{\prime} = \left( S_{1},
  \dots, S_{n} \right)\) and \(S_{i} =
  \mathbb{S} \left( Z_{i} \right)\) for each \(i \in \mathbb{N}_{n}\).
\item \label{asm--cov-adapt}
  For every \(n \in \mathbb{N}\), the conditional distribution of treatment
  assignments given the profile of sample strata
  \begin{equation}
    \alpha_{n} \left( \mathbf{a}_{n} \middle| \mathbf{s}_{n} \right) := \Pr
    \left( \mathbf{A}_{n} = \mathbf{a}_{n} \middle| \mathbf{S}_{n} =
    \mathbf{s}_{n} \right) \qquad \forall \mathbf{a}_{n} \in \{0, 1\}^{n},
    \forall \mathbf{s}_{n} \in \mathbb{N}_{\mathcal{S}}^{n}
    \label{eqn--cond-dist-A-S-sample}
  \end{equation}
  is completely known to the experimenter and does not depend on the population
  distribution \(Q_{0}\).
\item \label{asm--treat-prop}
  Under \(Q_{0}\) defined in \th\ref{asm--Q} and \(\alpha_{n} (\cdot | \cdot)\)
  chosen by the experimenter above, for each \(s \in \mathbb{N}_{\mathcal{S}}\),
  \begin{equation}
    \frac{N_{n} (1, s)}{N_{n} (s)} \overset{\mathrm{p}}{\to} \pi (s).
    \label{eqn--treat-strat-prop-conv-target}
  \end{equation}
\end{enumerate}
\end{assumption}

Denote the distribution of \(\mathbf{X}_{n}\) by \(P_{0, n}\). Note that \(P_{0,
n}\) is determined by the population distribution \(Q_{0}\), the stratification
scheme \(\mathbb{S}\), equation \eqref{eqn--obsoutproc}, and the choice of
randomization scheme. \th\ref{asm--treat-strat} places the same set of
restrictions on \(P_{0, n}\) as in
\citet{2019bugniInferenceCovariateAdaptiveRandomization} on the relationship
between the randomization scheme, the underlying population and the
strata. \th\ref{asm--treat-strat} \ref{asm--outcomes} requires that the observed
outcome for a given observation is exactly the potential outcome associated with
the assigned treatment (as in equation \eqref{eqn--pop-obsoutproc} for the
target experiment population
\(\mathbf{P}\)). \th\ref{asm--treat-strat} \ref{asm--treat-exog} requires that
the treatment assignment be ignorable (or exogenous) given the strata. In
particular, the assignment scheme can only be a function of the stratum labels
and a randomization device exogenous to the information contained in the
potential outcomes and the baseline covariates beyond that afforded by the
strata. This assumption is analogous to \eqref{eqn--car-pop-A-exog} in
\th\ref{asm--car-pop-ssra}. Its main use is to guarantee identification of the
ATE within each stratum and it is additionally a sufficient condition to
identify the overall ATE. \th\ref{asm--treat-strat} \ref{asm--cov-adapt}
requires that the randomization procedure be fully known to the
experimenter. This assumption plays a role in providing a simple
characterization of the joint distribution of the observed data and allows us to
avoid mathematical complications when talking about parametric sub-models during
the discussion of semiparametric efficiency. Finally,
\th\ref{asm--treat-strat} \ref{asm--treat-prop} requires the randomization
procedure to reach the target treatment proportion within a given stratum at
least asymptotically in the sense of convergence in probability.
This is analogous to
\eqref{eqn--car-pop-A-cond-Bern} in \th\ref{asm--car-pop-ssra}. There are a
number of examples of randomization schemes which will satisfy the requirements
imposed by \th\ref{asm--treat-strat}. We provide two examples of such
randomization schemes.

\begin{example}[Simple Stratified Random Assignment (SSRA)]
\th\label{eg--sra}
Let the assignments, \(\left\{ A_{i} \right\}_{i = 1}^{n}\), be i.i.d. Bernoulli
random variables such that \(\Pr \left( A_{i} = 1 | S_{i} = s \right) = \pi
(s)\) and \([(Y_{i} (0), Y_{i} (1), Z_{i}) \indep A_{i}] | S_{i}\). The
observations \(X_{i}^{\prime} = \left( Y_{i}, A_{i}, Z_{i}^{\prime} \right)\)
generated by this process form an i.i.d. sample from \(P_{0}\) defined in
\th\ref{asm--car-pop-ssra}. The conditional mass function \(\alpha_{n} (\cdot |
\cdot)\) here is
\begin{equation}
  \alpha_{n} \left( \mathbf{a}_{n} \middle| \mathbf{s}_{n} \right) = \prod_{i =
  1}^{n} \pi \left( s_{i} \right)^{a_{i}} \left( 1 - \pi \left( s_{i} \right)
  \right)^{1 - a_{i}} \qquad \forall \mathbf{a}_{n} \in \{0, 1\}^{n}, \forall
  \mathbf{s}_{n} \in \mathbb{N}_{\mathcal{S}}^{n}, \forall n \in \mathbb{N}.
  \label{eqn--cond-dist-A-S-SSRA}
\end{equation}
\th\ref{asm--treat-strat} \ref{asm--treat-prop} can be verified by appealing to
the Strong Law of Large Numbers.
\end{example}

\begin{example}[Stratified Permuted Block Randomization (SPBR)]
\th\label{eg--sbr}
Denote the integer floor function by \(\lfloor \cdot \rfloor\). Within stratum
\(s\), let
\begin{equation*}
  N_{n} (1, s) = \left\lfloor \pi (s) N_{n} (s) \right\rfloor \quad \text{and}
  \quad \frac{1}{c_{s, n}} = \binom{N_{n} (s)}{N_{n} (1, s)}.
\end{equation*}
There are \(c_{s, n}^{- 1}\) distinct subsets (or blocks) of size \(N_{n} (1,
s)\) from the overall stratum which has size \(N_{n} (s)\). We can choose one of
these blocks uniformly at random, i.e. each distinct block meeting the size
requirements gets assigned a probability mass of \(c_{n
s}\). \th\ref{asm--treat-strat} \ref{asm--treat-exog} can be satisfied by
ensuring the randomization device used to choose the treatment block is
independent to any outcome and covariate information within strata. Since
\(|N_{n} (1, s) - \pi (s) N_{n} (s)| \leq 1\) almost surely for every
\(s \in \mathbb{N}_{\mathcal{S}}\) and \(N_{n} (s) \overset{\mathrm{a.s.}}{\to}
\infty\) if \(Q (\mathbb{S} (Z) = s) > 0\), \th\ref{asm--treat-strat}
\ref{asm--treat-prop} is also satisfied. Additionally, for each \(\mathbf{a}_{n}
\in \{0, 1\}^{n}\), \(\mathbf{s}_{n} \in \mathbb{N}_{\mathcal{S}}^{n}\) and \(n
\in \mathbb{N}\),
\begin{equation}
  \alpha_{n} \left( \mathbf{a}_{n} \middle| \mathbf{s}_{n} \right) = \prod_{s =
  1}^{\mathcal{S}} \binom{\sum_{i = 1}^{n} \mathbb{I} \left\{ s_{i} = s
  \right\}}{\sum_{i = 1}^{n} a_{i} \mathbb{I} \left\{ s_{i} = s \right\}}^{- 1}
  \mathbb{I} \left\{ \sum_{i = 1}^{n} a_{i} \mathbb{I} \left\{ s_{i} = s
  \right\} = \left\lfloor \pi (s) \sum_{i = 1}^{n} \mathbb{I} \left\{ s_{i} = s
  \right\} \right\rfloor \right\}.
  \label{eqn--cond-dist-A-S-SPBR}
\end{equation}
\end{example}

Both examples satisfy the requirement imposed by \th\ref{asm--treat-strat}
\ref{asm--treat-prop} for approaching the stratum-specific target proportions,
but they do so at different rates of convergence. The SSRA method provides
convergence to the stratum-specific targets at a \(n^{- \frac{1}{2}}\) rate due
to the Lindeberg-L\'evy Central Limit Theorem. That is,
\begin{equation*}
  \left( \frac{N_{n} (1, s)}{N_{n} (s)} - \pi (s) : s = 1, \dots, \mathcal{S}
  \right) = O_{\mathrm{p}} \left( \frac{1}{\sqrt{n}} \right) \text{ under SSRA.}
\end{equation*}
It can be shown using the Strong Law of Large Numbers that SPBR achieves the
same targeting property at a faster \(n^{- 1}\) rate, so that
\begin{equation*}
  \left( \frac{N_{n} (1, s)}{N_{n} (s)} - \pi (s) : s = 1, \dots, \mathcal{S}
  \right) = O_{\mathrm{p}} \left( \frac{1}{n} \right) = o_{\mathrm{p}} \left(
  \frac{1}{\sqrt{n}} \right) \text{ under SPBR.}
\end{equation*}
When a randomization scheme achieves this faster convergence property we say
that it achieves ``strong balance''. Numerous other kinds of stratified
randomization techniques have been developed and analyzed in the literature on
randomized control trials --- for instance Efron's biased coin design and Wei's
urn --- see \citet{2018bugniInferenceCovariateAdaptiveRandomization} and
references therein for detailed descriptions. Furthermore, the treatment
assignments are not required to be i.i.d. and hence, the observed outcomes are
all potentially non-i.i.d. For instance, with the SPBR procedure illustrated in
\th\ref{eg--sbr}, assignments are independent across strata, but they exhibit
dependence within strata. Indeed, most (if not all) assignment schemes that
achieve strong balance will result in treatment assignments and observed
outcomes that are non-i.i.d. However, we will show that the same efficiency
bound as that for the i.i.d. sampling scheme in \th\ref{eg--sra} holds for all
CAR schemes that satisfy \th\ref{asm--treat-strat}. Finally while the
assumptions are satisfied by a fairly broad class of randomization schemes,
there are ones that violate them that are used in practice, e.g. the
minimization method of \citet{1975pocockSequentialTreatmentAssignment}.
Developing efficiency theory for these methods may be of interest but they are
not covered in this work.

%%% Local Variables:
%%% mode: latex
%%% TeX-master: "../2023_semipar_eff_car"
%%% End:
% LocalWords:  SSRA SPBR ATT

% ! TEX root = ../2023_semipar_eff_car.tex

\section{Semiparametric efficiency bound}
\label{sec--speb}

In this section, we first state the semiparametric efficiency bound for the ATE
and discuss it briefly. We also discuss the derivation of the bound in two
subsections. As stated in the introduction, of the efficiency bounds established
in the literature, the closest one to our setting is that of
\citet{1998hahnRolePropensityScore}. This is derived under i.i.d. observational
data under the assumption of ignorable treatment assignment conditional on
observable covariates (\citet[Section 1.3]{1983rosenbaumCentralRolePropensity}).
The \citet{1998hahnRolePropensityScore} bound does not apply immediately to our
setting since the treatment assignment rule is allowed to depend on the entire
profile of sample strata. For instance, it is not clear \emph{a priori} if the
choice of assignment mechanism will affect the efficiency bound since it can
clearly influence limit behavior of estimators as in
\citet{2018bugniInferenceCovariateAdaptiveRandomization,
2019bugniInferenceCovariateAdaptiveRandomization}. Furthermore, even when
treatment assignments are i.i.d. in a CAR context, the covariates being
conditioned on during assignment are the strata, so it is again not immediate
from the bound in \citet{1998hahnRolePropensityScore} what role the additional
baseline covariates can play in providing efficiency gains. The SPEB derived
here provides an answer to these questions. Before stating our bound, let \(m
\left( a, z; Q \right) = \mathbb{E}_{Q} [Y (a) | Z = z]\) for each \(a
\in \{0, 1\}\), and for brevity, let \(m_{\ast} (a, z) = m \left( a, z;
Q_{0} \right)\) Furthermore, define
\begin{equation}
  \mathbb{V}_{\ast} = \mathbb{E} \left[ \frac{\mathrm{Var} [Y (1) | Z]}{\pi
  (\mathbb{S} (Z))} + \frac{\mathrm{Var} [Y (0) | Z]}{1 - \pi (\mathbb{S} (Z))}
  + \left\{ m_{\ast} (1, Z) - m_{\ast} (0, Z) - \beta_{0} \right\}^{2} \right].
  \label{eqn--speb-car}
\end{equation}
The following theorem establishes that \(\mathbb{V}_{\ast}\) above is the SPEB.

\begin{theorem}
\th\label{thm--speb-car}
Under \th\ref{asm--Q,asm--iid-Q,asm--treat-strat}, the semiparametric efficiency
bound for estimating the ATE from a CAR experiment is \(\mathbb{V}_{\ast}\) in
\eqref{eqn--speb-car}.
\end{theorem}

The proof of \th\ref{thm--speb-car} provides formal justification for
\(\mathbb{V}_{\ast}\) as an efficiency bound by appealing to extensions of
the famed convolution (\citet{1970hajekCharacterizationLimitingDistributions})
and local asymptotic minimax theorems (\citet{1972hajekLocalAsymptoticMinimax})
to semiparametric problems. The efficiency bound in \eqref{eqn--speb-car} is
that of \citet{1998hahnRolePropensityScore} for i.i.d. data, if treatment is
assigned conditional on baseline covariates according to the propensity score
\(\pi (\mathbb{S} (\cdot))\). The intuition for this is that even though
randomization happens conditional on strata, the baseline covariates are
independent to treatment assignment after conditioning on strata
(\th\ref{asm--treat-strat} \ref{asm--treat-exog}). Thus, the minimum possible
limit variance in estimating the ATE can be achieved by utilizing any additional
information the baseline covariates may have about potential outcomes. This can
be better understood by considering what happens if only information from the
strata are used during estimation. For instance, consider the fully saturated
regression estimator of the ATE from
\citet{2019bugniInferenceCovariateAdaptiveRandomization}. This estimator is
constructed from regressing observed outcomes \(Y_{n i}\) on the stratum
indicators and their interactions with treatment status \(A_{n i}\). The
coefficients from this regression are combined to then form the ATE
estimate. The resulting estimator has a ``weighted difference in means''
form:
\begin{equation}
  \widehat{\beta}_{n, \mathrm{SAT}} = \sum_{s = 1}^{\mathcal{S}} \frac{N_{n}
  (s)}{n} \cdot \left[ \frac{\sum_{i = 1}^{n} Y_{n i} \cdot A_{n i} \cdot
  \mathbb{I} \left( S_{i} = s \right)}{N_{n} (1, s)} - \frac{\sum_{i = 1}^{n}
  Y_{n i} \cdot \left( 1 - A_{n i} \right) \cdot \mathbb{I} \left( S_{i} = s
  \right)}{N_{n} (0, s)} \right].
  \label{eqn--est-sat}
\end{equation}
\citet{2019bugniInferenceCovariateAdaptiveRandomization} show that
\(\widehat{\beta}_{n, \mathrm{SAT}}\) is asymptotically normal with mean-zero
and limit variance given by
\begin{equation}
  \mathbb{V}_{\mathrm{SAT}} = \mathbb{E} \left[ \frac{\mathrm{Var} [Y (1) |
  \mathbb{S} (Z)]}{\pi (\mathbb{S} (Z))} + \frac{\mathrm{Var} [Y (0) |
  \mathbb{S} (Z)]}{[1 - \pi (\mathbb{S} (Z))]} + \left\{ \mathbb{E} [Y (1) |
  \mathbb{S} (Z)] - \mathbb{E} [Y (0) | \mathbb{S} (Z)] - \beta_{0} \right\}^{2}
  \right].
  \label{eqn--speb-sat}
\end{equation}
Our results thus establish that the fully saturated regression estimator of the
ATE achieves the SPEB among all estimators that only use information from the
stratum labels. This follows from comparing \eqref{eqn--speb-car} and
\eqref{eqn--speb-sat} and noting that we have \(\mathbb{V}_{\mathrm{SAT}} =
\mathbb{V}_{\ast}\) if \(Z = \mathbb{S} (Z)\) (up to information preserving
relabelling of strata). The latter condition amounts to only using information
given in stratum labels. An additional implication is that when conditional
means or conditional variances of the potential outcomes given baseline
covariates exhibit variation within strata, \(\mathbb{V}_{\ast}\) can
be a strict improvement over \(\mathbb{V}_{\mathrm{SAT}}\) - an observation also
made by \citet{1998hahnRolePropensityScore}. This is confirmed by
Monte Carlo simulation evidence in Section \ref{sec--sims} in which
\(\mathbb{V}_{\ast}\) can be approximately half of \(\mathbb{V}_{\mathrm{SAT}}\)
in some simulation designs (i.e. there is up to a 50\% possible reduction in
asymptotic variance).

We provide a few additional comments on the SPEB in \eqref{eqn--speb-car}. Note
that the choice of assignment mechanism only affects the bound through the
choice of the target assignment proportions in \(\pi (\cdot)\). All else held
equal, the SPBR assignment mechanism in \th\ref{eg--sbr}, or any other procedure
satisfying \th\ref{asm--treat-strat}, will produce the same SPEB as i.i.d. draws
from \(\mathbf{P}\) (\th\ref{eg--sra}). Furthermore,
\citet{1998hahnRolePropensityScore} notes that knowledge of the propensity score
is ancillary to the SPEB for the ATE. In the context of CAR, we will show that
this knowledge is useful for achieving the SPEB during estimation. An important
intermediate product of our derivation for this purpose is the \emph{efficient
influence function}, \(\varphi_{0}\), which is defined by
\begin{equation}
  \varphi_{0} (y, a, z) = \frac{a}{\pi (\mathbb{S} (z))} \left[ y - m_{\ast} (1,
  z) \right] - \frac{1 - a}{1 - \pi (\mathbb{S} (z))} \left[ y - m_{\ast} (0, z)
  \right] + \left[ m_{\ast} (1, z) - m_{\ast} (0, z) - \beta_{0} \right].
  \label{eqn--eif}
\end{equation}
\(\varphi_{0}\) is called the efficient influence function because \(\mathbb{E}
\left[ \varphi_{0} (Y, A, Z) \right] = 0\) and its variance is the SPEB in
\eqref{eqn--speb-car}, i.e. \(\mathbb{E} \left[ \varphi_{0} (Y, A, Z)^{2}
\right] = \mathbb{V}_{\ast}\). This function is a key ingredient to achieving
the efficiency bound as illustrated in Section
\ref{sec--efficient-estimation}. In the subsequent two subsections, we provide
some additional discussion of how \eqref{eqn--speb-car} is derived, but
mathematical details are left to the appendix.

\subsection{A product characterization of distribution of observed data}

We establish here that the joint distribution of the observed data
\(\mathbf{X}_{n}\), \(P_{0, n}\), has a product structure even though outcomes
and treatment assignments are potentially non-i.i.d. This product structure will
then imply that \(P_{0, n}\) has a density that also has a product form. To that
end, we first define a dominating measure for a single observation. For Borel
sets \(\mathcal{Y} \subseteq \mathbb{R}\) and \(\mathcal{Z} \subseteq
\mathbb{R}^{k}\) and for \(\mathcal{A} \subseteq \{0, 1\}\), define the measures
\(\nu_{a}\) with \(a \in \{0, 1\}\) and \(\nu\) by
\begin{equation}
  \begin{split}
    \nu_{a} (\mathcal{Y} \times \mathcal{Z}) =
    & \ \mu_{a} (\mathcal{Y}) \cdot \mu_{Z} (\mathcal{Z}), \\
    \nu (\mathcal{Y} \times \mathcal{A} \times \mathcal{Z}) =
    & \ \nu_{1} (\mathcal{Y} \times \mathcal{Z}) \cdot \mathbb{I} (1 \in
      \mathcal{A}) + \nu_{0} (\mathcal{Y} \times \mathcal{Z}) \cdot \mathbb{I}
      (0 \in \mathcal{A}).
  \end{split}
  \label{eqn--nu}
\end{equation}

Let \(q_{a} (y , z; Q)\) denote the marginal density of \((Y (a), Z)\) if the
underlying population has distribution \(Q \in \mathbf{Q}\). Furthermore, let
\(P_{n} (\cdot; Q)\) denote the distribution of the observed data
\(\mathbf{X}_{n}\) if \(\mathbf{W}_{n}\) is an i.i.d. sample from \(Q \in
\mathbf{Q}\). Under this notation, we have \(P_{0, n} \equiv P_{n} \left( \cdot;
Q_{0} \right)\).

\begin{lemma}
\th\label{lem--prod-struct}
Let \(\mathbf{Y}_{n}^{\prime} = \left( Y_{n1}, \dots, Y_{nn} \right)\) and
\(\mathbf{Z}_{n}^{\prime} = \left( Z_{1}, \dots, Z_{n} \right)\). For each
\(\mathbf{y}_{n} \in \mathbb{R}^{n}\), \(\mathbf{a}_{n} \in {\{0, 1\}}^{n}\),
\(\mathbf{s}_{n} \in {\mathbb{N}_{\mathcal{S}}}^{n}\) and
\(\mathbf{z}_{n}^{\prime} = \left( z_{1}, \dots, z_{n} \right) \in \mathbb{R}^{k
\times n}\), denote the event
\begin{equation*}
  E_{n} \left( \mathbf{y}_{n}, \mathbf{a}_{n}, \mathbf{z}_{n}, \mathbf{s}_{n}
  \right) = \left\{ \mathbf{Y}_{n} \leq \mathbf{y}_{n}, \mathbf{A}_{n} =
  \mathbf{a}_{n}, \mathbf{Z}_{n} \leq \mathbf{z}_{n}, \mathbf{S}_{n} =
  \mathbf{s}_{n} \right\}.
\end{equation*}
where the inequalities are understood to be element-wise. Then under
\th\ref{asm--Q}, \th\ref{asm--car-pop-ssra}, \th\ref{asm--treat-strat} and
\eqref{eqn--obsoutproc}, if \(\mathbf{W}_{n}\) is an i.i.d. sample from \(Q \in
\mathbf{Q}\),
\begin{equation}
  P_{n} \left( E_{n} \left( \mathbf{y}_{n}, \mathbf{a}_{n}, \mathbf{z}_{n},
  \mathbf{s}_{n} \right); Q \right) = \ \alpha_{n} \left( \mathbf{a}_{n}
  \middle| \mathbf{s}_{n} \right) \times \prod_{i = 1}^{n} \left\{
  \begin{array}{l}
    Q \left( Y_{i} (1) \leq y_{i}, Z_{i} \leq z_{i}, S_{i} = s_{i}
    \right)^{a_{i}} \\
    \times Q \left( Y_{i} (0) \leq y_{i}, Z_{i} \leq z_{i}, S_{i} = s_{i}
    \right)^{1 - a_{i}}
  \end{array} \right\}.
  \label{eqn--data-joint-dist-prod}
\end{equation}
Thus, \(P_{n} (\cdot; Q)\) is absolutely continuous against the \(n\)-fold
product measure formed from \(\nu\) in \eqref{eqn--nu} with density
\begin{equation}
  p_{n} \left( \mathbf{y}_{n}, \mathbf{a}_{n}, \mathbf{z}_{n}; Q \right) =
  \alpha_{n} \left( \mathbf{a}_{n} \middle| \mathbf{s}_{n} \right) \prod_{i =
  1}^{n} q_{1} \left( y_{i}, z_{i}; Q \right)^{a_{i}} q_{0}
  \left(y_{i}, z_{i}; Q \right)^{1 - a_{i}}.
  \label{eqn--data-joint-density-prod}
\end{equation}
where in \eqref{eqn--data-joint-density-prod}, we restrict
\(\mathbf{s}_{n}^{\prime} = \left( \mathbb{S} \left( z_{1} \right), \dots,
\mathbb{S} \left( z_{n} \right) \right)\).
\end{lemma}

The main consequence of \th\ref{lem--prod-struct} is as follows. Let
\(\mathcal{P} = \{\mathbf{P}_{n} : n \in \mathbb{N}\}\) be the sequence of
families of distributions for the observed data \(\mathbf{X}_{n}\) such that for
each \(n \in \mathbb{N}\), every member of \(\mathbf{P}_{n}\) is determined by
an i.i.d. sample from some distribution in \(\mathbf{Q}\) (as defined in
\th\ref{asm--Q}), the observed outcome equation \eqref{eqn--obsoutproc} and a
treatment assignment mechanism that satisfies \th\ref{asm--treat-strat}. Any
member of \(\mathbf{P}_{n}\) must then satisfy the product structure in
\eqref{eqn--data-joint-dist-prod} and \eqref{eqn--data-joint-density-prod}.
Furthermore, different covariate adaptive randomization schemes give rise to
different sequences of families \(\mathcal{P}\) that only vary according to the
choice of the sequence of conditional treatment assignment distributions
\(\left\{ \alpha_{n} (\cdot) : n \in \mathbb{N} \right\}\).

\begin{remark}[Implications for Simple Stratified Random Assignment]
\th\label{rem--prod-struct-ssra}
Consider the family \(\mathbf{P}\) defined by \th\ref{asm--car-pop-ssra} in
\th\ref{rem--identification-ate-P}. An immediate additional consequence of
\th\ref{lem--prod-struct} is that any distribution \(P \in \mathbf{P}\) has a
density against \(\nu\) of the form
\begin{equation*}
  p (y, a, z; P) = \left[ q_{1} (y, z) \cdot \pi (\mathbb{S} (z)) \right]^{a}
  \left[ q_{0} (y, z) \cdot \pi (1 - \mathbb{S} (z)) \right]^{1 - a}.
\end{equation*}
Recall also that in \th\ref{eg--sra}, we consider a treatment assignment rule
that essentially produces an i.i.d. sample of size \(n\) from a distribution in
\(\mathbf{P}\). In this case, we would of course have that \(\mathbf{P}_{n} =
\mathbf{P}^{n}\) where the latter is the set of all \(n\)-fold product measures
formed from some measure in \(\mathbf{P}\).
\end{remark}

\subsection{Deriving the semiparametric efficiency bound}

In this subsection, we first introduce and discuss definitions of parametric
submodels and regular parametric submodels which are essential concepts to the
derivation of the semiparametric efficiency bound. It is then shown in
\th\ref{lem--parametric-submodel-lan} that the log likelihood ratios in regular
parametric submodels under \(1 / \sqrt{n}\) local alternatives exhibit the local
asymptotic normality (LAN) property of
\citet{1960lecamLocallyAsymptoticallyNormal}. An efficiency bound is
established for parametric submodels
(\th\ref{lem--parametric-submodel-efficiency}). Then, via a discussion on
differentiability of the ATE in regular parametric submodels we provide an
intuitive description of how the parametric efficiency bound is extended to the
semiparametric efficiency bound presented in \th\ref{thm--speb-car}.

\begin{definition}
\th\label{def--parametric-submodel}
A parametric submodel of \(\mathcal{P}\) is a sequence of families of
distributions, \(\mathcal{P}^{0} = \left\{ \mathbf{P}_{n}^{0} : n \in \mathbb{N}
\right\}\) that satisfies the following conditions.
\begin{enumerate}[label=(\alph*)]
\item \label{def--parametric-submodel-subset}
  For every \(n \in \mathbb{N}\), \(\mathbf{P}^{0}_{n} \subseteq
  \mathbf{P}_{n}\). That is, every member, \(P_{n} \in \mathbf{P}^{0}_{n}\) is
  determined by a population distribution from a subfamily \(\mathbf{Q}^{0}
  \subseteq \mathbf{Q}\) and satisfies \th\ref{asm--treat-strat}.
\item \label{def--parameteric-submodel-densities}
  There is a \(\Theta \subseteq \mathbb{R}^{d}\) (where \(d \in \mathbb{N}\)),
  and a map \(\theta \mapsto \left( q_{0} (\cdot; \theta), q_{1} (\cdot; \theta)
  \right)\) from \(\Theta\) to \((Y (a), Z)\)-marginal densities of
  \(\mathbf{Q}^{0}\) such that every member of \(\mathbf{P}^{0}_{n}\) has a
  density (as in \th\ref{lem--prod-struct}) of the form
  \begin{align*}
    p_{n} \left( \mathbf{y}_{n}, \mathbf{a}_{n}, \mathbf{z}_{n}; \theta \right)
    =
    & \ \alpha_{n} \left( \mathbf{a}_{n} | \mathbf{s}_{n} \right) \cdot \prod_{i
      = 1}^{n} q_{1} \left( y_{i}, z_{i}; \theta \right)^{a_{i}} q_{0} \left(
      y_{i}, z_{i}; \theta \right)^{1 - a_{i}}, \\
    \text{where } \mathbf{s}_{n}^{\prime} =
    & \ \left( \mathbb{S} \left( z_{1} \right), \dots, \mathbb{S} \left( z_{n}
      \right) \right)
  \end{align*}
  Additionally, the map \(\theta \mapsto \left( q_{0} (\cdot; \theta),
  q_{1} (\cdot; \theta) \right)\) satisfies the \(Z\)-marginal
  restriction: there exists a \(\mu_{Z}\)-density function \(g (\cdot; \theta)\)
  (i.e. \(g (z; \theta) \geq 0\) for all \(z \in \mathbb{R}^{k}\) and \(\int g
  (z; \theta) \mu_{Z} (\mathrm{d} z) = 1\)) such that
  \begin{equation*}
    \int_{\mathbb{R}} q_{0} (y, z; \theta) \; \mu_{0} \left( \mathrm{d} y
    \right) = \int_{\mathbb{R}} q_{1} (y, z; \theta) \; \mu_{1} \left(
    \mathrm{d} y \right) = g (z; \theta) \quad \forall \theta \in \Theta
    \text{ for } z \in \mathbb{R}^{k} \ \mu_{Z} \text{-a.e.}
  \end{equation*}
  We write \(P_{\theta, n}\) for the element of \(\mathbf{P}_{n}^{0}\) whose
  density is \(p_{n} (\cdot; \theta)\).
\item \label{def--parameteric-submodel-contains-true}
  The population subfamily \(\mathbf{Q}^{0}\) contains the true distribution
  \(Q_{0}\) so that \(P_{0, n} \in \mathbf{P}_{n}^{0}\). The parametrization
  \(\theta \mapsto p_{n} (\cdot; \theta)\) identifies the true distribution so
  that there is a unique \(\theta_{0} \in \Theta\) such that \(P_{0, n} =
  P_{\theta_{0}, n}\).
\end{enumerate}
\end{definition}

In \th\ref{def--parametric-submodel}, condition
\ref{def--parametric-submodel-subset} requires that a parametric submodel be a
subset of the overall semiparametric model. Condition
\ref{def--parameteric-submodel-densities} requires that the parametrization
\(\theta \mapsto p_{n} (\cdot; \theta)\) produces densities of the same form as
those in \th\ref{lem--prod-struct}. The additional \(Z\)-marginal restriction
ensures that both densities \(q_{0} (\cdot; \theta)\) and \(q_{1}
(\cdot; \theta)\) can be derived from a common distribution \(Q_{\theta} \in
\mathbf{Q}^{0}\). We do not require the submodel to specify the joint density of
\((Y (0), Y (1), Z)\) since only the marginals with respect to \((Y (0), Z)\)
and \((Y (1), Z)\) are identified by the randomized experiment (essentially due
to \eqref{eqn--data-joint-dist-prod} in \th\ref{lem--prod-struct}). Finally,
\ref{def--parameteric-submodel-contains-true} requires the parametric submodel
to contain the true distribution \(P_{0, n}\). Note that a parametric submodel
does not parameterize the conditional mass function of the treatment
assignments, \(\alpha_{n}\), since it is known and does not depend on any
population unknowns. This is different to the case of observational data
considered in \citet{1998hahnRolePropensityScore} where the propensity score
is unknown and has to be parameterized.

The idea behind the use of parametric submodels is as follows. Since \(\Theta\)
(and therefore \(\mathbf{P}_{n}\)) in \th\ref{def--parametric-submodel} is
finite dimensional, estimation of the ATE restricted to \(\mathcal{P}^{0}\)
cannot be more difficult than estimation within \(\mathcal{P}\). One can imagine
estimation of the ATE by first estimating the nuisance parameter \(\theta_{0}\)
by maximum likelihood and then forming a plug-in estimate of the ATE. If
\(\mathcal{P}^{0}\) is a parametric submodel with a well-defined efficiency
bound, any semiparametric estimator for the ATE that is consistent and
asymptotically normal under the distributions in \(\mathcal{P}\) cannot have
asymptotic variance lower than the efficiency bound under
\(\mathcal{P}^{0}\). Taking the supremum over all parametric submodels, we get a
variance lower bound for the semiparametric model. The statement about
parametric submodels having well-defined efficiency bounds requires some
qualification. Establishing efficiency requires restricting our analysis to
\emph{regular} parametric submodels (in which efficiency bounds are
well-defined) and to \emph{regular} estimators (to rule out the phenomenon of
super-efficiency). We provide the definition of a regular parametric submodel in
our case below. The definition of a regular estimator can be found in
\citet[p. 115 and p. 365]{1998vandervaartAsymptoticStatistics} or
\citet[p. 413]{1996vandervaartWeakConvergenceEmpirical}.

\begin{definition}
\th\label{def--qmd}
Let \(\mathcal{P}^{0}\) be a parametric subfamily of \(\mathcal{P}\) as in
\th\ref{def--parametric-submodel}. \(\mathcal{P}^{0}\) is called \emph{regular}
at \(\theta_{0}\) if the parametrization \(\theta \mapsto \left( q_{0} (\cdot;
\theta), q_{1} (\cdot; \theta) \right)\) satisfies the following.
\begin{enumerate}[label=(\alph*)]
\item \(\Theta\) is an bounded and open subset of \(\mathbb{R}^{d}\).
\item \label{def--qmd-differentiability}
  The maps \(\theta \to \sqrt{q_{a} (\cdot; \theta)}\) are differentiable in
  the quadratic mean at \(\theta_{0}\). That is, for each \(a \in \{0, 1\}\),
  there is a measurable function \(D_{a} : \mathbb{R}^{1 + k} \to
  \mathbb{R}^{d}\) such that \(\int D_{a}^{\prime} D_{a} \mathrm{d} \nu_{a} <
  \infty\) and
  \begin{equation*}
    \lim_{\|t\| \to 0} \|t\|^{- 2} \int \left( \sqrt{q_{a} (y, z; \theta_{0} +
    t)} - \sqrt{q_{a} \left( y, z; \theta_{0} \right)} - D_{a} (y, z)^{\prime}
    \left( \theta - \theta_{0} \right) \right)^{2} \nu_{a} (\mathrm{d} y,
    \mathrm{d} z) = 0.
  \end{equation*}
\item The information matrix at \(\theta_{0}\), \(\mathcal{I}_{a}\), is
  non-singular, where
  \begin{equation}
    \mathcal{I}_{a} = 4 \int_{\mathbb{R}^{1 + k}} D_{a} (y, z) D_{a} (y,
    z)^{\prime} \; \nu_{a} (\mathrm{d} y, \mathrm{d} z).
    \label{eqn--parametric-info-mat}
  \end{equation}
\end{enumerate}
\end{definition}

The log-likelihood in any parametric submodel is
\begin{equation}
  \begin{split}
    \ell_{n} (\theta) =
    & \ \log p_{n} \left( \mathbf{X}_{n}; \theta \right) = \log \alpha_{n}
      \left( \mathbf{A}_{n} \middle| \mathbf{S}_{n} \right) + \sum_{i = 1}^{n}
      \ell \left( Y_{n i}, A_{n i}, Z_{i}; \theta \right), \\
    \text{where } \ell \left( Y_{n i}, A_{n i}, Z_{i}; \theta \right) =
    & \ A_{n i} \log q_{1} \left( Y_{n i} , Z_{i}; \theta \right) +
      \left(1 - A_{n i} \right) \log q_{0} \left( Y_{n i}, Z_{i}; \theta
      \right).
  \end{split}
  \label{eqn--sample-loglhood}
\end{equation}
When the parametric submodel is regular, the associated score function at the
truth, \(\theta_{0}\), is
\begin{equation}
  \begin{split}
    \dot{\ell}_{n} =
    & \ \sum_{i = 1}^{n} \dot{\ell} \left( Y_{n i}, A_{n i}, Z_{i} \right), \\
    \text{where } \dot{\ell} \left( Y_{n i}, A_{n i}, Z_{i} \right) =
    & \ A_{n i} \dot{\ell}_{1} \left( Y_{n i}, Z_{i} \right) + \left( 1 - A_{n
      i} \right) \dot{\ell}_{0} \left( Y_{n i}, Z_{i} \right), \\
    \dot{\ell}_{a} (y, z) =
    & \ 2 \frac{D_{a} (y, z)}{\sqrt{q_{a} \left( y, z; Q_{0}
      \right)}} \cdot \mathbb{I} \left\{ q_{a} \left( y, z; Q_{0}
      \right) > 0 \right\}.
  \end{split}
  \label{eqn--sample-score}
\end{equation}
Note that the score in \eqref{eqn--sample-score} is expressed in ``root
density'' terms since it is technically defined by the quadratic mean
differentiability assumption in \th\ref{def--qmd}.\footnote{When the parametric
submodels are such that the densities are continuously differentiable with
respect to \(\theta\), then the usual log-derivative form and
\eqref{eqn--sample-score} coincide. That is, we can write \(\dot{\ell}_{a}
(\cdot) = \nabla_{\theta} \left[ \log q_{a} \right] \left( \cdot; \theta_{0}
\right)\).} Let \(X^{\prime} = (Y, A, Z)\) be as in the CAR target  experiment
defined by \th\ref{asm--car-pop-ssra}. Under CAR, the associated Fisher
information matrix will be shown to be
\begin{equation}
  \begin{split}
    \mathcal{I} =
    & \ \mathbb{E} \left[ \dot{\ell} (Y, A, Z) \dot{\ell} (Y, A, Z)^{\prime}
      \right] \\
    = & \ \sum_{s = 1}^{\mathcal{S}} Q_{0} (\mathbb{S} (Z) = s) \left\{
        \begin{array}{l}
          \pi (s) \mathbb{E} \left[ \dot{\ell}_{1} (Y (1), Z)
          \dot{\ell}_{1} (Y (1), Z)^{\prime} \middle| \mathbb{S}
          (Z) = s \right] \\
          + (1 - \pi (s)) \mathbb{E} \left[ \dot{\ell}_{0} (Y (0), Z)
          \dot{\ell}_{0} (Y (0), Z)^{\prime} \middle| \mathbb{S} (Z) = s \right]
        \end{array}
        \right\}.
  \end{split}
  \label{eqn--strata-info-mat}
\end{equation}
Note that the information matrix, \(\mathcal{I}\) in
\eqref{eqn--strata-info-mat} is related to the marginal information matrices
\(\mathcal{I}_{a}\) for \(a \in \{0, 1\}\) in \eqref{eqn--parametric-info-mat}
but adjusts these for stratification. Furthermore the score \(\dot{\ell}\) and
information matrix \(\mathcal{I}\) both depend on the choice of parametric
submodel since the quadratic mean derivatives \(D_{a}\) depend on the choice of
parametric submodel. \th\ref{lem--parametric-submodel-lan} below shows that the
main consequence of \th\ref{def--qmd} is that the likelihood ratio in any
regular parametric submodel has the usual quadratic expansion and exhibits the
local asymptotic normality (LAN) property of
\citet{1960lecamLocallyAsymptoticallyNormal} under general CAR procedures.

\begin{lemma}
\th\label{lem--parametric-submodel-lan}
Suppose the parametric submodel \(\mathcal{P}^{0}\) of \(\mathcal{P}\) is
regular at the true value \(\theta_{0}\). Then, \(\mathcal{P}^{0}\) has the
local asymptotic normality property of
\citet{1960lecamLocallyAsymptoticallyNormal} at \(\theta_{0} \in \Theta\). In
particular, for \(t \in \mathbb{R}^{d}\) such that \(\theta_{0} + n^{-
\frac{1}{2}} t \in \Theta\) for all \(n \in \mathbb{N}\), let the quadratic
remainder \(R_{n} \left( \theta_{0}, t \right)\) be defined by
 \begin{equation}
  \ell_{n} \left( \theta_{0} + n^{- \frac{1}{2}} t \right) - \ell_{n}
  \left( \theta_{0} \right) = t^{\prime} \frac{1}{\sqrt{n}} \dot{\ell}_{n} -
  \frac{1}{2} t^{\prime} \mathcal{I} t + R_{n} \left( \theta_{0}, t \right)
  \label{eqn--lr-quad-breakdown}
\end{equation}
where \(\dot{\ell}_{n}\) is as in \eqref{eqn--sample-score} and
\(\mathcal{I}\) is defined as in \eqref{eqn--strata-info-mat}. Let \(G
\sim \mathcal{N} \left( \mathbf{0}_{d}, \mathcal{I} \right)\). Then the
following hold.
\begin{enumerate}[label=(\alph*)]
\item\label{lem--parametric-submodel-lan-wc}
  Under \(P_{0, n}\), \(n^{- \frac{1}{2}} \dot{\ell}_{n}
  \overset{\mathrm{d}}{\to} G\).
\item\label{lem--parametric-submodel-lan-uan}
  For any given \(M > 0\),
  \(R_{n} \left( \theta_{0}, t \right) \overset{\mathrm{p}}{\to} 0\) uniformly
  over \(\|t\| \leq M\) under \(P_{0, n}\). That is, for any \(\varepsilon >
  0\),
  \begin{equation}
    \lim_{n \to \infty} \sup_{\|t\| \leq M} P_{0, n} \left( \left| R_{n} \left(
    \theta_{0}, t \right) \right| > \varepsilon \right) = 0.
  \end{equation}
\end{enumerate}
\end{lemma}

The proof of \th\ref{lem--parametric-submodel-lan} in the appendix of this paper
uses the tools developed
\citet{2018bugniInferenceCovariateAdaptiveRandomization} and
\citet{2019bugniInferenceCovariateAdaptiveRandomization} to adapt the proof
of Proposition 2.1.2 in the appendix of
\citet{1998bickelEfficientAdaptiveEstimation} to the context of CAR. In
particular, a partial-sums empirical process argument is used to establish
conclusion \ref{lem--parametric-submodel-lan-wc} in
\th\ref{lem--parametric-submodel-lan}. Using similar arguments, the sample
information matrix is shown to converge in probability to \(\mathcal{I}\). The
remainder \(R_{n} \left( \theta_{0}, t \right)\) summarizes both the error
replacing the sample information matrix with its population (or limiting)
analogue as well as the error inherent in a second-order Taylor expansion. The
latter is shown to be negligible as well. Under the LAN phenomenon,
\th\ref{lem--parametric-submodel-efficiency} below establishes an efficiency
result for estimation of the nuisance parameter \(\theta_{0}\).

\begin{lemma}
\th\label{lem--parametric-submodel-efficiency}
Suppose \(\mathcal{P}^{0}\) is a parametric submodel of
\(\mathcal{P}\) that is regular at \(\theta_{0}\). Let \(\mathcal{I}\) be the
Fisher information matrix defined in \eqref{eqn--strata-info-mat}. Then
\(\mathcal{I}\) is non-singular and \(\mathcal{I}^{- 1}\) is the efficiency
bound for estimation of \(\theta_{0}\). In particular, let \(G \sim \mathcal{N}
\left( \mathbf{0}_{d}, \mathcal{I}^{- 1} \right)\). Then the following hold.
\begin{enumerate}[label=(\alph*)]
\item \label{lem--parametric-submodel-efficiency-convolution}
  If \(\widehat{\theta}_{n}\) is a sequence of regular estimators of
  \(\theta_{0}\), then
  \begin{equation}
    \sqrt{n} \left( \widehat{\theta}_{n} - \theta_{0} - \frac{t}{\sqrt{n}}
    \right) \overset{\mathrm{d}}{\to} G + U.
    \label{eqn--convolution-regular}
  \end{equation}
  where \(U\) is a random \(\mathbb{R}^{d}\)-vector independent to \(G\) that is
  specific to the estimator sequence \(\widehat{\theta}_{n}\).
\item \label{lem--parametric-submodel-efficiency-lam}
  Let \(\mathcal{L} : \mathbb{R}^{d} \to \mathbb{R}\) be a bowl-shaped loss
  function, i.e. \(\mathcal{L}\) is non-negative, satisfies \(\mathcal{L} (y) =
  \mathcal{L} (- y)\) and has convex sublevel sets. If \(\widehat{\theta}_{n}\)
  is any sequence of estimators of \(\theta_{0}\), then
  \begin{equation}
    \lim_{M \to \infty} \liminf_{n \to \infty} \ \sup_{\theta_{n} \in \Theta,
    \sqrt{n} \left\| \theta_{n} - \theta_{0} \right\| \leq M}
    \mathbb{E}_{P_{\theta_{n}, n}} \left[ \mathcal{L} \left( \sqrt{n} \left(
    \widehat{\theta}_{n} - \theta_{n} \right) \right) \right] \geq \mathbb{E}
    \left[ \mathcal{L} \left( G \right) \right].
    \label{eqn--theta-lam}
  \end{equation}
\end{enumerate}
\end{lemma}

\th\ref{lem--parametric-submodel-efficiency} establishes that the usual
Cram\'er-Rao parametric information bound holds for estimation of \(\theta_{0}\)
in regular parametric submodels. This is true even under general CAR procedures
that may result in dependence within observed data.
\th\ref{lem--parametric-submodel-efficiency} is a direct consequence of the LAN
property established in \th\ref{lem--parametric-submodel-lan}. In
\th\ref{lem--parametric-submodel-efficiency},
conclusion \ref{lem--parametric-submodel-efficiency-convolution}
follows from the celebrated convolution theorem of
\citet{1970hajekCharacterizationLimitingDistributions} and
conclusion \ref{lem--parametric-submodel-efficiency-lam} follows from the local
asymptotic minimax theorem of \citet{1972hajekLocalAsymptoticMinimax}.
These are standard tools used to justify the parametric Cram\'er-Rao bound. Note
that in \th\ref{lem--parametric-submodel-efficiency} the random vectors \(G\)
and \(U\) are again specific to the parametric submodel chosen and in the case
of \(U\), also specific to the choice of estimator sequence.

\th\ref{lem--parametric-submodel-efficiency} establishes an efficiency bound for
estimation of the nuisance parameter \(\theta\) whereas the content of
\th\ref{thm--speb-car} is about estimation of the ATE, \(\beta\) in
\eqref{eqn--ate-P}. To relate \th\ref{lem--parametric-submodel-efficiency} to a
version of \th\ref{thm--speb-car}, we have to first show that \(\beta\) is a
pathwise differentiable parameter (see \citet[Definition 3.3.1,
p. 57]{1998bickelEfficientAdaptiveEstimation} or \citet[Section
3]{1990neweySemiparametricEfficiencyBounds}). We do this by considering regular
parametric submodels of the target CAR experiment, \(\mathbf{P}\),
as defined by \th\ref{asm--car-pop-ssra} and
\th\ref{rem--identification-ate-P}. For a regular parametric submodel
\(\mathbf{P}_{0} = \left\{ P_{\theta} : \theta \in \Theta \right\}\) of
\(\mathbf{P}\), the average treatment effect parameter is
\begin{equation*}
  \gamma (\theta) := \beta \left( P_{\theta} \right) = \int_{\mathbb{R}^{1 + k}}
  y \cdot q_{1} (y, z; \theta) \; \nu_{1} (\mathrm{d} y, \mathrm{d} z) -
  \int_{\mathbb{R}^{1 + k}} y \cdot q_{0} (y, z; \theta) \; \nu_{0} (\mathrm{d}
  y, \mathrm{d} z).
\end{equation*}
Adapting the pathwise differentiability arguments in
\citet{1998hahnRolePropensityScore} to the context of a CAR experiment, the
derivative (gradient) of \(\gamma (\theta)\) at \(\theta_{0}\) in any given
regular parametric submodel can be written as
\begin{equation}
  \nabla \gamma \left( \theta_{0} \right) = \mathbb{E} \left[ \varphi_{0} (Y, A,
  Z) \cdot \dot{\ell} (Y, A, Z) \right],
  \label{eqn--path-derivative-beta}
\end{equation}
where \(\varphi_{0}\) is the efficient influence function from \eqref{eqn--eif}.
The full derivation of this adaptation is included the appendix for completeness
(\th\ref{lem--tangent-space-P,lem--pathwise-differentiability-ate}). Combining
\th\ref{lem--parametric-submodel-efficiency} and
\eqref{eqn--path-derivative-beta}, it follows via the usual ``delta method
style'' argument that the information bound for estimating the ATE in the
regular parametric submodel \(\mathcal{P}^{0}\) is
\begin{equation*}
  \mathbb{V} \left( \mathcal{P}^{0} \right) = \mathbb{E} \left[ \varphi_{0} (Y,
  A, Z) \dot{\ell} \left( Y, A, Z \right)^{\prime} \right] \mathbb{E} \left[
  \dot{\ell} \left( Y, A, Z \right) \dot{\ell} \left( Y, A, Z \right)^{\prime}
  \right]^{- 1} \mathbb{E} \left[\dot{\ell} \left( Y, A, Z \right) \varphi_{0}
  (Y, A, Z) \right].
\end{equation*}
The variant of Cauchy-Schwarz inequality for random vectors established in
\citet{1999tripathiMatrixExtensionCauchySchwarz} can be used to conclude that
\(\mathbb{V} \left( \mathcal{P}^{0} \right) \leq \mathbb{E} \left[\varphi_{0}
(Y, A, Z)^{2} \right] = \mathbb{V}_{\ast}\). This establishes that
\(\mathbb{V}_{\ast}\) in \eqref{eqn--speb-car} is an upper bound over all
parametric information bounds for estimation of the ATE. However, one still
needs to show that \(\mathbb{V}_{\ast}\) is in fact a supremum. This final step
is done by establishing that \(\varphi_{0}\) can be approximated by the score
functions of regular parametric submodels.

%%% Local Variables:
%%% mode: latex
%%% TeX-master: "../2023_semipar_eff_car"
%%% End:
% LocalWords:  SSRA SPBR ATT

% ! TEX root = ../2023_semipar_eff_car.tex

\section{Semiparametrically efficient estimation of the ATE under CAR}
\label{sec--efficient-estimation}

The efficiency bound derived in \th\ref{thm--speb-car} is a lower bound on
asymptotic variances of regular semiparametric estimators for the ATE. However,
the question remains as to whether it can be achieved. Typically, the conditions
under which estimators can achieve semiparametric efficiency bounds are stronger
than those required to derive the bounds.
\citet{1990ritovAchievingInformationBounds} provide examples for cases
when an efficiency bound can be established and shown to be finite as well as
non-singular, but even consistent estimation (let alone efficient) is impossible
without restricting the semiparametric model. One of their examples is the
partially linear model (\citet{1986engleSemiparametricEstimatesRelation},
\citet{1988robinsonRootNConsistent}). The restrictions they require are in the
form of additional smoothness conditions for the nonparametric nuisance
parameters that appear in the efficient influence function. These smoothness
conditions are to ensure that the nonparametric nuisance parameters can be
estimated consistently with a rate of convergence of at least \(n^{- 1 / 4}\)
(see also \citet{1994neweyAsymptoticVarianceSemiparametric}). Furthermore,
achieving \(n^{- 1 / 4}\)-consistency for a non-parametric estimator becomes
more difficult when there are many covariates due to the curse of
dimensionality, and thus higher order (i.e. stronger) smoothness conditions
are required.

In the context of estimating average treatment effects with the assumption of
selection on observables (or ignorability conditional on covariates) and
i.i.d. data, \citet{1998hahnRolePropensityScore} proposes nonparametric
imputation estimators. These estimators impute the potential outcome \(Y_{i}
(a)\) whenever it is unobserved by using an estimate of the predicted value \(m
\left( a, Z_{i} \right)\), say \(\widehat{m}_{n} \left( a, Z_{i} \right)\). The
difference of the imputed potential outcomes is then averaged to estimate the
ATE. \citet{1998hahnRolePropensityScore} suggests the use of series estimators
for the propensity score \(\Pr (A = 1 | Z = z)\) and the conditional
expectations \(\mathbb{E} [Y \mathbb{I} \{A = a\} | Z = z]\) which then get
combined to construct estimators for \(m_{\ast} (a, \cdot)\). Similar estimators
appear in various parts of the literature on semiparametric estimation of
treatment effects under treatment ignorability - see
\citet{2004imbensNonparametricEstimationAverage} for a review and examples with
kernel estimators for the aforementioned propensity score and conditional
expectations. The requirement of higher order smoothness conditions on the
propensity score and the conditional expectations \(m_{\ast} (a, \cdot)\) are
ubiquitous in these examples.

In our case, the propensity score does not need to be estimated since the target
proportions by strata, \((\pi (1), \dots, \pi (\mathcal{S}))\), are chosen by
the experimenter. Furthermore, we should be able to use sample treatment
proportions by strata in place of true target proportions without any issue
since these form a finite-dimensional additional nuisance parameter. This
feature makes the efficient influence function linear in the nonparametric
nuisance parameters \(m_{\ast} (a, \cdot)\). It is also straightforward
to verify that \(\varphi\) has the Neyman orthogonality or local robustness
property of \citet{2018chernozhukovDoubleDebiasedMachine} and
\citet{2020chernozhukovLocallyRobustSemiparametric} with respect to estimation
of \(m _{\ast}(a, \cdot)\). To see this, denote
\begin{equation}
  \begin{split}
    \varphi \left( y, a, z; \widetilde{m} (0, \cdot), \widetilde{m} (1, \cdot),
    b \right) =
    & \ \frac{a}{\pi (\mathbb{S} (z))} \left[ y - \widetilde{m} (1, z)
      \right] - \frac{1 - a}{1 - \pi (\mathbb{S} (z))} \left[ y - \widetilde{m}
      (0, z) \right] \\
    & + \left[ \widetilde{m} (1, z) - \widetilde{m} (0, z) - b \right]
  \end{split}
  \label{eqn--eif-moment-cond}
\end{equation}
so that the efficient influence function satisfies \(\varphi_{0} (\cdot) =
\varphi \left( \cdot; m_{\ast} (0, \cdot), m_{\ast} (1, \cdot), \beta_{0}
\right)\). If \(m_{\ast} (0, \cdot), m_{\ast} (1, \cdot)\) were known, then we
could estimate \(\beta\) by solving the sample equivalent of the moment
condition
\begin{equation*}
  \mathbb{E} \left[ \varphi \left( Y, A, Z; m_{\ast} (0, \cdot), m_{\ast} (1,
  \cdot), \beta_{0} \right) \right] = 0.
\end{equation*}
The local robustness property stems from noting that at \(\beta_{0}\),
\begin{equation*}
  \frac{d}{d \tau} \mathbb{E} \left[ \varphi \left( Y, A, Z; (1 - \tau) m_{\ast}
  (0, \cdot) + \tau \widetilde{m} (0, \cdot), (1 - \tau) m_{\ast} (1, \cdot) +
  \tau \widetilde{m} (1, \cdot), \beta_{0} \right) \middle] \right|_{\tau = 0} =
  0.
\end{equation*}
This local robustness property allows for \(n^{- \frac{1}{2}}\) consistent and
efficient estimation of the ATE by using plug-in estimates of \(m_{\ast} (a,
\cdot)\) under much weaker conditions. In particular, we do not require
smoothness conditions for \(m_{\ast} (a, \cdot)\) to achieve the efficiency
bound in \eqref{eqn--speb-car}.

In the i.i.d. case (\th\ref{eg--sra}), an ``ideal'' efficient estimator of the
ATE would be \(\widetilde{\beta}^{\ast}_{n} = \beta_{0} + n^{- 1} \sum_{i =
1}^{n} \varphi_{0} \left( Y_{n i}, A_{n i}, Z_{i} \right)\). Since adding
\(b\) to \(\varphi\) in \eqref{eqn--eif-moment-cond} removes it as an argument,
this estimator is
\begin{equation}
  \widetilde{\beta}^{\ast}_{n} = \frac{1}{n} \sum_{i = 1}^{n} \left\{
  \frac{A_{n i} \left[ Y_{n i} - m_{\ast} \left( 1, Z_{i} \right) \right]}{\pi
  \left( \mathbb{S} \left( Z_{i} \right) \right)} - \frac{\left( 1 - A_{n i}
  \right) \left[ Y_{n i} - m_{\ast} \left( 0, Z_{i} \right) \right]}{1 - \pi
  \left( \mathbb{S} \left( Z_{i} \right) \right)} + m_{\ast} \left( 1, Z_{i}
  \right) - m_{\ast} \left( 0, Z_{i} \right) \right\}.
  \label{eqn--infeasible-eif-estimator}
\end{equation}
\(\widetilde{\beta}_{n}^{\ast}\) is the familiar AIPW estimator of
\citet{1994robinsEstimationRegressionCoefficients},
\citet{1995robinsAnalysisSemiparametricRegression}, and
\citet{1999scharfsteinAdjustingNonignorableDrop} if \(m_{\ast} (a, \cdot)\) were
known. A feasible estimator must replace these with first stage estimates that
can be computed from the sample. Since it will be convenient later on, we allow
for the estimates of \(m_{\ast} (a, \cdot)\) to vary across observations. We
will also allow for the experimenter to use estimated treatment proportions
rather than the known true ones. That is, let \(\widehat{\pi}_{n} (s)\) be a
consistent estimator for \(\pi (s)\) for each \(s \in
\mathbb{N}_{\mathcal{S}}\). A feasible estimator for \(\beta_{0}\) would be
\begin{equation}
  \begin{split}
    \widehat{\beta}^{\ast}_{n} =
    & \ \frac{1}{n} \sum_{i = 1}^{n} \frac{A_{n i} \left[ Y_{n i} -
      \widehat{m}_{n i} \left( 1, Z_{i} \right) \right]}{\widehat{\pi}_{n}
      \left( \mathbb{S} \left( Z_{i} \right) \right)} - \frac{1}{n} \sum_{i =
      1}^{n} \frac{\left( 1 - A_{n i} \right) \left[ Y_{n i} - \widehat{m}_{n i}
      \left( 0, Z_{i} \right) \right]}{1 - \widehat{\pi}_{n} \left( \mathbb{S}
      \left( Z_{i} \right) \right)} \\
    & + \frac{1}{n} \sum_{i = 1}^{n} \left[ \widehat{m}_{n i} \left( 1, Z_{i}
      \right) - \widehat{m}_{n i} \left( 0, Z_{i} \right) \right].
  \end{split}
  \label{eqn--feasible-eif-estimator}
\end{equation}
It is straightforward to show that \(\widetilde{\beta}^{\ast}_{n}\) achieves the
efficiency bound \(\mathbb{V}_{\ast}\) in \eqref{eqn--speb-car} under the CAR
procedures satisfying our assumptions. Our plug-in estimators \(\widehat{m}_{n
i}\) will be cross-fitted estimators as described in
\citet{2018chernozhukovDoubleDebiasedMachine} (see their DML2 estimators). We
now provide conditions on the estimators \(\widehat{m}_{n i}\), and
\(\widehat{\pi}_{n, a} (s)\) as well as the overall model under which the
difference between \(\widehat{\beta}^{\ast}_{n}\) to
\(\widetilde{\beta}^{\ast}_{n}\) in \eqref{eqn--eif-feas-infeas-diff} are
asymptotically negligible after scaling by \(\sqrt{n}\). We start with
assumptions on a sequence of nonparametric estimators \(\widehat{m} (n, a,
\cdot; \cdot)\) from which \(\widehat{m}_{n i} (a, \cdot)\) will be constructed.

\begin{assumption}
\th\label{asm--mhat-L2}
For \(a \in \{0, 1\}\), \(\widehat{m}\left( n, a, z; y_{1} (a), z_{1}, \dots,
y_{n} (a), z_{n} \right)\) is an estimator sequence satisfying
\begin{equation}
  \lim_{n \to \infty} \mathbb{E}_{Q^{n + 1}} \left[ \left( \widehat{m} \left( n,
  a, Z; Y_{1} (a), Z_{1}, \dots, Y_{n} (a), Z_{n} \right) - m (a, Z; Q)
  \right)^{2} \right] = 0,
  \label{eqn--mhat-univ-L2-consistent}
\end{equation}
when \(Z\), \(\left\{ Y_{i} (0), Y_{i} (1), Z_{i} \right\}_{i = 1}^{n}\) are
i.i.d. with distribution \(Q \in \mathbf{Q} \left( \widehat{m} \right)\) with
\(\mathbf{Q} \left( \widehat{m} \right) \subseteq \mathbf{Q}\). In the above,
\(m (a, z; Q) = \mathbb{E}_{Q} \left[ Y (a) | Z = z \right]\) and
\(\mathbb{E}_{Q^{n + 1}}\) is expectation under the product measure \(Q^{n +
1}\). Furthermore, \(Q_{0} \in \mathbf{Q} \left( \widehat{m} \right)\).
\end{assumption}

The estimator \(\widehat{m} (n, a \cdot; \cdot)\) is allowed to use the entire
sample of potential outcomes for \(a\) as well as the covariates. Of course, all
of these are not available - \(\widehat{m} (n, a \cdot; \cdot)\) can only be
computed from the treatment group for treatment \(a\). This will be accounted
for later on when we describe the cross-fitting procedure that leads to
\(\widehat{m}_{n i}\). \th\ref{asm--mhat-L2} requires that \(\widehat{m}\) be
\(L_{2}\) consistent for the true conditional expectation \(m (\cdot; Q)\) when
the data are generated according to \(Q\) which belongs to an appropriately
chosen subset \(\mathbf{Q} \left( \widehat{m} \right) \subseteq
\mathbf{Q}\) which also contains \(Q_{0}\). Note that the \(L_{2}\) convergence
rate of \(\widehat{m}\) can be arbitrarily slow. For our purposes, simple
consistency of this kind is enough. However, \(\mathbf{Q} \left( \widehat{m}
\right)\) can be formed by restrictions that impose a rate condition on
\(\widehat{m}\). For instance, the experimenter might impose
smoothness conditions required in existing literature for the particular
\(\widehat{m}\) to obtain particular rates - e.g. a H\"older or Sobolev ball. If
\(\widehat{m}\) is a series estimator, these can be found for
example in \citet{1997neweyConvergenceRatesAsymptotic} and
\citet{2015belloniSomeNewAsymptotic}. For nonlinear sieves such as
neural networks, one can see for
\citet{1994barronApproximationEstimationBounds},
\citet{1999chenImprovedRatesAsymptotic},
\citet{2020schmidt-hieberNonparametricRegressionUsing} and
\citet{2021farrellDeepNeuralNetworks}. For kernel
estimators with uniform rates, see \citet{1996masryMultivariateLocalPolynomial}
and \citet{2008hansenUniformConvergenceRates}. Other restrictions may be
functional form restrictions, e.g., random forests are known to be \(L_{2}\)
consistent for additive models - see
\citet{2015scornetConsistencyRandomForests}. When \(\mathbf{Q} \left(
\widehat{m} \right) = \mathbf{Q}\), we will say that \(\widehat{m}\) is
\emph{universally \(L_{2}\) consistent}. This is a property enjoyed by the
Nadaraya-Watson Kernel estimator - see for instance
\citet{1980devroyeDistributionFreeConsistencyResults} and
\citet{1980spiegelmanConsistentWindowEstimation}. Next, we state assumptions
that describe the cross-fitting approach to be used here.

\begin{assumption}
\th\label{asm--sample-split}
Let \(J \in \mathbb{N} \setminus \{1\}\) be given. Let \(\mathbf{U}_{n} =
\left\{ U_{n} (a, s) : s \in \mathbb{N}_{\mathcal{S}}, a \in \{0, 1\} \right\}\)
be a set of independent \(\mathrm{Uniform} ([0, 1])\) random variables on
\((\Omega, \mathcal{F}, \Pr)\) such that \(\mathbf{U}_{n} \indep \left(
\mathbf{W}_{n}, \mathbf{A}_{n} \right)\). Then,
\begin{equation*}
  \left\{ \mathcal{G}_{n, j} (a, s) : j \in \mathbb{N}_{J}, a \in \{0, 1\}, s
  \in \mathbb{N}_{\mathcal{S}} \right\}
\end{equation*}
are random subsets of \(\mathbb{N}_{n}\) satisfying the following for each \(s
\in \mathbb{N}_{s}\) and \(a \in \{0, 1\}\).
\begin{enumerate}
\item \label{asm--sample-split-indep}
  \(\mathcal{G}_{n 1} (a, s), \dots, \mathcal{G}_{n J} (a, s)\) depend only on
  \(U_{n} (a, s)\) and \(N_{n} (a, s)\).\footnote{Formally, \(\mathcal{G}_{n 1}
  (a, s), \dots, \mathcal{G}_{n J} (a, s)\) are measureable with respect to the
  \(\sigma\)-algebra generated by the pair \(U_{n} (a, s), N_{n} (a, s)\)}
\item \label{asm--sample-split-size}
  For each \(s \in \mathbb{N}_{\mathcal{S}}\) and \(a \in \{0, 1\}\),
  \(\mathcal{G}_{n 1} (a, s), \dots, \mathcal{G}_{n J} (a, s)\) forms a
  partition of \(\left\{ i \in \mathbb{N}_{n} : A_{n i} = a, S_{i} = s
  \right\}\) satisfying
  \begin{equation}
    \begin{split}
      \left| \mathcal{G}_{n} (a, s, j) \right| =
      & \ \left\lfloor N_{n} (a, s) / J \right\rfloor \quad j = 1, \dots, J -
        1. \\
      \left| \mathcal{G}_{n} (a, s, J) \right| =
      & \ N_{n} (a, s) - (J - 1) \cdot \left\lfloor N_{n} (a, s) / J
        \right\rfloor.
    \end{split}
    \label{eqn--sample-split-size}
  \end{equation}
\end{enumerate}
For each \(s \in \mathbb{N}_{\mathcal{S}}\) and \(j \in
\mathbb{N}_{J}\), define the \(j\)\textsuperscript{th} fold within stratum \(s\)
by
\begin{equation}
  \mathcal{G}_{n} (s, j) = \mathcal{G}_{n} (1, s, j) \cup \mathcal{G}_{n} (0,
  s, j).
  \label{eqn--treat-control-combine-sample-split}
\end{equation}
\end{assumption}

\begin{assumption}
\th\label{asm--split-m}
Let \(\left\{ \mathcal{G}_{n j} (s) : j \in \mathbb{N}_{J}, s \in
\mathbb{N}_{\mathcal{S}} \right\}\) be as in
\eqref{eqn--treat-control-combine-sample-split} in
\th\ref{asm--sample-split} and let \(\widehat{m}\) be as in
\th\ref{asm--mhat-L2}. For each \(i \in \mathbb{N}_{n}\), denote
\begin{equation}
  \Gamma_{n} (a, i) = \bigcup \left\{ \mathcal{G}_{n j} \left( a, S_{i} \right)
  : j \in \mathbb{N}_{J} \text{ such that } i \notin \mathcal{G}_{n j} \left(
  S_{i} \right) \right\}.
  \label{eqn--estimation-set-m-for-i}
\end{equation}
The estimator \(\widehat{m}_{n i}\) is defined by
\begin{equation}
  \widehat{m}_{n i} (a, z) = \widehat{m} \left( \left| \Gamma_{n} (a, i)
  \right|, a, z; \left\{ Y_{\iota} (a), Z_{\iota} : \iota \in \Gamma_{n} (a, i)
  \right\} \right).
\end{equation}
\end{assumption}

\begin{assumption}
\th\label{asm--weak-balance}
\(\widehat{\pi}_{n} (s)\) is either \(\pi (s)\) or \(N_{n} (s)^{- 1} N_{n} (1,
s)\). Furthermore, if the latter is true, under \(Q_{0}\) defined in
\th\ref{asm--Q} and \(\alpha_{n} (\cdot | \cdot)\) in \th\ref{asm--treat-strat},
for each \(s \in \mathbb{N}_{\mathcal{S}}\),
\begin{equation}
  \sqrt{n} \left| \frac{N_{n} (1, s)}{N_{n} (s)} - \pi (s) \right| =
  O_{\mathrm{p}} (1).
  \label{eqn--weak-balance}
\end{equation}
\end{assumption}

\th\ref{asm--sample-split} first requires the experimenter to split each
treatment group within a given stratum into \(J\) folds in a manner that is
independent to the overall sample using information only about the group
size. These folds are further required to satisfy an ``equal size'' requirement
specific to the treatment group within the stratum. For a given stratum, the two
treatment groups corresponding to fold \(j\) are then combined which produces
\(J\) folds for the overall stratum. We first split treatment groups by strata
into folds to ensure that a given combined fold has units in both treatment and
control groups. \th\ref{asm--split-m} requires that the estimator
\(\widehat{m}_{n i}\) for \(i\)\textsuperscript{th} observation be
constructed from \(\widehat{m}\) in \th\ref{asm--mhat-L2} using observations in
the same stratum as \(i\) but not in the same fold as \(i\). Finally,
\th\ref{asm--weak-balance} requires that the treatment proportion estimator
\(\widehat{\pi}_{n} (s)\) will be either the true treatment proportions or the
sample treatment proportions. For the latter, \th\ref{asm--weak-balance}
requires require consistency at a \(1 / \sqrt{n}\) rate. This corresponds to
requiring that \(\alpha_{n} (\cdot | \cdot)\) in \th\ref{asm--treat-strat} at
least satisfy ``weak balance'' as defined in the paragraph after
\th\ref{eg--sbr}. The key result is \th\ref{thm--efficient-estimation-car}
below.

\begin{theorem}
\th\label{thm--efficient-estimation-car}
Let \th\ref{asm--Q,asm--iid-Q,asm--treat-strat} hold. Then
 \begin{equation}
  \sqrt{n} \left( \widetilde{\beta}^{\ast}_{n} - \beta_{0} \right) =
  \frac{1}{\sqrt{n}} \sum_{i = 1}^{n} \varphi_{0} \left( Y_{n i}, A_{n i}, Z_{i}
  \right) \overset{\mathrm{d}}{\to} \mathcal{N} \left( 0, \mathbb{V}_{\ast}
  \right).
 \label{eqn--eif-gives-speb}
\end{equation}
Suppose in addition that
 \th\ref{asm--mhat-L2,asm--sample-split,asm--split-m,asm--weak-balance}
hold. Then \(\widehat{\beta}^{\ast}_{n}\) in \eqref{eqn--feasible-eif-estimator}
satisfies
\begin{equation}
  \sqrt{n} \left( \widehat{\beta}^{\ast}_{n} - \beta_{0} \right) = \sqrt{n}
  \left( \widetilde{\beta}^{\ast}_{n} - \beta_{0} \right) + o_{p} (1)
  \label{eqn--beta-star-efficient}
\end{equation}
so that \(\widehat{\beta}^{\ast}_{n}\) achieves the semiparametric efficiency
bound \(\mathbb{V}_{\ast}\).
\end{theorem}

\th\ref{thm--efficient-estimation-car} first establishes in
\eqref{eqn--eif-gives-speb} that the infeasible ideal estimator
\(\widetilde{\beta}^{\ast}_{n}\) achieves the efficiency bound in
\eqref{eqn--speb-car} under general CAR procedures. This implies that any
estimator of the ATE that is asymptotically linear with influence function
\(\varphi_{0}\) in \eqref{eqn--eif} achieves the semiparametric efficiency
bound. \th\ref{thm--efficient-estimation-car} then establishes in
\eqref{eqn--beta-star-efficient} that \(\widehat{\beta}^{\ast}_{n}\) as defined
by \eqref{eqn--feasible-eif-estimator} has this aforementioned property (under
the additional assumption of this section). The only additional restrictions
placed on the overall semiparametric model \(\mathbf{Q}\) are ones defining
\(\mathbf{Q} \left( \widehat{m} \right)\) specific to the choice of estimator
\(\widehat{m}\) in \th\ref{asm--mhat-L2}. As a further corollary
(\th\ref{cor--efficient-estimation-car} below), we can show that \(\mathbf{Q}
\left( \widehat{m} \right) = \mathbf{Q}\) when \(\widehat{m}\) is the
Nadaraya-Watson kernel regression estimator. Note that the significance of
\(\mathbf{Q} \left( \widehat{m} \right) = \mathbf{Q}\) is that efficient
estimation is possible (pointwise) over the whole of \(\mathbf{Q}\). That is,
the efficiency bound can be achieved under the same conditions used for its
derivation. This is a novel result since the overall semiparametric problem is
non-trivial and involves estimation of (possibly) infinite-dimensional
conditional means. We now first describe the Nadaraya-Watson estimator, state
assumptions on its tuning parameters and then provide the result. In what
follows, we use the normalization that \(0/0 = 0\). Let \(\widehat{m}\)
in \th\ref{asm--mhat-L2} be defined by
\begin{equation}
  \widehat{m} \left( n, a, z; y_{1} (a), z_{1}, \dots, Y_{n} (a), z_{n}
  \right) = \frac{\sum_{j = 1}^{n} y_{j} \cdot \kappa \left( h_{n}^{- 1}
  \left( z_{j} - z \right) \right)}{\sum_{j = 1}^{n} \kappa \left(
  h_{n}^{- 1} \left( z_{j} - z \right) \right)}.
  \label{eqn--kernel-estimators}
\end{equation}
\th\ref{asm--kernel-bw} below states the assumptions required for the bandwidth
sequence \(\left\{ h_{n} : n \in \mathbb{N} \right\}\) and the kernel function
\(\kappa\).
\begin{assumption}
\th\label{asm--kernel-bw}
The sequence \(\left\{ h_{n} : n \in \mathbb{N} \right\}\) and the function
\(\kappa : \mathbb{R}^{k} \to \mathbb{R}\) satisfy the following.
\begin{enumerate}[label=(\alph*)]
\item
  \label{asm--bw}
  \(h_{n} > 0\) for each \(n \in \mathbb{N}\), \(\lim_{n \to \infty} h_{n} =
  0\) and \(\lim_{n \to \infty} n h_{n}^{k} = \infty\).
\item \label{asm--kernel}
  For some \(\underline{\kappa}, \overline{\kappa}, r, R \in (0, \infty)\),
  \(\underline{\kappa} \cdot \mathbb{I} (\| u \| \leq r) \leq \kappa (u) \leq
  \overline{\kappa} \cdot \mathbb{I} (\| u \| \leq R)\) for each \(u \in
  \mathbb{R}^{k}\).
\end{enumerate}
\end{assumption}

\th\ref{asm--kernel-bw} maintains exactly the same conditions required by
\citet{1980devroyeDistributionFreeConsistencyResults} and
\citet{1980spiegelmanConsistentWindowEstimation} for universal \(L_{2}\)
consistency. Note that \th\ref{asm--kernel-bw} \ref{asm--kernel} requires the
kernel function \(\kappa\) to be non-negative, to have compact support and to be
strictly positive in a closed neighborhood of the origin with non-empty
interior. The following result shows that efficient estimation is possible under
exactly the same conditions required to derive the efficiency bound.

\begin{corollary}
\th\label{cor--efficient-estimation-car}
In addition to the conditions of \th\ref{thm--efficient-estimation-car}, suppose
that \(\widehat{m}\) in \th\ref{asm--mhat-L2} is defined by
\eqref{eqn--kernel-estimators} and that \(\left\{ h_{n} : n \in \mathbb{N}
\right\}\) and \(\kappa\) satisfy \th\ref{asm--kernel-bw}. Then \(\mathbf{Q}
\left( \widehat{m} \right) = \mathbf{Q}\) and \(\widehat{\beta}^{\ast}_{n}\) in
\eqref{eqn--beta-star-efficient} achieves the efficiency bound over all of
\(\mathbf{Q}\).
\end{corollary}

The only requirement for the phenomenon in
\th\ref{cor--efficient-estimation-car} to occur is the universal consistency
requirement, i.e. that \(\mathbf{Q} \left( \widehat{m} \right) =
\mathbf{Q}\). This is not special to the Nadaraya-Watson estimator - one can
replace this with a nearest neighbor estimator and replace the conditions in
\th\ref{asm--kernel-bw} with those of
\citet{1977stoneConsistentNonparametricRegression}. As for local polynomial
kernel estimators, these are not universally \(L_{2}\) consistent in their
typical form (see \citet[Section 5.4,
p. 80-81]{2002gyoerfiDistributionFreeTheory}), but adjusted versions can have
the universal \(L_{2}\) consistency property (see
\citet{2002kohlerUniversalConsistencyLocal}). These adjustments introduce a
number of additional tuning parameters beyond the bandwidth. For universal
\(L_{2}\) consistency conditions for series estimators, see \citet[Chapter
10]{2002gyoerfiDistributionFreeTheory}.

We now provide a brief sketch of the main arguments that lead to
\eqref{eqn--beta-star-efficient} in \th\ref{thm--efficient-estimation-car}.
The difference between \(\sqrt{n} \left( \widehat{\beta}_{n}^{\ast} - \beta_{0}
\right)\) and \(\sqrt{n} \left( \widetilde{\beta}_{n}^{\ast} - \beta_{0}
\right)\) and can be summarized as
\begin{equation}
  \begin{split}
    \sqrt{n} \left( \widehat{\beta}^{\ast}_{n} - \beta_{0} \right) =
    & \ \sqrt{n} \left( \widetilde{\beta}^{\ast}_{n} - \beta_{0} \right) + R_{1,
      n} - R_{0, n} + \widetilde{R}_{1 n} - \widetilde{R}_{0 n} \\
    \text{where} \quad R_{a, n} =
    & \ \frac{1}{\sqrt{n}} \sum_{i = 1}^{n} \frac{\mathbb{I} \left\{ A_{n i} = a
      \right\} - \widehat{\pi}_{n, a} \left( \mathbb{S} \left( Z_{i} \right)
      \right)}{\widehat{\pi}_{n, a} \left( \mathbb{S} \left( Z_{i} \right)
      \right)} \left[ \widehat{m}_{n i} \left(a, Z_{i} \right) - m_{\ast} \left(
      a, Z_{i} \right) \right]. \\
    \text{and} \quad \widetilde{R}_{a, n} =
    & \frac{1}{\sqrt{n}} \sum_{i = 1}^{n} \left[ \frac{1}{\widehat{\pi}_{n, a}
      \left( \mathbb{S} \left( Z_{i} \right) \right)} - \frac{1}{\pi_{a} \left(
      \mathbb{S} \left( Z_{i} \right) \right)} \right] \mathbb{I} \left\{ A_{n
      i} = a \right\} \left[ Y_{n i} - m_{\ast} \left( a, Z_{i} \right) \right].
  \end{split}
  \label{eqn--eif-feas-infeas-diff}
\end{equation}
In the above, \(\pi_{a} (s) = \pi (s)^{a} [1 - \pi (s)]^{1 - a}\) and
\(\widehat{\pi}_{n, a} (s)\) is defined analogously for each \(s \in
\mathbb{N}_{\mathcal{S}}\) and \(a \in \{0, 1\}\). Showing
\eqref{eqn--beta-star-efficient} is a matter of showing that \(R_{a, n}
= o_{\mathrm{p}} (1)\) and \(\widetilde{R}_{a, n} = o_{\mathrm{p}}
(1)\). The remainder term \(\widetilde{R}_{a, n}\) is straightforward to deal
with - a central limit theorem applies to
\begin{equation*}
  (1 / \sqrt{n}) \sum_{i = 1}^{n} \mathbb{I} \left\{ A_{n i} = a, \mathbb{S}
  \left( Z_{i} \right) = s \right\} \left[ Y_{n i} - m_{\ast} \left( a, Z_{i}
  \right) \right]
\end{equation*}
and the difference between estimated and true inverse treatment probabilities is
constant across observations within a given stratum and converging to zero. For
\(R_{a, n}\), we deal with these by first splitting them into a sum of
\(\mathcal{S} \cdot J\) terms each corresponding to a given stratum and a given
fold within that stratum. Each of these terms can be treated separately since we
are trying to prove convergence in probability to zero. The argument then
proceeds by noting that within the stratum-fold remainder each component of the
form \(\widehat{\pi}_{n, a} (s)^{- 1} \left( \mathbb{I} \left\{ A_{n
i} = a \right\} - \widehat{\pi}_{n, a} (s) \right)\) is ``approximately'' mean
zero and has finite variance. Furthermore, for the term \(\widehat{m}_{n i}
\left(a, Z_{i} \right) - m_{\ast} \left( a, Z_{i} \right)\) has second moment
tending to zero. By careful use of the independence between
the evaluation point \(Z_{i}\) and the estimation sample (observations outside
the corresponding fold), the stratum-fold specific version of \(R_{a, n}\) is
shown to be \(o_{\mathrm{p}} (1)\) via Chebychev's inequality. Since the
overall \(R_{a, n}\) is the sum of its stratum-fold specific counterparts and
there are \(\mathcal{S} \cdot J\) (i.e. finitely many) of these, the conclusion
follows from Slutsky's theorem. A rigorous version of these arguments is left to
the appendix.

The key finding here is that knowledge of the propensity score and its finite
support structure reduces the burden for efficient estimation considerably. The
fact that knowledge of the propensity score can reduce this burden was also
noticed by \citet{2021aronowNonparametricIdentificationNot} and
\citet{2018rotheFlexibleCovariateAdjustments}.
\citet{2021aronowNonparametricIdentificationNot}
show that potentially non-linear parametric regression adjustments can improve
estimation accuracy, generalizing \citet{2001yangEfficiencyStudyEstimators} and
\citet{2013linAgnosticNotesRegression} to the nonlinear
case. \citet{2018rotheFlexibleCovariateAdjustments} shows that
semiparametrically efficient estimation is possible under a fairly broad class
of estimators using Donsker conditions.
\citet{2018rotheFlexibleCovariateAdjustments} also shows that
efficient estimation is possible under slightly weaker conditions by using a
leave-one-out locally linear kernel regression estimator by leveraging the
results of \citet{2019rothePropertiesDoublyRobust}. However, both still require
dimension-dependent smoothness conditions on the conditional means to be
estimated. To the best of our knowledge, existing results using nonparametric
estimators in this context all require existence of densities for baseline
covariates that are bounded away from zero on their support and higher order
differentiability on conditional means and possibly also densities. Our results
do not require any of these for efficient estimation.

It should be noted that while \(\widehat{\beta}^{\ast}_{n}\) in
\eqref{eqn--feasible-eif-estimator} achieves the efficiency bound a few issues
prevent it from being usable immediately in the context of statistical
inference. In particular, we have not provided consistent estimators of the
asymptotic variance \(\mathbb{V}_{\ast}\) in \eqref{eqn--speb-car}. It may be
possible to construct these directly from the sample and the nonparametric
estimators \(\widehat{m}_{n i}\). Another possibility might be a bootstrap
procedure to produce confidence intervals for hypothesis tests. Another issue we
have not addressed is data-based choice of tuning parameters for the estimators
\(\widehat{m}_{n i}\) when these are based on the Nadaraya-Watson kernel
estimator. Since our main questions are the efficiency bound and the conditions
under which it can be achieved, we leave the construction of consistent
inference procedures from \(\widehat{\beta}^{\ast}_{n}\) and data-based tuning
parameter selection for future research.

%%% Local Variables:
%%% mode: latex
%%% TeX-master: "../2023_semipar_eff_car"
%%% End:
% LocalWords:  SSRA SPBR ATT

% ! TEX root = ../2023_semipar_eff_car.tex

\section{Monte Carlo evidence}
\label{sec--sims}

In this section, we examine the finite sample performance of feasible efficient
estimator \(\widehat{\beta}_{n}^{\ast}\) from
\eqref{eqn--feasible-eif-estimator}. We start by describing the data generating
processes (DGPs) for our simulations. We consider values of \(n \in \{500, 1000,
2000, 4000, 8000\}\) and \(k \in \{1, 5\}\). The DGPs are all of the form
\begin{equation}
  Y_{i} (a) = m_{\ast} \left( a, Z_{i} \right) + \sigma \left( a, Z_{i} \right)
  \cdot \varepsilon_{i}.
\end{equation}
In all cases, we have \(Z_{i} \sim \mathrm{Uniform} \left([-1, 1]^{k} \right)\)
and \(\varepsilon_{i} \sim \mathcal{N} (0, 1)\). For each DGP, we conduct 5000
Monte Carlo simulations. In each simulation round, we compute
\(\widetilde{\beta}_{n}^{\ast}\) in \eqref{eqn--infeasible-eif-estimator} and
\(\widehat{\beta}_{n}^{\ast}\) in \eqref{eqn--feasible-eif-estimator} as well as
additional estimators of the ATE to be described below. We report mean squared
errors as well as biases of these estimators relative to the true value of the
ATE. For stratification, we stratify on the basis of the first component of the
covariates \(Z_{i}\). We consider \(\mathcal{S} \in \{5, 20\}\) and for each
value of \(\mathcal{S}\), strata are constructed by splitting the interval
\([-1, 1]\) into adjacent segments of length \(2 / \mathcal{S}\). We report
results with constant assignment proportions across all strata as well as
proportions that vary by strata. We present results with SPBR here, and defer
results with SSRA to the appendix. Inspecting both, the reader can see that the
choice of assignment mechanism has no visible impact on the performance of any
of the estimators considered here. For constant proportions, we set assignment
proportions to 1/2 and for varying proportions we use
\begin{equation}
  \bm{\pi} =
  \begin{cases}
    (0.3, 0.4, 0.5, 0.6, 0.7) & \text{if } \mathcal{S} = 5, \\
    (0.325, 0.35, 0.375, \dots, 0.75, 0.775, 0.8) & \text{if } \mathcal{S} = 20.
  \end{cases}
\end{equation}
We use the boldfaced \(\bm{\pi}\) here to distinguish assignment probabilities
from the mathematical constant \(\pi = 3.14159\dots\), since our choice of
conditional mean functions involve the latter through trigonometric
functions. We consider four different DGP's by choice of \(k\), \(m_{\ast} (a,
\cdot)\) and \(\sigma (a, \cdot)\). In what follows, DGP's 1 and 2 have \(k =
1\) whereas DGP's 3 and 4 have \(k = 5\). In all DGP's, the true value of the
ATE is \(\beta_{0} = 0\). This is done for convenience.
\begin{enumerate}
\item DGP 1: \(m_{\ast} (0, z) = \sin (10 \pi z)\), \(m_{\ast} (1, z) = m_{\ast}
  (0, z) + 2 \cos (10 \pi z)\), \(\sigma (0, z) = 1 + |z|\) and \(\sigma (1, z)
  = \sqrt{2} \cdot \sigma (0, z)\)
\item DGP 2: \(m_{\ast} (0, z) = \mathrm{sign} (z) \cdot \lfloor 10 z \rfloor /
  10\), \(m_{\ast} (1, z) = m_{\ast} (0, z) + 2 m_{\ast} (0, z)^{3}\), \(\sigma
  (\cdot)\) as in DGP 1.
\item DGP 3: \(\sigma (0, z) = 1\), \(\sigma (1, z) = \sqrt{2}\), \(m_{\ast} (0,
  z) = \lambda \left( z_{5}, z_{4}, z_{3}, z_{2}, z_{1} \right)\) and \(m_{\ast}
  (1, z) = \lambda (z) + 2 \sin \left( 2 \pi z_{1} z_{2} \right)\)
  \begin{equation*}
    \lambda (z) = \cos \left( 2 \pi z_{1} z_{2} \right) + \left( z_{3} + z_{4} -
    1 \right)^{2} + \frac{z_{5}}{2}
  \end{equation*}
\item DGP 4: \(\sigma (\cdot)\) and \(m_{\ast} (\cdot)\) as in DGP 3, but with
  \(\lambda\) defined instead by
  \begin{equation*}
    \lambda (z) = \ \cos \left( 2 \pi z_{1} z_{2} \right) + \left( z_{3} + z_{4}
    - 1 \right)^{2} + \frac{z_{5}}{2} + \mathbb{I} \left( z_{1} \geq 0 \right).
  \end{equation*}
\end{enumerate}
In both DGP 1 and DGP 3, \(m_{\ast} (1, \cdot)\) and \(m_{\ast} (0, \cdot)\) are
analytic and exhibit considerable non-linear variation within strata with either
\(\mathcal{S} = 5\) or \(\mathcal{S} = 20\). In DGP 2, \(m_{\ast} (1, \cdot)\)
and \(m_{\ast} (0, \cdot)\) have 19 discontinuities and are exactly fully
saturated regressions with \(\mathcal{S} = 20\) but not \(\mathcal{S} = 5\). DGP
4 adds a jump discontinuity at \(z_{1} = 0\) to the conditional mean functions
of DGP 3.

For comparison, we also examine the finite sample performance of the infeasible
efficient estimator \(\widetilde{\beta}_{n}^{\ast}\), the fully saturated
regression estimator of the ATE \(\widehat{\beta}_{n, \mathrm{SAT}}\) in
\eqref{eqn--est-sat}, as well as the imputation estimator of
\citet{1998hahnRolePropensityScore} using the cross-fitted Nadaraya-Watson
estimator:
\begin{equation}
  \widehat{\beta}_{n, \mathrm{IMP}} = \frac{1}{n} \sum_{i = 1}^{n} \left[ Y_{n
  i} A_{n i} + \left( 1 - A_{n i} \right) \widehat{m}_{n, i} \left( 1, Z_{i}
  \right) \right] - \frac{1}{n} \sum_{i = 1}^{n} \left[ \left( 1 - A_{n i}
  \right) Y_{n i} + A_{n i} \widehat{m}_{n, i} \left( 0, Z_{i} \right) \right].
  \label{eqn--est-reg}
\end{equation}
The estimator above is an imputation estimator since it imputes the treatment
\(a\) potential outcome for observation \(i\) with the predicted value
\(\widehat{m}_{n i} \left( a, Z_{i} \right)\) whenever that potential outcome is
unobserved. This estimator is also asymptotically linear and with influence
function \(\varphi_{0}\) when the functions \(m_{\ast} (a, \cdot)\) are
sufficiently smooth - see for instance
\citet{1994chengNonparametricEstimationMean}. For our nonparametric estimators
\(\widehat{m}_{n, i} \left( a, \cdot \right)\), we use bandwidths of the form
\(c_{k} n^{- 1 / (4 + k)}\) and the uniform kernel so that
\th\ref{asm--kernel-bw} is satisfied. For DGP 1 and DGP 3, the bandwidths are
rate-optimal in terms of integrated mean-squared error (see
\citet{1982stoneOptimalGlobalRates}) since the conditional mean functions
\(m_{\ast} (a, \cdot)\) are twice continuously differentiable. For \(k = 1\), we
set \(c_{k} = 1 / \sqrt{3}\) and for \(k = 5\), we set \(c_{k} = 3\). In our
results, estimated mean squared errors for \(\widetilde{\beta}_{n}^{\ast}\) act
as an estimate of the semiparametric efficiency bound. The comparison between
\(\widehat{\beta}_{n}^{\ast}\) and \(\widetilde{\beta}_{n}^{\ast}\) illustrates
how far from optimal the feasible estimator can be due to noise in the
estimation of the conditional mean functions. Following the discussion in
Section \ref{sec--speb}, the fully saturated regression estimator
\(\widehat{\beta}_{n, \mathrm{SAT}}\) in \eqref{eqn--est-sat} achieves the
semiparametric efficiency bound among all estimators that use only information
contained in the strata. This is \(\mathbb{V}_{\mathrm{SAT}}\), defined in
\eqref{eqn--speb-sat}. Thus, comparing \(\widehat{\beta}_{n, \mathrm{SAT}}\) to
\(\widetilde{\beta}_{n}^{\ast}\) illustrates the magnitude of efficiency loss
from ignoring information from the covariates (except in DGP 2 with
\(\mathcal{S} = 20\)). Additionally, comparing \(\widehat{\beta}_{n,
\mathrm{SAT}}\) to \(\widehat{\beta}_{n}^{\ast}\) illustrates the possible
tradeoff between finite sample and asymptotic accuracy when estimates of the
nonparametric components are potentially far from the truth. The additional
nonparametric estimator \(\widehat{\beta}_{n, \mathrm{IMP}}\) has the same
influence function as \(\widehat{\beta}_{n}^{\ast}\) but does not directly use
the structure of the influence function. Thus, comparisons between these
estimators illustrates the debiasing effect of using the efficient influence
function directly.

Tables \ref{tab--cons5} and \ref{tab--cons20} present results from simulations
with constant target treatment proportions across strata with \(\mathcal{S} =
5\) and \(\mathcal{S} = 20\) respectively. Tables \ref{tab--var5} and
\ref{tab--var20} present results with varying target treatment proportions
across strata with \(\mathcal{S} = 5\) and \(\mathcal{S} = 20\)
respectively. Some common patterns appear across examples. In the univariate
case (DGP's 1 and 2), \(\widehat{\beta}^{\ast}_{n}\) is fairly close to the
infeasible estimator \(\widetilde{\beta}^{\ast}_{n}\) in terms of mean squared
error. Furthermore, while \(\widehat{\beta}_{n, \mathrm{SAT}}\) is far from
optimal under the nonlinearities in DGP 1 and \(\mathcal{S} = 5\), the distance
is greatly reduced by using finer strata (i.e. \(\mathcal{S} = 20\)).
However, in higher dimensions, both estimators suffer albeit in different
ways. Under DGP's 3 and 4 the feasible optimal estimator
\(\widehat{\beta}^{\ast}_{n}\), despite still being consistent and
asymptotically unbiased, is much slower in its tendency towards the SPEB.
This is mainly due to the curse of dimensionality. In simulations with \(n \in
\{10000, 20000\}\) for instance, \(\widehat{\beta}^{\ast}_{n}\) is closer to
achieving the SPEB, but these sample sizes may be unrealistic for most
experiments. As we can see, \(\widehat{\beta}_{n}^{\ast}\) is also sensitive to
the number of strata - having more strata reduces the effective estimation
sample used for estimation of the conditional mean. This is particularly stark
in the examples with varying treatment proportions - in these, some treatment
groups are going to have small sample sizes within strata by design and the
aforementioned phenomenon is exacerbated. In the multivariate case, increasing
the number of strata does not help reduce the MSE of \(\widehat{\beta}_{n,
\mathrm{SAT}}\) since stratification is done on the basis of a single component
of a continuously distributed covariate. On the other hand, though its
performance in some instances (e.g. \ref{tab--var20}, DGP 3 and 4) leaves much
to be desired, in most cases with DGP 3 and 4, \(\widehat{\beta}^{\ast}_{n}\)
picks up on the nonlinear variation in the conditional means across all
dimensions. To summarize the comparison between \(\widehat{\beta}^{\ast}_{n}\)
and \(\widehat{\beta}_{n, \mathrm{SAT}}\) we find that across all simulation
designs, \(\widehat{\beta}^{\ast}_{n}\) is on average 13\% more efficient
(averaging the quotient of the MSEs). The maximal gain from using
\(\widehat{\beta}^{\ast}_{n}\) is a 40\% reduction in MSE - Table
\ref{tab--var5},  DGP 4 with \(n = 8000\).

For \(\widehat{\beta}_{n, \mathrm{IMP}}\), while it seems to perform well under
\(k = 1\) and \(\mathcal{S} = 20\) the impact of the curse of dimensionality and
its interaction with the lack of debiasing is quite stark as can be seen from
the other cases. Furthermore, the need for undersmoothing to reduce bias for
this estimator is apparent since it seems to sometimes do worse with more
effective observations. For instance, comparing results for DGP 1 in Tables
\ref{tab--cons5} and \ref{tab--cons20}, we see that increasing the number of
strata (i.e. decreasing the effective number of observations per strata)
produces bias reductions without seriously affecting variance for
\(\widehat{\beta}_{n, \mathrm{SAT}}\) or \(\widehat{\beta}_{n,
\mathrm{IMP}}\). However, reducing bandwidth to undersmooth will reduce bias and
inflate the variance of this estimator.

\begin{table}[H]
  \centering
  \caption{\(\mathcal{S} = 5\) and constant target assignments across strata.}
    \begin{tabular}{ccccccccccccc}
    \toprule
    \toprule
          &       & \multicolumn{2}{c}{\(\sqrt{n} \cdot \widetilde{\beta}^{\ast}_{n}\)} &       & \multicolumn{2}{c}{\(\sqrt{n} \cdot \widehat{\beta}^{\ast}_{n}\)} &       & \multicolumn{2}{c}{\(\sqrt{n} \cdot \widehat{\beta}_{n, \mathrm{SAT}}\)} &       & \multicolumn{2}{c}{\(\sqrt{n} \cdot \widehat{\beta}_{n, \mathrm{IMP}}\)} \\
\cmidrule{3-4}\cmidrule{6-7}\cmidrule{9-10}\cmidrule{12-13}    DGP   & \(n\) & MSE   & Bias  &       & MSE   & Bias  &       & MSE   & Bias  &       & MSE   & Bias \\
    \midrule
\multirow{5}[2]{*}{DGP 1} & 500   & 16.442 & -0.041 &       & 19.883 & -0.078 &       & 20.931 & -0.071 &       & 22.536 & -1.471 \\
      & 1000  & 15.731 & 0.024 &       & 17.724 & -0.014 &       & 19.589 & -0.040 &       & 21.316 & -1.661 \\
      & 2000  & 16.380 & -0.064 &       & 17.888 & -0.087 &       & 20.631 & -0.082 &       & 22.458 & -1.908 \\
      & 4000  & 16.084 & -0.039 &       & 17.022 & -0.076 &       & 19.891 & -0.085 &       & 21.825 & -2.011 \\
      & 8000  & 15.749 & 0.025 &       & 16.466 & 0.033 &       & 20.167 & 0.058 &       & 20.990 & -1.955 \\
    \midrule
\multirow{5}[2]{*}{DGP 2} & 500   & 14.679 & -0.012 &       & 15.139 & -0.020 &       & 14.928 & -0.014 &       & 15.090 & -0.016 \\
      & 1000  & 14.000 & -0.012 &       & 14.274 & -0.009 &       & 14.235 & -0.005 &       & 14.266 & -0.010 \\
      & 2000  & 14.604 & -0.058 &       & 14.733 & -0.061 &       & 14.836 & -0.067 &       & 14.749 & -0.060 \\
      & 4000  & 14.672 & -0.077 &       & 14.786 & -0.087 &       & 14.912 & -0.078 &       & 14.802 & -0.085 \\
      & 8000  & 14.296 & 0.006 &       & 14.327 & 0.006 &       & 14.511 & 0.000 &       & 14.326 & 0.005 \\
\midrule
\multirow{5}[2]{*}{DGP 3} & 500   & 12.593 & -0.011 &       & 18.453 & -0.016 &       & 24.129 & -0.078 &       & 18.291 & -0.685 \\
      & 1000  & 13.148 & -0.034 &       & 17.397 & -0.007 &       & 24.508 & -0.070 &       & 18.098 & -0.792 \\
      & 2000  & 12.974 & -0.019 &       & 16.380 & -0.036 &       & 24.091 & -0.008 &       & 17.481 & -0.912 \\
      & 4000  & 12.813 & -0.044 &       & 15.290 & -0.047 &       & 23.656 & -0.005 &       & 16.836 & -1.015 \\
      & 8000  & 13.060 & -0.014 &       & 15.155 & -0.030 &       & 24.447 & -0.107 &       & 16.987 & -1.087 \\
\midrule
\multirow{5}[2]{*}{DGP 4} & 500   & 12.624 & -0.021 &       & 18.743 & -0.030 &       & 24.668 & -0.090 &       & 18.566 & -0.702 \\
      & 1000  & 13.120 & -0.040 &       & 17.667 & -0.012 &       & 25.168 & -0.070 &       & 18.351 & -0.795 \\
      & 2000  & 12.932 & 0.001 &       & 16.449 & -0.014 &       & 24.615 & 0.011 &       & 17.523 & -0.893 \\
      & 4000  & 12.762 & -0.037 &       & 15.431 & -0.034 &       & 24.227 & 0.018 &       & 16.968 & -1.000 \\
      & 8000  & 13.006 & -0.020 &       & 15.197 & -0.031 &       & 25.097 & -0.109 &       & 17.076 & -1.090 \\
    \bottomrule
    \bottomrule
    \end{tabular}%
  \label{tab--cons5}
\end{table}%

\begin{table}[H]
  \centering
  \caption{\(\mathcal{S} = 20\) and constant target assignments across strata.}
    \begin{tabular}{ccccccccccccc}
    \toprule
    \toprule
          &       & \multicolumn{2}{c}{\(\sqrt{n} \cdot \widetilde{\beta}^{\ast}_{n}\)} &       & \multicolumn{2}{c}{\(\sqrt{n} \cdot \widehat{\beta}^{\ast}_{n}\)} &       & \multicolumn{2}{c}{\(\sqrt{n} \cdot \widehat{\beta}_{n, \mathrm{SAT}}\)} &       & \multicolumn{2}{c}{\(\sqrt{n} \cdot \widehat{\beta}_{n, \mathrm{IMP}}\)} \\
\cmidrule{3-4}\cmidrule{6-7}\cmidrule{9-10}\cmidrule{12-13}    DGP   & \(n\) & MSE   & Bias  &       & MSE   & Bias  &       & MSE   & Bias  &       & MSE   & Bias \\
    \midrule
\multirow{5}[2]{*}{DGP 1} & 500   & 16.290 & -0.125 &       & 18.864 & -0.126 &       & 19.008 & -0.128 &       & 18.919 & -0.133 \\
      & 1000  & 16.171 & -0.077 &       & 17.921 & -0.060 &       & 18.268 & -0.058 &       & 18.042 & -0.064 \\
      & 2000  & 16.582 & -0.151 &       & 17.939 & -0.146 &       & 18.680 & -0.143 &       & 18.225 & -0.149 \\
      & 4000  & 15.854 & -0.015 &       & 17.104 & 0.002 &       & 18.390 & 0.002 &       & 17.512 & 0.001 \\
      & 8000  & 15.462 & -0.014 &       & 16.265 & -0.024 &       & 17.811 & -0.020 &       & 16.688 & -0.021 \\
\midrule
\multirow{5}[2]{*}{DGP 2} & 500   & 14.549 & -0.102 &       & 14.785 & -0.100 &       & 14.744 & -0.100 &       & 14.769 & -0.108 \\
      & 1000  & 14.298 & -0.104 &       & 14.436 & -0.104 &       & 14.399 & -0.105 &       & 14.444 & -0.109 \\
      & 2000  & 14.878 & -0.161 &       & 14.966 & -0.158 &       & 14.933 & -0.162 &       & 15.039 & -0.161 \\
      & 4000  & 14.354 & -0.035 &       & 14.395 & -0.033 &       & 14.378 & -0.036 &       & 14.530 & -0.034 \\
      & 8000  & 13.917 & 0.000 &       & 13.954 & -0.001 &       & 13.928 & 0.000 &       & 14.091 & 0.001 \\
\midrule
\multirow{5}[2]{*}{DGP 3} & 500   & 13.114 & -0.039 &       & 24.858 & -0.640 &       & 24.637 & -0.038 &       & 18.169 & -1.151 \\
      & 1000  & 12.669 & -0.049 &       & 21.884 & -0.296 &       & 24.590 & -0.086 &       & 19.145 & -1.079 \\
      & 2000  & 12.894 & 0.024 &       & 19.287 & -0.114 &       & 24.830 & -0.003 &       & 18.681 & -1.070 \\
      & 4000  & 12.471 & -0.011 &       & 16.999 & -0.086 &       & 24.358 & -0.046 &       & 17.859 & -1.169 \\
      & 8000  & 13.402 & 0.013 &       & 16.501 & 0.004 &       & 24.877 & 0.070 &       & 18.007 & -1.088 \\
\midrule
\multirow{5}[2]{*}{DGP 4} & 500   & 13.127 & -0.049 &       & 26.421 & -0.797 &       & 25.091 & -0.056 &       & 19.007 & -1.313 \\
      & 1000  & 12.711 & -0.052 &       & 22.815 & -0.349 &       & 25.092 & -0.086 &       & 19.735 & -1.142 \\
      & 2000  & 12.893 & 0.041 &       & 19.856 & -0.110 &       & 25.412 & 0.020 &       & 18.990 & -1.073 \\
      & 4000  & 12.446 & 0.003 &       & 17.231 & -0.080 &       & 24.891 & -0.024 &       & 18.045 & -1.166 \\
      & 8000  & 13.396 & 0.019 &       & 16.616 & 0.013 &       & 25.588 & 0.077 &       & 18.091 & -1.079 \\
    \bottomrule
    \bottomrule
    \end{tabular}%
    \label{tab--cons20}%
\end{table}%

\begin{table}[H]
  \centering
  \caption{\(\mathcal{S} = 5\) and varying target assignments across strata.}
    \begin{tabular}{ccccccccccccc}
    \toprule
    \toprule
          &       & \multicolumn{2}{c}{\(\sqrt{n} \cdot \widetilde{\beta}^{\ast}_{n}\)} &       & \multicolumn{2}{c}{\(\sqrt{n} \cdot \widehat{\beta}^{\ast}_{n}\)} &       & \multicolumn{2}{c}{\(\sqrt{n} \cdot \widehat{\beta}_{n, \mathrm{SAT}}\)} &       & \multicolumn{2}{c}{\(\sqrt{n} \cdot \widehat{\beta}_{n, \mathrm{IMP}}\)} \\
\cmidrule{3-4}\cmidrule{6-7}\cmidrule{9-10}\cmidrule{12-13}    DGP   & \(n\) & MSE   & Bias  &       & MSE   & Bias  &       & MSE   & Bias  &       & MSE   & Bias \\
    \midrule
\multirow{5}[2]{*}{DGP 1} & 500   & 18.224 & 0.050 &       & 22.458 & 0.031 &       & 22.951 & 0.059 &       & 24.469 & -1.342 \\
      & 1000  & 17.620 & -0.011 &       & 20.586 & -0.055 &       & 22.667 & -0.005 &       & 24.260 & -1.689 \\
      & 2000  & 17.951 & -0.036 &       & 19.725 & -0.040 &       & 22.685 & -0.034 &       & 24.258 & -1.866 \\
      & 4000  & 17.347 & -0.039 &       & 18.410 & -0.077 &       & 21.766 & -0.101 &       & 23.362 & -2.011 \\
      & 8000  & 17.519 & 0.029 &       & 18.221 & 0.003 &       & 22.196 & -0.027 &       & 23.027 & -1.996 \\
\midrule
\multirow{5}[2]{*}{DGP 2} & 500   & 16.578 & 0.035 &       & 17.596 & 0.003 &       & 16.942 & 0.047 &       & 17.289 & 0.019 \\
      & 1000  & 16.054 & -0.038 &       & 16.472 & -0.057 &       & 16.471 & -0.050 &       & 16.364 & -0.059 \\
      & 2000  & 16.054 & -0.031 &       & 16.248 & -0.021 &       & 16.359 & -0.027 &       & 16.245 & -0.025 \\
      & 4000  & 15.899 & -0.072 &       & 16.069 & -0.085 &       & 16.134 & -0.067 &       & 16.102 & -0.076 \\
      & 8000  & 16.179 & 0.046 &       & 16.208 & 0.049 &       & 16.379 & 0.048 &       & 16.244 & 0.049 \\
\midrule
\multirow{5}[2]{*}{DGP 3} & 500   & 13.967 & -0.004 &       & 21.786 & -0.046 &       & 27.781 & 0.023 &       & 20.290 & -0.527 \\
      & 1000  & 13.941 & 0.037 &       & 20.019 & 0.041 &       & 27.937 & 0.045 &       & 19.804 & -0.643 \\
      & 2000  & 13.738 & -0.049 &       & 18.115 & -0.097 &       & 27.328 & -0.098 &       & 19.272 & -0.955 \\
      & 4000  & 13.533 & -0.020 &       & 17.139 & -0.050 &       & 27.022 & -0.066 &       & 18.684 & -1.024 \\
      & 8000  & 13.858 & -0.023 &       & 16.250 & -0.001 &       & 26.544 & 0.013 &       & 18.027 & -1.028 \\
\midrule
\multirow{5}[2]{*}{DGP 4} & 500   & 14.112 & -0.011 &       & 22.398 & -0.059 &       & 28.679 & 0.019 &       & 20.594 & -0.392 \\
      & 1000  & 13.959 & 0.043 &       & 20.355 & 0.046 &       & 28.533 & 0.050 &       & 19.955 & -0.564 \\
      & 2000  & 13.741 & -0.037 &       & 18.338 & -0.079 &       & 27.870 & -0.081 &       & 19.413 & -0.913 \\
      & 4000  & 13.599 & -0.016 &       & 17.232 & -0.045 &       & 27.575 & -0.060 &       & 18.777 & -1.012 \\
      & 8000  & 13.847 & -0.025 &       & 16.363 & -0.002 &       & 27.180 & 0.024 &       & 18.175 & -1.029 \\
    \bottomrule
    \bottomrule
    \end{tabular}%
    \label{tab--var5}%
\end{table}%

\begin{table}[H]
  \centering
  \caption{\(\mathcal{S} = 20\) and varying target assignments across strata.}
    \begin{tabular}{ccccccccccccc}
    \toprule
    \toprule
          &       & \multicolumn{2}{c}{\(\sqrt{n} \cdot \widetilde{\beta}^{\ast}_{n}\)} &       & \multicolumn{2}{c}{\(\sqrt{n} \cdot \widehat{\beta}^{\ast}_{n}\)} &       & \multicolumn{2}{c}{\(\sqrt{n} \cdot \widehat{\beta}_{n, \mathrm{SAT}}\)} &       & \multicolumn{2}{c}{\(\sqrt{n} \cdot \widehat{\beta}_{n, \mathrm{IMP}}\)} \\
\cmidrule{3-7}\cmidrule{9-10}\cmidrule{12-13}    DGP   & \(n\) & MSE   & Bias  &       & MSE   & Bias  &       & MSE   & Bias  &       & MSE   & Bias \\
    \midrule
\multirow{5}[2]{*}{DGP 1} & 500   & 17.538 & -0.023 &       & 20.095 & -0.009 &       & 19.849 & -0.020 &       & 19.852 & -0.018 \\
      & 1000  & 17.204 & -0.124 &       & 19.118 & -0.129 &       & 19.446 & -0.135 &       & 19.259 & -0.112 \\
      & 2000  & 17.350 & -0.052 &       & 19.123 & -0.050 &       & 19.965 & -0.050 &       & 19.456 & 0.001 \\
      & 4000  & 17.194 & 0.062 &       & 18.189 & 0.052 &       & 19.246 & 0.043 &       & 18.695 & 0.142 \\
      & 8000  & 17.317 & 0.083 &       & 18.153 & 0.105 &       & 19.540 & 0.107 &       & 18.505 & 0.233 \\
\midrule
\multirow{5}[2]{*}{DGP 2} & 500   & 16.160 & -0.028 &       & 16.506 & -0.031 &       & 16.188 & -0.025 &       & 16.278 & -0.025 \\
      & 1000  & 15.856 & -0.149 &       & 15.956 & -0.145 &       & 15.874 & -0.149 &       & 16.040 & -0.153 \\
      & 2000  & 15.556 & -0.046 &       & 15.576 & -0.046 &       & 15.560 & -0.047 &       & 15.683 & -0.048 \\
      & 4000  & 15.868 & 0.020 &       & 15.916 & 0.022 &       & 15.869 & 0.020 &       & 16.205 & 0.020 \\
      & 8000  & 15.960 & 0.110 &       & 15.972 & 0.113 &       & 15.959 & 0.111 &       & 16.053 & 0.114 \\
\midrule
\multirow{5}[2]{*}{DGP 3} & 500   & 13.420 & -0.090 &       & 30.729 & -1.572 &       & 26.697 & -0.067 &       & 24.945 & 2.900 \\
      & 1000  & 13.443 & -0.104 &       & 27.120 & -0.862 &       & 26.943 & -0.080 &       & 30.147 & 3.458 \\
      & 2000  & 13.725 & 0.017 &       & 23.629 & -0.333 &       & 27.433 & -0.006 &       & 30.384 & 3.372 \\
      & 4000  & 13.581 & -0.092 &       & 20.450 & -0.160 &       & 27.035 & -0.097 &       & 23.979 & 2.354 \\
      & 8000  & 13.925 & -0.037 &       & 18.690 & -0.020 &       & 27.818 & 0.015 &       & 19.863 & 1.235 \\
\midrule
\multirow{5}[2]{*}{DGP 4} & 500   & 13.467 & -0.082 &       & 33.787 & -1.886 &       & 27.377 & -0.058 &       & 36.226 & 4.364 \\
      & 1000  & 13.486 & -0.122 &       & 28.984 & -1.029 &       & 27.681 & -0.104 &       & 43.866 & 5.022 \\
      & 2000  & 13.700 & 0.032 &       & 24.703 & -0.361 &       & 28.221 & 0.023 &       & 42.664 & 4.820 \\
      & 4000  & 13.649 & -0.076 &       & 21.122 & -0.157 &       & 27.870 & -0.084 &       & 30.253 & 3.392 \\
      & 8000  & 13.821 & -0.024 &       & 18.879 & -0.007 &       & 28.533 & 0.037 &       & 21.760 & 1.827 \\
    \bottomrule
    \bottomrule
    \end{tabular}
    \label{tab--var20}%
\end{table}

%%% Local Variables:
%%% mode: latex
%%% TeX-master: "../2023_semipar_eff_car"
%%% End:
% LocalWords:  SSRA SPBR ATT

% ! TEX root = ../2023_semipar_eff_car.tex

\section{Conclusion}
\label{sec--conclusion}

In this paper, we have characterized the maximal gains in efficiency from using
baseline covariates beyond strata in RCT's covariate adaptive randomization is
used. To do this, we have established a semiparametric efficiency bound under
CAR for such experiments. This is the minimum estimation variance achievable if
one uses all the relevant information the data have to offer about the parameter
of interest (the ATE here). If baseline covariates are used, we use the
information they contain about the ATE fully through differences in conditional
means, averaged over the support of the covariates. With continuous covariates,
or more covariates present than used in stratification, a strict improvement
over fully saturated regression estimates of the ATE is possible. The efficiency
bound established is shown to be achievable under the same conditions used for
its derivation by using a leave one out Nadaraya-Watson kernel regression
estimator. Simulation evidence presented shows that this estimator can indeed
reach the efficiency bound in large samples. However, such improvements are not
guaranteed in finite samples, and reaching the efficiency bound can be slow
especially in the presence of many covariates due to the curse of
dimensionality.

%%% Local Variables:
%%% mode: latex
%%% TeX-master: "../2023_semipar_eff_car"
%%% End:
% LocalWords:  SSRA SPBR ATT

{\printbibliography}

% ! TEX root = ../2023_semipar_eff_car.tex

\newpage

\appendix

\setcounter{page}{1}

\allowdisplaybreaks

\begin{center}
  \Large{Appendix for``Efficient Semiparametric Estimation of Average
  Treatment Effects Under Covariate Adaptive Randomization''} \\
  \large{Ahnaf Rafi}
\end{center}

\renewcommand{\theasm}{\Alph{asm}}

\begin{asm}
\th\label{asm--coupling-construct}
\(\mathbf{W}^{\ast} = \left\{ W_{i} (s) : i \in \mathbb{N},
s \in \mathbb{N}_{\mathcal{S}} \right\}\) is a sequence of random vectors
defined on \((\Omega, \mathcal{F}, \Pr)\) independent across \((i, s) \in
\mathbb{N} \times \mathbb{N}_{\mathcal{S}}\) such that \(W_{i}^{\prime} (s) =
\left( Y_{i} (0, s), Y_{i} (1, s), Z_{i}^{\prime} (s) \right)\) has the same
marginal distribution as \(W | \mathbb{S} (Z) = s\). Furthermore,
\(\mathbf{W}^{\ast}\) is independent to the collection \(\left\{ \mathbf{W}
\right\} \cup \left\{ \mathbf{A}_{n} : n \in \mathbb{N} \right\}\).
\end{asm}

% \input{appendix/01_main_theorems.tex}
% \input{appendix/02_main_lemmas.tex}
% \input{appendix/03_auxiliary.tex}

%%% Local Variables:
%%% mode: latex
%%% TeX-master: "../2023_semipar_eff_car"
%%% End:
% LocalWords:  SSRA SPBR ATT

% ! TEX root = ../2023_semipar_eff_car.tex

\section{Proofs of theorems in paper}

\subsection{Proof of Theorem \ref{thm--speb-car}}

\begin{proof}[Proof of \th\ref{thm--speb-car}]
By \th\ref{lem--pathwise-differentiability-ate}, the ATE parameter, \(\beta
(\cdot)\) defined by \eqref{eqn--ate-P} is pathwise differentiable with
efficient influence function \(\varphi_{0}\) defined in \eqref{eqn--eif}. This
means that \(\beta (\cdot)\) has a derivative operator \(\dot{\beta} (\cdot)\)
such that for any element of the tangent space in \th\ref{lem--tangent-space-P},
\(h \in \dot{\mathbf{P}}\),
\begin{equation*}
  \dot{\beta} (h) = \mathbb{E}_{P_{0}} \left[ \varphi_{0} (Y, A, Z) h (Y, A, Z)
  \right].
\end{equation*}
The adjoint operator is the unique linear operator \(\dot{\beta}^{\ast} :
\mathbb{R} \to L_{2} \left( P_{0} \right)\) defined by
\begin{equation*}
  c \dot{\beta} (h) = c \mathbb{E}_{P_{0}} \left[ \varphi_{0} (Y, A, Z) h (Y, A,
  Z) \right] = \mathbb{E}_{P_{0}} \left[ \left\{ \dot{\beta}^{\ast} (c) \right\}
  (Y, A, Z) h (Y, A, Z) \right].
\end{equation*}
Clearly, \(\left\{ \dot{\beta}^{\ast} (c) \right\} (y, a, z) = c
\varphi_{0} (y, a, z)\) for every \(c \in \mathbb{R}\). For a given \(h\), let
\(\mathbf{P}_{0 h}\) be any one-dimensional regular parametric submodel of
\(\mathbf{P}\) passing through \(P_{0}\) that has score \(h\) at \(P_{0}\). As
in \th\ref{rem--prod-struct-ssra} and \th\ref{def--parametric-submodel}, this
means that \(\mathbf{P}_{0 h} = \left\{ P_{\theta, h} \in \mathbf{P} : \theta
\in (-2, 2) \right\}\), each \(P_{\theta, h}\) has a density against \(\nu\) in
\eqref{eqn--nu} of the form
\begin{equation*}
  p (y, a, z; \theta, h) = \left[ q_{1} (y, z; \theta, h) \pi (\mathbb{S} (z))
  \right]^{a} \left[ q_{0} (y, z; \theta, h) \cdot \pi (1 - \mathbb{S} (z))
  \right]^{1 - a}.
\end{equation*}
The requirement that \(\mathbf{P}_{0 h}\) passes through \(P_{0}\) will be taken
to mean that \(P_{\theta, h} = P_{0}\) implies \(\theta = 0\) (i.e. \(P_{0}\) is
identified uniquely at \(\theta = 0\)). We set \(\theta_{0} = 0\) since we can
always use a translation to ensure this without changing any conclusions. The
choice \(\Theta = (-2, 2)\) is to ensure \(\Theta\) contains local alternatives
of the form \(\theta_{n} = 1 / \sqrt{n}\) - the conclusions are unaffected by
redefining \(\Theta\) to be an open interval centered at \(0\) and then
re-scaling. By \th\ref{def--qmd}, regularity of \(\mathbf{P}_{0 h}\) means that
the maps \(\theta \mapsto q_{a} (\cdot; \theta, h)\) has a derivative in
quadratic mean, \(D_{a, h}\) so that
\begin{equation*}
  \lim_{\theta \to 0} \theta^{- 2} \int \left(\sqrt{q_{a} (y, z; \theta, h)} -
  \sqrt{q_{a} \left( y, z; 0, h \right)} - D_{a, h} (y, z) \cdot \theta
  \right)^{2} \nu_{a} (\mathrm{d} y, \mathrm{d} z) = 0.
\end{equation*}
The requirement that \(h\) be the score function at \(P_{0}\) corresponds
to the requirement that (by \eqref{eqn--sample-score})
\begin{equation*}
  h (y, a, z) = 2 \frac{D_{a, h} (y, z)}{\sqrt{q_{a} \left( y, z; Q_{0}
  \right)}} \mathbb{I} \left\{ q_{a} \left( y, z; Q_{0} \right) > 0 \right\}.
\end{equation*}
Let \(\mathcal{P}_{h} = \left\{
\mathbf{P}_{n h} : n \in \mathbb{N} \right\}\) be a parametric submodel of
\(\mathcal{P}\) characterized by the maps \(\theta \mapsto q_{a} (\cdot; \theta,
h)\) for \(\theta \in (- 2, 2)\). That is, each element \(P_{n, \theta, h} \in
\mathbf{P}_{n h}\) for \(\theta \in (- 2, 2)\) has a density against the
\(n\)-fold product measure \(\nu^{n}\) of the form
\begin{align*}
  p_{n} \left( \mathbf{y}_{n}, \mathbf{a}_{n}, \mathbf{z}_{n}; \theta, h \right)
  =
  & \ \alpha_{n} \left( \mathbf{a}_{n} | \mathbf{s}_{n} \right) \cdot \prod_{i
    = 1}^{n} q_{1} \left( y_{i}, z_{i}; \theta, h \right)^{a_{i}} q_{0} \left(
    y_{i}, z_{i}; \theta, h \right)^{1 - a_{i}}, \\
  \text{where } \mathbf{s}_{n}^{\prime} =
  & \ \left( \mathbb{S} \left( z_{1} \right), \dots, \mathbb{S} \left( z_{n}
    \right) \right)
\end{align*}
and \(\alpha_{n} (\cdot | \cdot)\) satisfies \th\ref{asm--treat-strat}
\ref{asm--treat-prop}. By \th\ref{lem--parametric-submodel-lan}, every regular
parametric submodel of \(\mathcal{P}\) has the LAN property, and hence
\(\mathcal{P}_{h}\) necessarily exhibits the LAN property. In what follows, set
\(\beta_{n h} = \beta \left( P_{\theta_{n}, h} \right)\) for \(\theta_{n} = 1 /
\sqrt{n}\). Note that a direct implication of
\th\ref{lem--pathwise-differentiability-ate} is that the local parameters
\(\beta_{n h}\) are regular in the sense of
\citet[p. 413]{1996vandervaartWeakConvergenceEmpirical}.

By Theorem 3.11.2 of \citet[p. 414]{1996vandervaartWeakConvergenceEmpirical} if
\(\left\{ \widehat{\beta}_{n} \right\}\) is a sequence of regular estimators for
\(\beta (\cdot)\), then \(\sqrt{n} \left( \widehat{\beta}_{n} - \beta_{0}
\right) \overset{\mathrm{d}}{\to} G + W\) where \(G\) and \(W\) are independent
random variables and
\begin{equation}
  G \sim \mathcal{N} \left( 0, \mathbb{E}_{P_{0}} \left[ \varphi_{0} (Y, A,
  Z)^{2} \right] \right) = \mathcal{N} \left( 0, \mathbb{V}_{\ast} \right).
  \label{eqn--eff-dist-lim}
\end{equation}
This is the usual convolution theorem based justification of
\(\mathbb{V}_{\ast}\) as a semiparametric efficiency bound. Next, let
\(\mathcal{L} : \mathbb{R} \to \mathbb{R}\) be any bowl-shaped loss function,
i.e. \(\mathcal{L}\) is non-negative, satisfies \(\mathcal{L} (y) = \mathcal{L}
(- y)\) and has convex sublevel sets. By Theorem 3.11.5 of
\citet[p. 417]{1996vandervaartWeakConvergenceEmpirical}, for any (not
necessarily regular) estimator sequence \(\left\{ \widehat{\beta}_{n} \right\}\)
for \(\beta (\cdot)\),
\begin{equation*}
  \sup_{\substack{I \subseteq \dot{\mathbf{P}} \\ I \text{ is finite}}}
  \liminf_{n \to \infty} \sup_{h \in I} \mathbb{E}_{n h}
  \left[ \mathcal{L} \left( \sqrt{n} \left( \widehat{\beta}_{n} - \beta_{n h}
  \right) \right) \right] \geq \mathbb{E} \left[ \mathcal{L} (G) \right].
\end{equation*}
Here, \(G\) is as in \eqref{eqn--eff-dist-lim} and \(\mathbb{E}_{n h}\) denotes
expectations taken with respect to local perturbations of \(P_{0, n}\) of the
form \(P_{n, \theta_{n}, h}\) with \(\theta_{n} = 1 / \sqrt{n}\). This is the
justification of \(\mathbb{V}_{\ast}\) as the semiparametric efficiency bound
via the local asymptotic minimax theorem.
\end{proof}

\subsection{Proof of Theorem \ref{thm--efficient-estimation-car}}

For the purposes of the proof, it will be useful to define some additional
notation. For each observation \(i \in \mathbb{N}_{n}\), let the corresponding
stratum-fold label in \eqref{eqn--treat-control-combine-sample-split}
(\th\ref{asm--sample-split}) be
\begin{equation}
  \mathbb{J}_{i} = \sum_{s = 1}^{\mathcal{S}} \sum_{j = 1}^{J} j \cdot
  \mathbb{I} \left\{ \mathbb{S} \left( Z_{i} \right) = s, i \in \mathcal{G}_{n
  j} (s) \right\}.
  \label{eqn--split-indicator}
\end{equation}
Let
\begin{align}
  \widehat{N}_{n} (a, s, j) =
  & \ \sum_{i = 1}^{n} \mathbb{I} \left\{ A_{n i} = a, \mathbb{S} \left( Z_{i}
    \right) = s, \mathbb{J}_{i} = j \right\},
    \label{eqn--strat-treat-fold-size}
  \\
  \widehat{N}_{n} (s, j) =
  & \ \sum_{i = 1}^{n} \mathbb{I} \left\{ \mathbb{S} \left(Z_{i} \right) = s,
    \mathbb{J}_{i} = j \right\}
    \label{eqn--strat-fold-size} \\
  \widehat{N}_{n}^{\ast} (a, s, j) =
    & \ N_{n} (a, s) - \widehat{N}_{n} (a, s, j).
    \label{eqn--strat-fold-estimation-sampe-size}
\end{align}
\(\widehat{N} (a, s, j)\) is the size of the \(j\)\textsuperscript{th} fold for
treatment group \(a\) within stratum \(s\). Similarly, \(\widehat{N}_{n} (s,
j)\) is the size of the \(j\)\textsuperscript{th} fold for
stratum \(s\). Finally, \(\widehat{N}_{n}^{\ast} (a, s, j)\) in
\eqref{eqn--strat-fold-estimation-sampe-size} is the size of the estimation
sample used for estimation of \(m (a, \cdot)\) in fold \(j\) of stratum \(s\).

\begin{proof}[Proof of \th\ref{thm--efficient-estimation-car}]
In \th\ref{thm--efficient-estimation-car}, \eqref{eqn--eif-gives-speb} follows
from \th\ref{lem--car-vector-clt}. To see that the conditions are satisfied,
first decompose \(\varphi_{0}\) as
\begin{align*}
  \varphi_{0} (y, a, z) =
  & \ h (y, a, z) + \xi (\mathbb{S} (z)) \\
  h (y, a, z) =
  & \ \frac{a}{\pi (\mathbb{S} (z))} \left[ y - m_{\ast} (1, z) \right] -
    \frac{1 - a}{1 - \pi (\mathbb{S} (z))} \left[ y - m_{\ast} (0, z) \right] \\
  & + \left( m_{\ast} (1, z) - m_{\ast} (0, z) - \mathbb{E} \left[ m_{\ast} (1,
    Z) - m_{\ast} (0, Z) | \mathbb{S} (Z) = \mathbb{S} (z) \right] \right) \\
  \xi (s) =
  & \ \mathbb{E} \left[ m_{\ast} (1, Z) - m_{\ast} (0, Z) | \mathbb{S} (Z) = s
    \right] - \beta_{0}.
\end{align*}
Note that \(\mathbb{E} \left[ (\varphi_{0} (Y (a), a, Z))^{2} \right] < \infty\)
which immediately implies that both \(h\) and \(\xi\) have finite second
moments. Furthermore, the (conditional and unconditional) mean zero
requirements of \th\ref{lem--car-vector-clt} are satisfied by the Law of
Iterated Expectations. The fact that the corresponding variance terms
\(\mathbb{V}_{h}\) and \(\mathbb{V}_{\xi}\) from \eqref{eqn--car-vector-clt} in
\th\ref{lem--car-vector-clt} satisfy \(\mathbb{V}_{h} + \mathbb{V}_{\xi} =
\mathbb{V}_{\ast}\) follows from the Law of Total Variance:
\begin{align*}
  \mathbb{E} \left[ \left( m_{\ast} (1, z) - m_{\ast} (0, z) - \beta_{0}
  \right)^{2} \right] =
  & \ \mathrm{Var} \left[ m_{\ast} (1, z) - m_{\ast} (0, z) \right] \\
  =
  & \ \mathbb{E} \left[ \mathrm{Var} \left[ m_{\ast} (1, z) - m_{\ast} (0, z)
    \middle| \mathbb{S} (Z) \right] \right] + \mathrm{Var} \left[ \mathbb{E}
    \left[ m_{\ast} (1, z) - m_{\ast} (0, z) \middle| \mathbb{S} (Z) \right]
    \right] \\
  =
  & \
    \mathbb{E} \left[ \mathbb{E} \left[ \left( m_{\ast} (1, z) - m_{\ast} (0, z)
    - \mathbb{E} \left[ m_{\ast} (1, Z) - m_{\ast} (0, Z) \middle| \mathbb{S}
    (Z) \right] \right)^{2} \middle| \mathbb{S} (Z) \right] \right] \\
  & + \mathbb{E} \left[ \left( \mathbb{E} \left[ m_{\ast} (1, Z) - m_{\ast} (0,
    Z) | \mathbb{S} (Z) \right] - \beta_{0} \right)^{2} \right] \\
  =
  & \
    \mathbb{E} \left[ \left( m_{\ast} (1, z) - m_{\ast} (0, z) - \mathbb{E}
    \left[ m_{\ast} (1, Z) - m_{\ast} (0, Z) \middle| \mathbb{S} (Z) \right]
    \right)^{2} \right] \\
  & + \mathbb{E} \left[ \left( \mathbb{E} \left[ m_{\ast} (1, Z) - m_{\ast} (0,
    Z) \middle| \mathbb{S} (Z) \right] - \beta_{0} \right)^{2} \right].
\end{align*}

Next, we show that \eqref{eqn--beta-star-efficient} holds. By
\eqref{eqn--eif-feas-infeas-diff}, it suffices to show that \(\widetilde{R}_{a,
n} = o_{\mathrm{p}} (1)\) and \(R_{a, n} = o_{\mathrm{p}} (1)\) for each \(a \in
\{0, 1\}\). We start with \(\widetilde{R}_{a, n}\). Write
\begin{align*}
  \widetilde{R}_{a, n} =
  & \sum_{s = 1}^{\mathcal{S}} \widetilde{R}_{n} (a, s), \\
  \text{where} \quad \widetilde{R}_{n} (a, s) =
  & \ \left[ \frac{1}{\widehat{\pi}_{n, a} (s)} - \frac{1}{\pi_{a} (s)} \right]
    \cdot \frac{1}{\sqrt{n}} \sum_{i = 1}^{n} \mathbb{I} \left\{ A_{n i} = a,
    \mathbb{S} \left( Z_{i} \right) = s \right\} \left[ Y_{n i} - m_{\ast}
    \left( a, Z_{i} \right) \right].
\end{align*}
Note that \(Y_{i} (a) - m_{\ast} (a, Z_{i})\) has mean
zero conditional on any given stratum label (by the tower property of
conditional expectations) and of course has finite variance. Therefore, applying
\th\ref{lem--car-vector-clt}, we find immediately that for any given \(s \in
\mathbb{N}_{\mathcal{S}}\),
\begin{equation*}
  \frac{1}{\sqrt{n}} \sum_{i = 1}^{n} \mathbb{I} \left\{ A_{n i} = a,
  \mathbb{S} \left( Z_{i} \right) = s \right\} \left[ Y_{n i} - m_{\ast}
  \left( a, Z_{i} \right) \right]
\end{equation*}
has a limit normal distribution. Then, for each \(a \in \{0, 1\}\) and \(s \in
\mathbb{N}_{\mathcal{S}}\),
\begin{equation*}
  \widetilde{R}_{n} (a, s) = \left[ \frac{1}{\widehat{\pi}_{n, a} (s)}
  - \frac{1}{\pi_{a} (s)} \right] \cdot O_{\mathrm{p}} (1) = o_{\mathrm{p}} (1)
  \cdot O_{\mathrm{p}} (1) = o_{\mathrm{p}} (1).
\end{equation*}
Note that the first term is \(o_{\mathrm{p}} (1)\) since \(\widehat{\pi}_{n, a}
(s)\) is either \(\pi_{a} (s)\) (in which case this is trivially true), or
\(N_{n} (s)^{- 1} N_{n} (a, s)\) which is assumed to be consistent for \(\pi_{a}
(s) \in (0, 1)\). The conclusion follows for \(\sqrt{n} \widetilde{R}_{a, n}\)
using Slutsky's Theorem.

Next, we show that \(R_{a, n} = o_{\mathrm{p}} (1)\) for \(R_{a, n}\)
in \eqref{eqn--eif-feas-infeas-diff} for any \(a \in \{0, 1\}\).
To do this, we will use the sample-splitting structure of the estimators
\(\widehat{m}_{n i}\). Inspecting \eqref{eqn--eif-feas-infeas-diff}, write
\(R_{a, n}\) as
\begin{equation}
  \begin{split}
    R_{a, n} =
    & \ \sum_{s = 1}^{\mathcal{S}} \sum_{j = 1}^{J} R_{n} (a, s, j) \\
    \text{where} \quad R_{n} (a, s, j) =
    & \ \frac{1}{\sqrt{n}} \sum_{i = 1}^{n} \frac{\mathbb{I} \left\{ A_{n i} = a
      \right\} - \widehat{\pi}_{n, a} (s)}{\widehat{\pi}_{n, a} (s)} \left[
      \widehat{m}_{n i} \left(a, Z_{i} \right) - m_{\ast} \left( a, Z_{i}
      \right) \right] \mathbb{I} \left\{ \mathbb{S} \left( Z_{i} \right) = s,
      \mathbb{J}_{i} = j \right\}.
  \end{split}
\end{equation}
Since there are finitely many folds within finitely many strata, it is
sufficient to show that \(R_{n} (a, s, j) = o_{\mathrm{p}} (1)\) for each \(s
\in \mathbb{N}_{\mathcal{S}}\), \(j \in \mathbb{N}_{J}\) and \(a \in \{0,
1\}\). Since we are arguing convergence in probability, we can treat each of
these components separately. To that end let \(\varepsilon > 0\) be given, and
fix \(s \in \mathbb{N}_{\mathcal{S}}\) \(a \in \{0, 1\}\) and \(j \in
\mathbb{N}_{\mathcal{J}}\). Our argument will take the following form. Suppose
\(R_{n}^{\ast} (a, s, j)\) is an auxiliary random variable with the same
distribution as \(R_{n} (a, s, j)\). Furthermore, suppose that we have an
arbitrary sequence of conditioning sets \(\mathcal{F}_{n} (a, s, j)\)
constructed from observable data (formally, each \(\mathcal{F}_{n} (a, s, j)\)
is a sub \(\sigma\)-algebra of that generated by \(\mathbf{X}_{n}\)). Then
\begin{equation}
  \Pr \left( \left| R_{n} (a, s, j) \right| > \varepsilon \right) = \Pr \left(
  \left| R_{n}^{\ast} (a, s, j) \right| > \varepsilon \right) = \mathbb{E}
  \left[ \Pr \left( \left| R_{n}^{\ast} (a, s, j) \right| > \varepsilon \middle|
  \mathcal{F}_{n} (a, s, j) \right) \right]
  \label{eqn--coupling-and-lie}
\end{equation}
where the last equality is due to the Law of Iterated Expectations. Furthermore,
is (by definition) a bounded random variable since it takes values in \([0,
1]\). By the extension of the Dominated Convergence Theorem for convergence in
probability, it therefore suffices to argue that
\begin{equation}
  \Pr \left( \left| R_{n}^{\ast} (a, s, j) \right| > \varepsilon \middle|
  \mathcal{F}_{n} (a, s, j) \right) \overset{\mathrm{p}}{\to} 0
  \label{eqn--sufficient-conditional-convergence-in-prob}
\end{equation}
By Chebychev's Inequality, for an appropriately chosen \(R_{n}^{\ast} (a, s,
j)\) and \(\mathcal{F}_{n} (a, s, j)\), it is sufficient to argue that
\begin{equation}
  \mathbb{E} \left[ \left| R_{n}^{\ast} (a, s, j) \right|^{2} \middle|
  \mathcal{F}_{n} (a, s, j) \right] \overset{\mathrm{p}}{\to} 0.
  \label{eqn--sufficient-conditional-convergence-in-prob-chebychev}
\end{equation}
We now construct \(R_{n}^{\ast} (a, s, j)\) and \(\mathcal{F}_{n}\) for the
above argument.

Let \(\mathbf{W}_{\ast \ast} = \left\{ Y_{i}^{\ast} (1, s), Y_{i}^{\ast} (0, s),
Z_{i}^{\ast} (s) : (i, s) \in \mathbb{N} \times \mathbb{N}_{\mathcal{S}}
\right\}\) be an independent copy of \(\mathbf{W}_{\ast}\) in
\th\ref{asm--coupling-construct} satisfying the following
\begin{align}
  \mathbf{W}_{\ast \ast} \overset{\mathrm{d}}{=}
  & \ \mathbf{W}_{\ast} \\
  \mathbf{W}_{\ast \ast} \indep
  & \ \left( \mathbf{W}, \mathbf{W}_{\ast}, \left\{ \mathbf{A}_{n} : n \in
    \mathbb{N} \right\}, \mathbf{U}_{n} \right).
\end{align}
We will use \(\mathbf{W}_{\ast}\) to proxy for the evaluation points \(Z_{i}\)
in the estimates of \(\widehat{m}_{n i}\) and \(\mathbf{W}_{\ast \ast}\) to
proxy for the estimation sample used for \(\widehat{m}_{n i}\).
Define the following.
\begin{equation}
  \begin{split}
   R_{n}^{\ast} (a, s, j) =
    & \ \frac{1}{\sqrt{n}} \left\{ \sum_{i = 1}^{\widehat{N}_{n} (a, s, j)}
      \frac{1 - \widehat{\pi}_{n, a} (s)}{\widehat{\pi}_{n, a} (s)} \xi_{n i}
      (a, s, j) - \sum_{i =
      \widehat{N}_{n} (a, s, j) + 1}^{\widehat{N}_{n} (s, j)}
      \xi_{n i} (a, s, j) \right\}
    \\
    \text{where } \xi_{n i} (a, s j) =
    & \ \widehat{m} \left( \widehat{N}_{n}^{\ast} (a, s, j) , a, Z_{i} (s);
      \mathcal{E}_{n} \right) - m_{\ast} (Z_{i} (s)), \\
    \mathcal{E}_{n} (a, s, j) =
    & \ \left\{ Y_{\iota}^{\ast} (a, s), Z_{\iota}^{\ast} (s) : \iota = 1,
      \dots, \widehat{N}_{n}^{\ast} (a, s, j) \right\} \\
  \end{split}
\end{equation}
In the above, \(\mathcal{E}_{n} (a, s, j)\) is an independent copy of the
estimation sample used for estimation of \(m_{\ast} (a, \cdot)\) in fold \(j\)
and \(\widehat{N}_{n}^{\ast} (a, s, j)\) in
\eqref{eqn--strat-fold-estimation-sampe-size}
is the number of observations in
\(\mathcal{E}_{n} (a, s, j)\). Note that the sample splits \(\left\{
\mathcal{G}_{n} (s, j) : s \in \mathbb{N}_{\mathcal{S}}, j \in \mathbb{N}_{J}
\right\}\) are independent across
strata and depend only on stratum sizes and the sizes of the treatment groups
within a given stratum. Furthermore, beyond the stratum treatment group sizes,
the particular fold that an individual observation falls into is determined by
the exogenous random variables \(U_{n} (1, s), U_{n} (0, s)\) which are
independent to the overall sample. Furthermore, for any individual observation
in the \(j\)\textsuperscript{th} fold within stratum \(s\), only observations
within that stratum, but outside that fold are used in estimation of
\(\widehat{m}_{n i}\). By repeated applications of \th\ref{lem--coupling-lemma},
it follows that in terms of comparing marginal distributions,
\begin{align*}
  \left. R_{n} (a, s, j) \middle| \mathbf{A}_{n}, \mathbf{S}_{n}, \mathbf{U}_{n}
  \right. \overset{\mathrm{d}}{=}
  & \ \left. R_{n}^{\ast} (a, s, j) \middle|
  \mathbf{A}_{n}, \mathbf{S}_{n}, \mathbf{U}_{n} \right. \\
  \implies R_{n} (a, s, j) \overset{\mathrm{d}}{=}
  & \ R_{n}^{\ast} (a, s, j).
\end{align*}
Let \(\mathbf{N}_{n} (a, s, j) (a, s, j) = \left( N_{n} (s), N_{n} (a, s),
\widehat{N}_{n} (a, s, j) \right)\). In
\eqref{eqn--sufficient-conditional-convergence-in-prob} and
\eqref{eqn--sufficient-conditional-convergence-in-prob-chebychev},
\(\mathcal{F}_{n} (a, s, j)\) will be the \(\sigma\)-algebra generated by
\(\mathbb{N}_{n} (a, s, j)\). Since \(\left( \mathbf{W}_{\ast}, \mathbf{W}_{\ast
\ast} \right) \indep \mathbf{N}_{n} (a, s, j)\), taking
conditional expectations,
\begin{align*}
  \mathbb{E} \left[ R_{n}^{\ast} (a, s, j) \middle| \mathbf{N}_{n} (a, s, j),
  \mathcal{E}_{n} (a, s, j) \right] =
  & \ \frac{1}{\sqrt{n}} \sum_{i = 1}^{\widehat{N}_{n} (a, s, j)}
    \frac{1 - \widehat{\pi}_{n, a} (s)}{\widehat{\pi}_{n, a} (s)} \mathbb{E}
    \left[ \xi_{n i} (a, s, j) \middle| \mathbf{N}_{n} (a, s, j),
    \mathcal{E}_{n} (a, s, j) \right] \\
  & - \frac{1}{\sqrt{n}} \sum_{i = \widehat{N}_{n} (a, s, j) +
    1}^{\widehat{N}_{n} (s, j)} \mathbb{E} \left[ \xi_{n i} (a, s, j) \middle|
    \mathbf{N}_{n} (a, s, j), \mathcal{E}_{n} (a, s, j) \right] \\
  =
  & \ \frac{1}{\sqrt{n}} \left\{ \frac{1 - \widehat{\pi}_{n, a}
    (s)}{\widehat{\pi}_{n, a} (s)} \widehat{N}_{n} (a, s, j) - \widehat{N}_{n}
    (s, j) - \widehat{N}_{n} (a, s, j) \right\} \\
  & \ \times \mathbb{E} \left[ \xi_{n 1} (a, s, j) \middle|
    \mathbf{N}_{n} (a, s, j), \mathcal{E}_{n} (a, s, j) \right]
\end{align*}
Hence, by Jensen's inequality,
\begin{equation}
  \begin{split}
    \left| \mathbb{E} \left[ R_{n}^{\ast} (a, s, j) \middle| \mathbf{N}_{n} (a,
    s, j), \mathcal{E}_{n} (a, s, j) \right] \right| \leq
    & \ \frac{1}{\sqrt{n}} \left\{ \frac{1}{\widehat{\pi}_{n, a} (s)}
      \widehat{N}_{n} (a, s, j) - \widehat{N}_{n} (s, j) \right\} \\
    & \times \mathbb{E} \left[ \xi_{n 1} (a, s, j)^{2} \middle| \mathcal{E}_{n}
      (a, s, j) \right]^{\frac{1}{2}}.
  \end{split}
  \label{eqn--ERasj-cond-En}
\end{equation}
Furthermore, conditional on \(\mathbf{N}_{n} (a, s, j), \mathcal{E}_{n} (a, s,
j)\), \(\xi_{n i} (a, s, j)\) are independent and identically distributed across
\(i = 1, \dots, \widehat{N}_{n} (s, j)\).
\begin{align*}
  \mathrm{Var} \left[ R_{n}^{\ast} (a, s, j) \middle| \mathbf{N}_{n} (a, s, j),
  \mathcal{E}_{n} (a, s, j) \right] =
  & \ \frac{1}{n} \sum_{i = 1}^{\widehat{N}_{n} (a, s, j)} \left( \frac{1 -
    \widehat{\pi}_{n, a} (s)}{\widehat{\pi}_{n, a} (s)} \right)^{2} \mathrm{Var}
    \left[ \xi_{n i} (a, s, j) \middle| \mathbf{N}_{n} (a, s, j),
    \mathcal{E}_{n} \right] \\
  & + \frac{1}{n} \sum_{i = \widehat{N}_{n} (a, s, j) + 1}^{\widehat{N}_{n} (s,
    j)} \mathrm{Var} \left[ \xi_{n i} (a, s, j) \middle| \mathbf{N}_{n} (a, s,
    j), \mathcal{E}_{n} \right] \\
  =
  & \ \frac{1}{n} \left[ \left( \frac{1 - \widehat{\pi}_{n, a}
    (s)}{\widehat{\pi}_{n, a} (s)} \right)^{2} \widehat{N}_{n} (a, s, j) +
    \widehat{N}_{n} (s, j) - \widehat{N}_{n} (a, s, j) \right] \\
  & \times \mathrm{Var} \left[ \xi_{n 1} (a, s, j) \middle| \mathbf{N}_{n} (a,
    s, j), \mathcal{E}_{n} \right] \\
  =
  & \ \frac{1}{n} \left[ \frac{1 - 2 \widehat{\pi}_{n, a} (s)}{\widehat{\pi}_{n,
    a} (s)} \left( \frac{1 - \widehat{\pi}_{n, a} (s)}{\widehat{\pi}_{n, a}
    (s)^{2}} \right)^{2} \widehat{N}_{n} (a, s, j) + \widehat{N}_{n} (s, j)
    \right] \\
  & \times \mathrm{Var} \left[ \xi_{n 1} (a, s, j) \middle| \mathbf{N}_{n} (a,
    s, j), \mathcal{E}_{n} (a, s, j) \right]
\end{align*}
Since the variance is always bounded above by the second moment,
\begin{equation}
  \begin{split}
    \mathrm{Var} \left[ R_{n}^{\ast} (a, s, j) \middle| \mathbf{N}_{n} (a, s,
    j), \mathcal{E}_{n} (a, s, j) \right] \leq
    & \frac{1}{n} \left[ \frac{1 - 2 \widehat{\pi}_{n, a} (s)}{\widehat{\pi}_{n,
      a} (s)^{2}} \widehat{N}_{n} (a, s, j) + \widehat{N}_{n} (s, j) \right] \\
    & \times \mathbb{E} \left[ \xi_{n 1} (a, s, j)^{2} \middle| \mathbf{N}_{n}
      (a, s, j), \mathcal{E}_{n} (a, s, j) \right]
  \end{split}
  \label{eqn--VRasj-cond-En}
\end{equation}
Combining the two upper bounds \eqref{eqn--ERasj-cond-En} and
\eqref{eqn--VRasj-cond-En}, we get
\begin{align*}
  \mathbb{E} \left[ R_{n}^{\ast} (a, s, j)^{2} \middle| \mathbf{N}_{n} (a, s,
  j), \mathcal{E}_{n} (a, s, j) \right] =
  & \ \mathbb{E} \left[ R_{n}^{\ast} (a, s, j)^{2} \middle| \mathbf{N}_{n} (a,
    s, j), \mathcal{E}_{n} (a, s, j) \right] \\
  & + \mathrm{Var} \left[ R_{n}^{\ast} (a, s, j) \middle| \mathbf{N}_{n} (a, s,
    j), \mathcal{E}_{n} (a, s, j) \right] \\
  \leq
  & \frac{1}{n} \left[
    \begin{array}{l}
      \left\{ \frac{1}{\widehat{\pi}_{n, a} (s)} \widehat{N}_{n} (a, s, j) -
      \widehat{N}_{n} (s, j) \right\}^{2} \\
      + \frac{1 - 2 \widehat{\pi}_{n, a} (s)}{\widehat{\pi}_{n, a} (s)^{2}}
      \widehat{N}_{n} (a, s, j) + \widehat{N}_{n} (s, j)
    \end{array}
    \right] \\
  & \times \mathbb{E} \left[ \xi_{n 1} (a, s, j)^{2} \middle| \mathbf{N}_{n} (a,
    s, j), \mathcal{E}_{n} (a, s, j) \right].
\end{align*}
By monotonicity of expectations,
\begin{align*}
  \mathbb{E} \left[ R_{n}^{\ast} (a, s, j)^{2} \middle| \mathbf{N}_{n} (a, s,
  j) \right] =
  & \ \mathbb{E} \left[ \mathbb{E} \left[ R_{n}^{\ast} (a, s, j)^{2} \middle|
    \mathbf{N}_{n} (a, s, j), \mathcal{E}_{n} (a, s, j) \right] \middle|
    \mathbf{N}_{n} (a, s, j) \right] \\
  \leq
  & \ \frac{1}{n} \left[
    \begin{array}{l}
      \left\{ \frac{1}{\widehat{\pi}_{n, a} (s)} \widehat{N}_{n} (a, s, j) -
      \widehat{N}_{n} (s, j) \right\}^{2} \\
      + \frac{1 - 2 \widehat{\pi}_{n, a} (s)}{\widehat{\pi}_{n, a} (s)^{2}}
      \widehat{N}_{n} (a, s, j) + \widehat{N}_{n} (s, j)
    \end{array}
    \right] \\
  & \times \mathbb{E} \left[ \mathbb{E} \left[ \xi_{n 1} (a, s, j)^{2} \middle|
    \mathbf{N}_{n} (a, s, j), \mathcal{E}_{n} (a, s, j) \right] \middle|
    \mathbf{N}_{n} (a, s, j) \right] \\
  = & \ \frac{1}{n} \left[
    \begin{array}{l}
      \left\{ \frac{1}{\widehat{\pi}_{n, a} (s)} \widehat{N}_{n} (a, s, j) -
      \widehat{N}_{n} (s, j) \right\}^{2} \\
      + \frac{1 - 2 \widehat{\pi}_{n, a} (s)}{\widehat{\pi}_{n, a} (s)^{2}}
      \widehat{N}_{n} (a, s, j) + \widehat{N}_{n} (s, j)
    \end{array}
    \right] \\
  & \times \mathbb{E} \left[ \xi_{n 1} (a, s, j)^{2} \middle| \mathbf{N}_{n} (a,
    s, j) \right] \\
  =
  & \ \left[
    \begin{array}{l}
      \frac{\widehat{N}_{n} (s, j)^{2}}{n^{2} \widehat{\pi}_{n, a} (s)^{2}}
      \left\{ \sqrt{n} \left[ \frac{\widehat{N}_{n} (a, s, j)}{\widehat{N}_{n}
      (s, j)} - \pi_{a} (s) \right] + \sqrt{n} \left[ \pi_{a} (s) -
      \widehat{\pi}_{n, a} (s) \right] \right\}^{2} \\
      + \frac{1 - 2 \widehat{\pi}_{n, a} (s)}{\widehat{\pi}_{n, a} (s)^{2}}
      \frac{\widehat{N}_{n} (a, s, j)}{n} + \frac{\widehat{N}_{n} (s, j)}{n}
    \end{array}
    \right] \\
  & \times \mathbb{E} \left[ \xi_{n 1} (a, s, j)^{2} \middle| \mathbf{N}_{n} (a,
    s, j) \right]
\end{align*}
Now, we have by \th\ref{lem--sample-split-tends-correct} and our maintained
hypotheses about \(\widehat{\pi}_{n, a} (s)\), we have
\begin{align*}
  \left[
    \begin{array}{l}
      \frac{\widehat{N}_{n} (s, j)^{2}}{n^{2} \widehat{\pi}_{n, a} (s)^{2}}
      \left\{ \sqrt{n} \left[ \frac{\widehat{N}_{n} (a, s, j)}{\widehat{N}_{n}
      (s, j)} - \pi_{a} (s) \right] + \sqrt{n} \left[ \pi_{a} (s) -
      \widehat{\pi}_{n, a} (s) \right] \right\}^{2} \\
      + \frac{1 - 2 \widehat{\pi}_{n, a} (s)}{\widehat{\pi}_{n, a} (s)^{2}}
      \frac{\widehat{N}_{n} (a, s, j)}{n} + \frac{\widehat{N}_{n} (s, j)}{n}
    \end{array}
    \right] =
  & \ \left[
    \begin{array}{l}
      O_{\mathrm{p}} (1) \cdot \left\{ O_{\mathrm{p}} (1) + O_{\mathrm{p}} (1)
      \right\} \\
      + O_{\mathrm{p}} (1) \cdot O_{\mathrm{p}} (1) + O_{\mathrm{p}} (1)
    \end{array}
    \right] \\
  = & \ O_{\mathrm{p}} (1)
\end{align*}
so that
\begin{equation}
  \mathbb{E} \left[ R_{n}^{\ast} (a, s, j)^{2} \middle| \mathbf{N}_{n} (a, s,
  j) \right] \leq O_{\mathrm{p}} (1) \cdot \mathbb{E} \left[ \xi_{n 1} (a, s,
  j)^{2} \middle| \mathbf{N}_{n} (a, s, j) \right].
  \label{eqn--Rnasj2-L2-bound}
\end{equation}
Then, it remains to be shown that \(\mathbb{E} \left[ \xi_{n 1} (a, s, j)^{2}
\middle| \mathbf{N}_{n} (a, s, j) \right] = o_{\mathrm{p}} (1)\). Note that
\(\widehat{N}_{n}^{\ast} (s, j) \overset{\mathrm{p}}{\to} \infty\) in
the sense of \th\ref{def--stochastic-divergence} by
\th\ref{lem--sample-split-tends-correct} and
\th\ref{lem--ratio-conv-implies-div}.
\begin{equation*}
  \zeta_{n} (a, s) = \widehat{m} \left( n , a, Z_{1} (s); Y_{1}^{\ast} (a, s),
  Z_{1}^{\ast} (s), \dots, Y_{n}^{\ast} (a, s), Z_{n}^{\ast} (s) \right) -
  m_{\ast} \left( a, Z_{1} (s) \right).
\end{equation*}
By \th\ref{asm--mhat-L2},

\(\lim_{n \to \infty} \mathbb{E} \left[ \zeta_{n} (a,
s)^{2} \right] = 0\), i.e. \(\zeta_{n} \overset{\mathrm{L}_{2}}{\to} 0\). Note
that we can write
\begin{equation*}
  \xi_{n 1} (a, s, j) = \zeta_{\widehat{N}_{n}^{\ast} (a, s, j)} (a, s).
\end{equation*}
By \th\ref{lem--random-subseq-conv-Lr}, \(\mathbb{E} \left[ \xi_{n 1} (a, s,
j)^{2} | \mathbf{N}_{n} (a, s, j) \right] = o_{\mathrm{p}} (1)\). Therefore,
by \eqref{eqn--Rnasj2-L2-bound},
\begin{equation}
  \mathbb{E} \left[ R_{n}^{\ast} (a, s, j)^{2} \middle| \mathbf{N}_{n} (a, s,
  j) \right] \leq O_{\mathrm{p}} (1) \cdot \mathbb{E} \left[ \xi_{n 1} (a, s,
  j)^{2} \middle| \mathbf{N}_{n} (a, s, j) \right] = O_{\mathrm{p}} (1) \cdot
  o_{\mathrm{p}} (1) = o_{\mathrm{p}} (1).
  \label{eqn--Rnasj2-L2-op1}
\end{equation}
which implies by Chebychev's Inequality
\begin{equation*}
  \Pr \left( \left| R_{n}^{\ast} (a, s, j) \right| > \varepsilon \middle|
  \mathbf{N}_{n} (a, s, j) \right) \leq \varepsilon^{- 2} \mathbb{E} \left[
  R_{n}^{\ast} (a, s, j)^{2} \middle| \mathbf{N}_{n} (a, s, j) \right] =
  o_{\mathrm{p}} (1).
\end{equation*}
By the Law of Iterated Expectations and the extension of the Dominated
Convergence Theorem to convergence in probability, since the conditional
probability takes values in the unit interval \([0, 1]\),
\begin{equation*}
  \lim_{n \to \infty} \Pr \left( \left| R_{n}^{\ast} (a, s, j) \right| >
  \varepsilon \right) = \lim_{n \to \infty} \mathbb{E} \left[ \Pr \left( \left|
  R_{n}^{\ast} (a, s, j) \right| > \varepsilon \middle| \mathbf{N}_{n} (a, s, j)
  \right) \right] = 0.
\end{equation*}
This concludes the proof.
\end{proof}

%%% Local Variables:
%%% mode: latex
%%% TeX-master: "../2023_semipar_eff_car"
%%% End:
% LocalWords:  SSRA SPBR ATT
% ! TEX root = ../2023_semipar_eff_car.tex

\section{Proofs of additional lemmas in the paper}

\subsection{Proof of Lemma \ref{lem--prod-struct}}

\begin{proof}[Proof of \th\ref{lem--prod-struct}]
We start by noting that \eqref{eqn--data-joint-dist-prod} implies
\eqref{eqn--data-joint-density-prod}, and so we prove
\eqref{eqn--data-joint-dist-prod}. To that end, for a given \(Q \in
\mathbf{Q}\), by first conditioning \(\mathbf{S}_{n} = \mathbf{s}_{n}\), we have
\begin{equation*}
  P_{n} \left( E_{n} \left( \mathbf{y}_{n}, \mathbf{a}_{n}, \mathbf{z}_{n},
  \mathbf{s}_{n} \right); Q \right) = P_{n} \left( \mathbf{Y}_{n} \leq
  \mathbf{y}_{n}, \mathbf{A}_{n} = \mathbf{a}_{n}, \mathbf{Z}_{n} \leq
  \mathbf{z}_{n} \middle| \mathbf{S}_{n} = \mathbf{s}_{n}; Q \right) P_{n}
  \left( \mathbf{S}_{n} = \mathbf{s}_{n}; Q \right).
\end{equation*}
The equation for the observed outcomes, \eqref{eqn--obsoutproc}, implies that
after writing \(\mathbf{a}_{n} = \left( a_{1}, \dots, a_{n} \right)\),
\begin{equation*}
  P_{n} \left( \mathbf{Y}_{n} \leq \mathbf{y}_{n}, \mathbf{A}_{n} =
  \mathbf{a}_{n}, \mathbf{Z}_{n} \leq \mathbf{z}_{n} \middle| \mathbf{S}_{n} =
  \mathbf{s}_{n}; Q \right) = P_{n} \left( Y_{i} \left( a_{i} \right) \leq y_{i}
  \ \forall i \in \mathbb{N}_{n}, \mathbf{A}_{n} = \mathbf{a}_{n},
  \mathbf{Z}_{n} \leq \mathbf{z}_{n} \middle| \mathbf{S}_{n} = \mathbf{s}_{n}; Q
  \right).
\end{equation*}
Under \th\ref{asm--treat-strat}\ref{asm--treat-exog}, it follows that
\begin{align*}
  P_{n} \left( Y_{i} \left( a_{i} \right) \leq y_{i} \ \forall i \in
  \mathbb{N}_{n}, \mathbf{A}_{n} = \mathbf{a}_{n}, \mathbf{Z}_{n} \leq
  \mathbf{z}_{n} \middle| \mathbf{S}_{n} = \mathbf{s}_{n}; Q \right)
  =
  & \ P_{n} \left( Y_{i} \left( a_{i} \right) \leq y_{i} \ \forall i \in
    \mathbb{N}_{n}, \mathbf{Z}_{n} \leq \mathbf{z}_{n} \middle| \mathbf{S}_{n} =
    \mathbf{s}_{n}; Q \right) \\
  & \times P_{n} \left( \mathbf{A}_{n} = \mathbf{a}_{n} \middle| \mathbf{S}_{n}
    = \mathbf{s}_{n}; Q \right).
\end{align*}
The second term in the product is \(\alpha_{n} \left( \mathbf{a}_{n} \middle|
\mathbf{s}_{n} \right)\) by \th\ref{asm--treat-strat}\ref{asm--cov-adapt}.
Combining the above two displayed equations alongside the fact that \(\left\{
W_{i} \right\}_{i = 1}^{n}\) are i.i.d, we have
\begin{equation}
  \begin{split}
    & P_{n} \left( \mathbf{Y}_{n} \leq \mathbf{y}_{n}, \mathbf{A}_{n} =
      \mathbf{a}_{n}, \mathbf{Z}_{n} \leq \mathbf{z}_{n} \middle| \mathbf{S}_{n}
      = \mathbf{s}_{n}; Q \right) \\
    & = \ \alpha_{n} \left( \mathbf{a}_{n} \middle| \mathbf{s}_{n} \right)
      \times \prod_{s = 1}^{\mathcal{S}} \prod_{i \in \mathbb{N}_{n} : s_{i} =
      s} \left\{
      \begin{array}{l}
        Q \left( Y_{i} (1) \leq y_{i}, Z_{i} \leq z_{i} \middle| S_{i} = s
        \right)^{a_{i}} \\
        \times Q \left( Y_{i} (0) \leq y_{i}, Z_{i} \leq z_{i} \middle| S_{i} =
        s \right)^{1 - a_{i}} \\
      \end{array}
      \right\}.
  \end{split}
  \label{eqn--ycondzs-prod}
\end{equation}
Finally, the fact that \(\left\{ W_{i} \right\}_{i = 1}^{n}\) are i.i.d. in
conjunction with the fact that \(a_{i} + \left( 1 - a_{i} \right) = 1\) for each
\(i \in \mathbb{N}_{n}\) implies that
\begin{equation}
  P_{n} \left( \mathbf{S}_{n} = \mathbf{s}_{n}; Q \right) = \prod_{s =
  1}^{\mathcal{S}} \prod_{i \in \mathbb{N}_{n} : s_{i} = s} Q \left( S_{i} = s
  \right) = \prod_{s = 1}^{\mathcal{S}} \prod_{i \in \mathbb{N}_{n} : s_{i} = s}
  Q \left( S_{i} = s \right)^{a_{i}} Q \left( S_{i} = s \right)^{1 - a_{i}}.
  \label{eqn--zsprod}
\end{equation}
Combining \eqref{eqn--ycondzs-prod} and \eqref{eqn--zsprod} yields
\eqref{eqn--data-joint-dist-prod}.
\end{proof}

\subsection{Proof of Lemma \ref{lem--parametric-submodel-lan}}

Before moving to the proof of \th\ref{lem--parametric-submodel-lan} we will make
a few observations that will be useful during the course of the proof. For \(Q
\in \mathbf{Q}\) as defined in \th\ref{asm--Q}, let
\begin{equation}
  \phi (s; Q) = Q (\mathbb{S} (Z) = s) \quad \text{and} \quad \widetilde{q}_{a}
  (y, z | s; Q) = \frac{q_{a} \left( y, z; Q \right)}{\phi (s; Q)} \mathbb{I}
  \left\{ \mathbb{S} (z) = s \right\}
  \label{eqn--S-marginal-q-cond-S}
\end{equation}
for each \(s \in \mathbb{N}_{\mathcal{S}}\). \(\phi (\cdot; Q)\) is the marginal
probability mass function of \(\mathbb{S} (Z)\) under \(Q\) and
\(\widetilde{q}_{a} (\cdot | s; Q)\) is the conditional density of \((Y (a),
Z)\) given \(\mathbb{S} (Z) = s\) under \(Q\). Note that of course, \(\phi\)
must be common across \(a \in \{0, 1\}\).

Now, in any regular parametric submodel \(\mathcal{P}^{0}\), let the analogous
versions of \eqref{eqn--S-marginal-q-cond-S} be
\begin{equation}
  \phi (s; \theta) = \int \mathbb{I} \{\mathbb{S} (z) = s\} q_{a} (y, z; \theta)
  \nu_{a} (\mathrm{d} y, \mathrm{d} z) \quad \text{and} \quad
  \widetilde{q}_{a} (y, z | s; \theta) = \frac{q_{a} \left( y, z; \theta
  \right)}{\phi (s; \theta)} \mathbb{I} \left\{ \mathbb{S} (z) = s \right\}
  \label{eqn--S-marginal-q-cond-S-rpsm}
\end{equation}
The map \(\theta \mapsto \sqrt{\phi (\cdot; \theta)}\) is differentiable in
quadratic mean at \(\theta_{0}\) due to Proposition A.5.5 of
\citet[p. 461]{1998bickelEfficientAdaptiveEstimation}, with the score function
\begin{equation}
  r (s) = \mathbb{E}_{P_{0}} \left[ \dot{\ell}_{a} (Y (a), Z) \middle|
  \mathbb{S} (Z) = s \right] \quad \text{for any choice of } a \in \{0, 1\}.
  \label{eqn--score-in-strata}
\end{equation}
Of course \(r (\cdot)\) must have mean zero, i.e.
\begin{equation}
  \sum_{s = 1}^{\mathcal{S}} r (s) \phi \left( s; Q_{0} \right) =
  \mathbf{0}_{d}.
\end{equation}
The associated information matrix will be denoted
\begin{equation*}
  \mathcal{J} = \mathbb{E}_{P_{0}} \left[ r (\mathbb{S} (Z)) r (\mathbb{S}
  (Z))^{\prime} \right].
  \label{eqn--info-in-strata}
\end{equation*}

Define the ``implied conditional score'' for \(\widetilde{q}_{a} (y, z | s;
\theta)\) by
\begin{equation}
  \lambda_{a} (y, z | s) = \left[ \dot{\ell}_{a} (y, z) - r (s) \right]
  \mathbb{I} (\mathbb{S} (z) = s).
  \label{eqn--conditional-score-in-strata}
\end{equation}
By the Law of Iterated Expectations,
\begin{equation}
  \mathbb{E}_{Q_{0}} \left[ \lambda_{a} (Y (a), Z | \mathbb{S} (Z)) \middle|
  \mathbb{S} (Z) = s \right] = \mathbf{0}_{d} \quad \text{for all } s \in
  \mathbb{N}_{\mathcal{S}}.
  \label{eqn--strat-score-cond-mean-zero}
\end{equation}
Let the ``implied conditional information matrix'' be
\begin{equation}
  \mathcal{J} (a, s) = \mathbb{E}_{Q_{0}} \left[ \lambda_{a} (Y (a), Z |
  \mathbb{S} (Z)) \lambda_{a} (Y (a), Z | \mathbb{S} (Z))^{\prime} \middle|
  \mathbb{S} (Z) = s \right].
  \label{eqn--conditional-info-in-strata}
\end{equation}
Note that by \eqref{eqn--strat-score-cond-mean-zero},
\begin{equation*}
  \mathbb{E}_{Q_{0}} \left[ \lambda_{a} (Y (a), Z | \mathbb{S} (Z)) r
  (\mathbb{S} (Z))^{\prime} | \mathbb{S} (Z) = s \right] = \mathbf{0}_{d \times
  d} \quad \text{for all } s \in \mathbb{N}_{\mathcal{S}}.
\end{equation*}
An implication of this is that we can write \(\mathcal{I}\) in
\eqref{eqn--strata-info-mat} as
\begin{equation}
  \mathcal{I} = \mathcal{J} + \sum_{s = 1}^{\mathcal{S}} Q_{0} (\mathbb{S} (Z) =
  s) \sum_{a \in \{0, 1\}}^{n} \pi_{a} (s) \mathcal{J} (a, s)
  \label{eqn--strata-info-mat-breakdown}
\end{equation}
where \(\pi_{a} (s) = \pi (s)^{a} (1 - \pi (s))^{1 - a}\). It is intuitive but
not obvious that the maps \(\theta \mapsto \sqrt{q_{a} \left( \cdot | \cdot;
\theta \right)}\) inherit the differentiability in quadratic mean property from
\(\theta \mapsto \sqrt{q_{a} (\cdot; \theta)}\) and \(\theta \mapsto \sqrt{\phi
(\cdot \theta)}\). A statement to that effect is proven in
\th\ref{lem--qmd-strata-cond-density} via a version of the quotient rule for
differentiability in quadratic mean. The resulting score and information matrix
are as in \eqref{eqn--conditional-score-in-strata} and
\eqref{eqn--conditional-info-in-strata} respectively. We now proceed with the
proof of \th\ref{lem--parametric-submodel-lan}.

\begin{proof}[Proof of \th\ref{lem--parametric-submodel-lan}]
For \th\ref{lem--parametric-submodel-lan} \ref{lem--parametric-submodel-lan-wc},
denote
\begin{equation}
  \lambda (y, a, z | s) = \left[ a \lambda_{1} (y, z | s) + (1 - a) \lambda_{0}
  (y, z | s) \right] \mathbb{I} (\mathbb{S} (z) = s).
  \label{eqn--conditional-score-in-strata1}
\end{equation}
By \eqref{eqn--sample-score}, \eqref{eqn--score-in-strata} and
\eqref{eqn--conditional-score-in-strata},
\begin{align*}
  \dot{\ell} (y, a, z) =
  & \ \sum_{s = 1}^{\mathcal{S}} \left[ a \lambda_{1} (y, z | s) + (1 - a)
    \lambda_{0} (y, z | s) + r (s) \right] \mathbb{I} (\mathbb{S} (z) = s) \\
  =
  & \ a \lambda_{1} (y, z | \mathbb{S} (z)) + (1 - a) \lambda_{0} (y, z |
    \mathbb{S} (z)) + r (\mathbb{S} (z)).
\end{align*}
By \eqref{eqn--sample-score} again,
\begin{equation*}
  \frac{1}{\sqrt{n}} \dot{\ell}_{n} = \frac{1}{\sqrt{n}} \sum_{i = 1}^{n} r
  \left( \mathbb{S} \left( Z_{i} \right)\right) + \frac{1}{\sqrt{n}} \sum_{i =
  1}^{n} \lambda \left( Y_{n i}, A_{n i}, Z_{i} \middle| \mathbb{S} \left( Z_{i}
  \right) \right).
\end{equation*}
Applying \th\ref{lem--car-vector-clt} to the right hand side of the above
expression,
\begin{equation*}
  \frac{1}{\sqrt{n}} \dot{\ell}_{n} \overset{\mathrm{d}}{\to}
  \mathcal{N} \left( \mathbf{0}_{d}, \mathcal{J} + \sum_{s =
  1}^{\mathcal{S}} Q_{0} (\mathbb{S} (Z) = s) \sum_{a \in \{0, 1\}}^{n} \pi_{a}
  (s) \mathcal{J} (a, s) \right) = \mathcal{N} \left( \mathbf{0}_{d},
  \mathcal{I} \right)
\end{equation*}
where the last equality follows by \eqref{eqn--strata-info-mat-breakdown}.

We now prove \th\ref{lem--parametric-submodel-lan}
\ref{lem--parametric-submodel-lan-uan}. As before, \(S_{i} = \mathbb{S} \left(
Z_{i} \right)\) for all \(i \in \mathbb{N}_{n}\). For any given \(\theta \in
\Theta\), we can rewrite the log-likelihood in \eqref{eqn--sample-loglhood} as
\begin{align*}
  \ell_{n} (\theta) =
  & \ \sum_{i = 1}^{n} \left[ A_{n i} \log \widetilde{q}_{1} \left( Y_{n i},
    Z_{i} \middle| S_{i}; \theta \right) + \left( 1 - A_{n i} \right) \log
    \widetilde{q}_{0} \left( Y_{n i}, Z_{i} \middle| S_{i}; \theta \right) +
    \sum_{i = 1}^{n} \log \phi \left( S_{i}; \theta \right) \right] \\
  & + \log \alpha_{n} \left( \mathbf{A}_{n} \middle| \mathbf{S}_{n} \right).
\end{align*}
Thus, the log-likelihood ratio under the \(t / \sqrt{n}\) alternatives in
\eqref{eqn--lr-quad-breakdown} becomes
\begin{align*}
  \ell_{n} \left( \theta_{0} + (t / \sqrt{n}) \right) - \ell_{n}
  \left(\theta_{0} \right) =
  & \ \sum_{i = 1}^{n} A_{n i} \log \frac{\widetilde{q}_{1} \left( Y_{i} (1),
    Z_{i} \middle| S_{i}; \theta_{0} + (t / \sqrt{n}) \right)}{\widetilde{q}_{1}
    \left( Y_{i} (1), Z_{i} \middle| S_{i}; \theta_{0} \right)} \\
  & + \sum_{i = 1}^{n} \left( 1 - A_{n i} \right) \log \frac{\widetilde{q}_{0}
    \left( Y_{i} (0), Z_{i} \middle| S_{i}; \theta_{0} + (t / \sqrt{n})
    \right)}{\widetilde{q}_{0} \left( Y_{i} (0), Z_{i} \middle| S_{i};
    \theta_{0} \right)} \\
  & + \sum_{i = 1}^{n} \log \frac{\phi \left( S_{i}; \theta_{0} + (t / \sqrt{n})
    \right)}{\phi \left( S_{i}; \theta_{0} \right)}
\end{align*}
where we have replaced \(Y_{n i}\) with the appropriate potential outcome
following \eqref{eqn--obsoutproc}. As a result, the remainder term \(R_{n}
\left( \theta_{0}, t \right)\) in \eqref{eqn--lr-quad-breakdown} can be written
as
\begin{align*}
  R_{n} \left( \theta_{0}, t \right) =
  & \ \ell_{n} \left( \theta_{0} + (t / \sqrt{n}) \right) - \ell_{n}
    \left(\theta_{0} \right) - t^{\prime} \frac{1}{\sqrt{n}} \dot{\ell}_{n} +
    \frac{1}{2} t^{\prime} \mathcal{I} t \\
  =
  & \ \sum_{i = 1}^{n} A_{n i} \left[ \log \frac{\widetilde{q}_{1} \left(Y_{i}
    (1), Z_{i} \middle| S_{i}; \theta_{0} + (t / \sqrt{n})
    \right)}{\widetilde{q}_{1} \left( Y_{i} (1), Z_{i} \middle| S_{i};
    \theta_{0} \right)} - t^{\prime} \frac{1}{\sqrt{n}} \lambda_{1} \left( Y_{i}
    (1), Z_{i} \middle| S_{i} \right) + \frac{1}{2 n} t^{\prime} \mathcal{J}
    \left( 1, S_{i} \right) t \right] \\
  & + \sum_{i = 1}^{n} \left( 1 - A_{n i} \right) \left[ \log
    \frac{\widetilde{q}_{0} \left( Y_{i} (0), Z_{i} \middle| S_{i}; \theta_{0} +
    (t / \sqrt{n}) \right)}{\widetilde{q}_{0} \left( Y_{i} (0), Z_{i} \middle|
    S_{i}; \theta_{0} \right)} - t^{\prime} \frac{1}{\sqrt{n}} \lambda_{0}
    \left( Y_{i} (0), Z_{i} \middle| S_{i} \right) + \frac{1}{2 n} t^{\prime}
    \mathcal{J} \left( 0, S_{i} \right) t \right] \\
  & + \sum_{i = 1}^{n} \left[ \log \frac{\phi \left( S_{i}; \theta_{0} + (t /
    \sqrt{n}) \right)}{\phi \left( S_{i}; \theta_{0} \right)} - t^{\prime}
    \frac{1}{\sqrt{n}} r \left( S_{i} \right) + \frac{1}{2 n} t^{\prime}
    \mathcal{J} t \right] \\
  & + \sum_{s = 1}^{\mathcal{S}} \sum_{a \in \{0, 1\}} \left[ \frac{N_{n} (a,
    s)}{n} - \pi_{a} (s) Q_{0} (\mathbb{S} (Z) = s) \right] t^{\prime}
    \mathcal{J} (a, s) t
\end{align*}
By rearranging terms, we can write
\begin{align}
  R_{n} \left( \theta_{0}, t \right) =
  & \ R_{n, 1} \left( \theta_{0}, t \right) + R_{n, 2} \left( \theta_{0}, t
    \right) + \sum_{s = 1}^{\mathcal{S}} \sum_{a \in \{0, 1\}} R_{n, 3} \left(
    a, s; \theta_{0}, t \right)
    \label{eqn--Rn-breakdown} \\
  \text{where } R_{n, 1} \left( \theta_{0}, t \right) =
  & \ \frac{1}{2} \sum_{s = 1}^{\mathcal{S}} \sum_{a \in \{0, 1\}} \left[
    \pi_{a} (s) Q_{0} (\mathbb{S} (Z) = s) - \frac{N_{n} (a, s)}{n} \right]
    t^{\prime} \mathcal{J} (a, s) t
  \label{eqn--Rn1}
  \\
  R_{n, 2} \left( \theta_{0}, t \right) =
  & \ \sum_{i = 1}^{n} \left[ \log \frac{\phi \left( S_{i}; \theta_{0} + (t /
    \sqrt{n}) \right)}{\phi \left( S_{i}; \theta_{0} \right)} - t^{\prime}
    \frac{1}{\sqrt{n}} r \left( S_{i} \right) + \frac{1}{2 n} t^{\prime}
    \mathcal{J} t \right]
  \label{eqn--Rn2}
  \\
  R_{n, 3} \left( a, s; \theta_{0}, t \right) =
  & \ \sum_{i = 1}^{n} \mathbb{I} \left\{ A_{n i} = a, S_{i} = s \right\} \left[
    \begin{array}{l}
      \log \frac{\widetilde{q}_{a} \left( Y_{i} (a), Z_{i} \middle| s;
      \theta_{0} + (t / \sqrt{n}) \right)}{\widetilde{q}_{a} \left( Y_{i} (a),
      Z_{i} \middle| s; \theta_{0} \right)} \\
      - t^{\prime} \frac{1}{\sqrt{n}} \lambda_{a} \left( Y_{i} (a), Z_{i}
      \middle| s \right) + \frac{1}{2 n} t^{\prime} \mathcal{J} \left( a, s
      \right) t
    \end{array}
    \right].
  \label{eqn--Rn3}
\end{align}
For \th\ref{lem--parametric-submodel-lan}, it suffices to show each of the above
terms satisfies
\begin{equation}
  \begin{split}
    & \text{for all } \rho_{n} \in \left\{ R_{n,1}, R_{n, 2} \right\} \cup
      \left\{ R_{n, 3} (a, s; \cdot) : a \in \{0, 1\}, s \in
      \mathbb{N}_{\mathcal{S}} \right\} \text{ and all } \varepsilon, M \in (0,
      \infty) \\
    & \lim_{n \to \infty} \sup_{\|t\| \leq M} P_{0 n} \left( \left| \rho_{n}
      \left( \theta_{0}, t \right) \right| > \varepsilon \right) = 0.
  \end{split}
  \label{eqn--sufficient-parametric-submodel-lan-uan}
\end{equation}

We start by showing that \eqref{eqn--sufficient-parametric-submodel-lan-uan}
holds for \eqref{eqn--Rn1}. For every \(\|t\| \leq M\),
\begin{equation}
  \left| R_{n, 1} \left( \theta_{0}, t \right) \right| \leq \frac{1}{2}
  \mathcal{S} M^{2} G \max_{a \in \{0, 1\}, s \in \mathbb{N}_{\mathcal{S}}}
  \left| \pi_{a} (s) Q_{0} (\mathbb{S} (Z) = s) - \frac{N_{n} (a, s)}{n} \right|
\end{equation}
for \(G = \max_{a \in \{0, 1\}, s \in \mathbb{N}_{\mathcal{S}}} \left\|
\mathcal{J} (a, s) \right\|\). Thus,
\begin{equation*}
  P_{0 n} \left( \left| R_{n, 1} \left( \theta_{0}, t \right) \right| >
  \varepsilon \right) \leq P_{0 n} \left( \max_{a \in \{0, 1\}, s \in
  \mathbb{N}_{\mathcal{S}}} \left| \pi_{a} (s) Q_{0} (\mathbb{S} (Z) = s) -
  \frac{N_{n} (a, s)}{n} \right| > \frac{2 \varepsilon}{\mathcal{S} M^{2} G}
  \right)
\end{equation*}
which implies
\begin{equation*}
  \sup_{\|t\| \leq M}P_{0 n} \left( \left| R_{n, 1} \left( \theta_{0}, t
  \right) \right| > \varepsilon \right) \leq P_{0 n} \left( \max_{a \in \{0,
  1\}, s \in \mathbb{N}_{\mathcal{S}}} \left| \pi_{a} (s) Q_{0} (\mathbb{S} (Z)
  = s) - \frac{N (a, s)}{n} \right| > \frac{2 \varepsilon}{\mathcal{S} M^{2} G}
  \right).
\end{equation*}
The right hand side of the above tends to zero as \(n \to \infty\) by
\th\ref{lem--treat-strat-sample-prop-conv}.
This shows that \eqref{eqn--sufficient-parametric-submodel-lan-uan} holds for
\eqref{eqn--Rn1}.

Next, to see that \eqref{eqn--sufficient-parametric-submodel-lan-uan} holds for
\eqref{eqn--Rn2}, we have already argued that \(\theta \mapsto \sqrt{\phi
(\cdot; \theta)}\) is differentiable in quadratic mean with score \(r
(\cdot)\) \eqref{eqn--score-in-strata} and information matrix \(\mathcal{J}\) in
\eqref{eqn--info-in-strata} due to Proposition A.5.5 of
\citet[p. 461]{1998bickelEfficientAdaptiveEstimation}. Then, an application of
Proposition 2.1.2 of \citet[p. 16]{1998bickelEfficientAdaptiveEstimation} yields
the desired conclusion. Therefore, it remains to show that
\eqref{eqn--sufficient-parametric-submodel-lan-uan} for \eqref{eqn--Rn3}. To
that end, note that by the coupling result in \th\ref{lem--coupling-lemma} and
the fact that the distribution of \(R_{n, 3} \left( \theta_{0}, t \right)\) is
invariant to permutation of the indices of the summands,
\begin{equation}
  R_{n, 3} \left( a, s; \theta_{0}, t \right) \overset{\mathrm{d}}{=} R_{n,
  3}^{\ast} \left( a, s; \theta_{0}, t \right) := \sum_{i = 1}^{N_{n} (a, s)}
  \left[
  \begin{array}{l}
    \log \frac{\widetilde{q}_{a} \left( Y_{i} (a, s), Z_{i} (s) \middle| s;
    \theta_{0} + (t / \sqrt{n}) \right)}{\widetilde{q}_{a} \left( Y_{i} (a, s),
    Z_{i} (s) \middle| s; \theta_{0} \right)} \\
    - t^{\prime} \frac{1}{\sqrt{n}} \lambda_{a} \left( Y_{i} (a, s), Z_{i} (s)
    \middle| s \right) + \frac{1}{2 n} t^{\prime} \mathcal{J} \left( a, s
    \right) t
  \end{array}
  \right]
\end{equation}
where \(\left\{ Y_{i} (a, s), Z_{i} (s) \right\}\) are defined by
\th\ref{asm--coupling-construct} and are independent to \(\mathbf{A}_{n},
\mathbf{W}_{n}\). By \th\ref{lem--qmd-strata-cond-density}, the map \(\theta
\mapsto \sqrt{\widetilde{q}_{a} (\cdot | s; \theta)}\) is differentiable in
quadratic mean at \(\theta_{0}\) with score \(\lambda_{a} (\cdot | s)\) in
\eqref{eqn--conditional-score-in-strata} and information matrix \(\mathcal{J}
(a, s)\) in \eqref{eqn--conditional-info-in-strata}. By Proposition 2.1.1 of
\citet[p. 16]{1998bickelEfficientAdaptiveEstimation},
\begin{equation}
  \sum_{i = 1}^{m}
  \left[
  \begin{array}{l}
    \log \frac{\widetilde{q}_{a} \left( Y_{i} (a, s), Z_{i} (s) \middle| s;
    \theta_{0} + (t / \sqrt{n}) \right)}{\widetilde{q}_{a} \left( Y_{i} (a, s),
    Z_{i} (s) \middle| s; \theta_{0} \right)} \\
    - t^{\prime} \frac{1}{\sqrt{n}} \lambda_{a} \left( Y_{i} (a, s), Z_{i} (s)
    \middle| s \right) + \frac{1}{2 n} t^{\prime} \mathcal{J} \left( a, s
    \right) t
  \end{array}
  \right] \overset{\mathrm{p}}{\to} 0
\end{equation}
uniformly over \(\|t\| \leq M\) in the sense of
\eqref{eqn--sufficient-parametric-submodel-lan-uan} as \(m \to
\infty\). Invariant over choices of \(\|t\| \leq M\),
\th\ref{lem--treat-strat-sample-prop-conv} and
\th\ref{lem--ratio-conv-implies-div} imply that \(N_{n} (a, s)
\overset{\mathrm{p}}{\to} \infty\). Thus, by \th\ref{lem--random-subseq-conv},
we have \eqref{eqn--sufficient-parametric-submodel-lan-uan} for
\eqref{eqn--Rn2}.
\end{proof}

\subsection{Proof of Lemma \ref{lem--parametric-submodel-efficiency}}

\begin{proof}[Proof of \th\ref{lem--parametric-submodel-efficiency}]
By \th\ref{lem--parametric-submodel-lan}, \(\mathcal{P}^{0}\) has the local
asymptotic normality (LAN) property of
\citet{1960lecamLocallyAsymptoticallyNormal} with information matrix
\(\mathcal{I}\) as in \eqref{eqn--strata-info-mat}. Now, assuming for the moment
that \(\mathcal{I}\) is indeed nonsingular,
\th\ref{lem--parametric-submodel-efficiency}
\ref{lem--parametric-submodel-efficiency-convolution} follows by
\citet{1970hajekCharacterizationLimitingDistributions}. This result only
requires the LAN property with a non-singular information matrix and regularity
of the estimator sequence. \th\ref{lem--parametric-submodel-efficiency}
\ref{lem--parametric-submodel-efficiency-lam} follows by the local asymptotic
minimax theorem (\citet[Theorem 4.2]{1972hajekLocalAsymptoticMinimax}), and its
extension to the vector-valued parameters (see for instance
\citet[p. 27]{1998bickelEfficientAdaptiveEstimation} or \citet[Theorem II.12.1
and Remark II.12.2]{1981ibragimovStatisticalEstimation}). The local asymptotic
minimax theorems only require the LAN property with a non-singular information
matrix to hold.

Thus it remains to show that \(\mathcal{I}\) is non-singular. Recall from
\eqref{eqn--strata-info-mat} that
\begin{align*}
  \mathcal{I} = \mathbb{E} \left[ \dot{\ell} (Y, A, Z) \dot{\ell} (Y, A,
  Z)^{\prime} \right] = \sum_{s = 1}^{\mathcal{S}} Q_{0} (\mathbb{S} (Z) = s)
  \sum_{a \in \{0, 1\}} \pi_{a} (s) \mathbb{E} \left[ \dot{\ell}_{a} (Y (a), Z)
  \dot{\ell}_{a} (Y (a), Z)^{\prime} \middle| \mathbb{S} (Z) = s \right].
\end{align*}
where as before, we set \(\pi_{a} (s) = \pi (s)^{a} (1 - \pi (s))^{1 - a}\).
First note that \(\mathcal{I}\) is positive semi-definite: for any vector \(v
\in \mathbb{R}^{d}\),
\begin{align}
  v^{\prime} \mathcal{I} v =
  & \ \sum_{s = 1}^{\mathcal{S}} Q_{0} (\mathbb{S} (Z) = s) \sum_{a \in \{0,
    1\}} \pi_{a} (s) v^{\prime} \mathbb{E} \left[ \dot{\ell}_{a} (Y (a), Z)
    \dot{\ell}_{a} (Y (a), Z)^{\prime} \middle| \mathbb{S} (Z) = s \right] v
    \nonumber \\
  =
  & \ \sum_{s = 1}^{\mathcal{S}} Q_{0} (\mathbb{S} (Z) = s) \sum_{a \in \{0,
    1\}} \pi_{a} (s)  \mathbb{E} \left[ \left( v^{\prime} \dot{\ell}_{a} (Y (a),
    Z) \right)^{2} \middle| \mathbb{S} (Z) = s \right]
    \label{eqn--treat-strat-quad-form}
  \\
  \geq
  & \ 0. \nonumber
\end{align}
To show non-singularity of \(\mathcal{I}\), it therefore suffices to show that
\(\mathcal{I}\) is in fact positive definite. We will argue this by
contradiction. Suppose that there exists a \(v \in \mathbb{R}^{d}\) such that
\eqref{eqn--treat-strat-quad-form} is equal to zero. Now,
\eqref{eqn--treat-strat-quad-form} is non-negatively weighted sum of
non-negative terms and we know that
\begin{equation*}
  \sum_{s = 1}^{\mathcal{S}} Q_{0} (\mathbb{S} (Z) = s) \sum_{a \in \{0, 1\}}
  \pi_{a} (s) = 1,
\end{equation*}
so there is at least one positive weight. Furthermore, since \(\mathbb{S} (Z)\)
is a discrete random variable, we can discard any \(s\)'s for which \(Q_{0}
(\mathbb{S} (Z) = s) = 0\) and retain only \(s\)'s for which \(Q_{0} (\mathbb{S}
(Z) = s) > 0\) without any issues since these latter probabilities will still
sum to one. So, let us assume \(Q_{0} (\mathbb{S} (Z) = s) > 0\) for all \(s \in
\mathbb{N}_{\mathcal{S}}\). In addition, \(\pi (s) \in (0, 1)\) by
assumption. So, \eqref{eqn--treat-strat-quad-form} is in fact a positively
weighted sum of non-negative terms with weights \(Q_{0} (\mathbb{S} (Z) = s)
\pi_{a} (s)\) across \(s \in \mathbb{N}_{\mathcal{S}}\) and \(a \in \{0,
1\}\). Thus having \eqref{eqn--treat-strat-quad-form} equal to zero means that
\begin{equation*}
  \left[ \left( v^{\prime} \dot{\ell}_{a} (Y (a), Z) \right)^{2} \middle|
  \mathbb{S} (Z) = s \right] = 0 \qquad \text{for all } s \in
  \mathbb{N}_{\mathcal{S}} \text{ and } a \in \{0, 1\}.
\end{equation*}
This then means that for each \(a \in \{0, 1\}\), \(\mathcal{I}_{a}\) in
\eqref{eqn--parametric-info-mat} is singular since for \(v\) satisfying the
above, we have
\begin{align*}
  v^{\prime} \mathcal{I}_{a} v =
  & \ v^{\prime} \mathbb{E} \left[ \dot{\ell}_{a} (Y (a), Z) \dot{\ell}_{a} (Y
    (a), Z)^{\prime} \right] v = \mathbb{E} \left[ \left( v^{\prime}
    \dot{\ell}_{a} (Y (a), Z) \right)^{2} \right] \\
  =
  & \ \sum_{s = 1}^{\mathcal{S}} Q_{0} (\mathbb{S} (Z) = s) \mathbb{E} \left[
    \left( v^{\prime} \dot{\ell}_{a} (Y (a), Z) \right)^{2} \middle| \mathbb{S}
    (Z) = s \right] \\
  = & \ \sum_{s = 1}^{\mathcal{S}} Q_{0} (\mathbb{S} (Z) = s) \cdot 0 = 0
\end{align*}
This is a contradiction to the hypothesis of regularity of \(\mathcal{P}^{0}\)
as in \th\ref{def--qmd}. Therefore, \(\mathcal{I}\) must necessarily be strictly
positive definite and hence, non-singular.
\end{proof}

%%% Local Variables:
%%% mode: latex
%%% TeX-master: "../2023_semipar_eff_car"
%%% End:
% LocalWords:  SSRA SPBR ATT
% ! TEX root = ../2023_semipar_eff_car.tex

\section{Auxiliary Results}

\begin{remark}[Notation for this section]
\th\label{rem--added-notation}
In what follows, recall that \(N_{n} (s) = \sum_{i = 1}^{n} \mathbb{I} \left\{
\mathbb{S} \left( Z_{i} \right) = s \right\}\) and \(N_{n} (a, s) = \sum_{i =
1}^{n} \mathbb{I} \left\{ A_{n i} = a, \mathbb{S} \left( Z_{i} \right)
= s \right\}\). Additionally, let \(F (1) = 0\), \(\widetilde{N}_{n} (1) = 0\)
and for \(s \geq 2\), \(F (s) = \sum_{t = 1}^{s - 1} Q_{0} (\mathbb{S} (Z) =
t)\), \(\widetilde{N}_{n} (s) = \sum_{t = 1}^{s - 1} N_{n} (t)\).
\end{remark}

\subsection{Pathwise differentiability of the ATE}

\begin{lemma}
\th\label{lem--tangent-space-P}
Let \(\mathbf{P}\) be as defined in \th\ref{asm--car-pop-ssra} and
\th\ref{rem--identification-ate-P}. The tangent space for \(\mathbf{P}\) at
\(P_{0}\) is \(\dot{\mathbf{P}} \subseteq L_{2} \left( P_{0} \right)\) such that
for every \(h \in \dot{\mathbf{P}}\),
\begin{enumerate}[label=(\alph*)]
\item \label{lem--tangent-space-P-marginal-mean-zero}
  \(\forall a \in \{0, 1\}\), \(\int h (y, a, z) q_{a} \left( y, z; Q_{0}
  \right) \nu_{a} (\mathrm{d} y, \mathrm{d} z) = 0\).
\item \label{lem--tangent-space-P-Z-marginal-equal}
  \(\int_{\mathbb{R}} h (y, 1, z) q_{1} (y, z; Q_{0}) \mu_{1} (\mathrm{d} y)
  = \int h (y, 0, z) q_{0} \left( y, z; Q_{0} \right) \mu_{0} (\mathrm{d} y)\)
  holds for \(z\) \(\mu_{Z}\)-a.e.
\end{enumerate}
\end{lemma}

\begin{proof}
Let \(\mathbf{P}_{0} = \left\{ P_{\theta} : \theta \in \Theta \right\} \subseteq
\mathbf{P}\) be a (one-dimensional) regular parametric submodel of
\(\mathbf{P}\) identifying \(P_{0}\) uniquely at \(\theta_{0}\). Following
\th\ref{rem--prod-struct-ssra}, each \(P_{\theta} \in \mathbf{P}_{0}\) has a
density against \(\nu\) in \eqref{eqn--nu} of the form
\begin{equation*}
  p (y, a, z; \theta) = \left[ q_{1} (y, z; \theta) \cdot \pi (\mathbb{S} (z))
  \right]^{a} \left[ q_{0} (y, z; \theta) \cdot (1 - \pi (\mathbb{S} (z)))
  \right]^{1 - a}.
\end{equation*}
Let \(D_{0} \left( \cdot; \mathbf{P}_{0} \right), D_{1} \left( \cdot;
\mathbf{P}_{0} \right)\) be the quadratic mean derivatives of \(\sqrt{q_{0}
(\cdot; \theta)}\) and \(\sqrt{q_{1} (\cdot; \theta)}\) respectively at
\(\theta_{0}\) (as in \th\ref{def--qmd}). The associated score function is then
\begin{equation}
  h \left( y, a, z; \mathbf{P}_{0} \right) \equiv 2 \frac{D_{a} \left( y,
  z; \mathbf{P}_{0} \right)}{\sqrt{q_{a} \left( y, z; Q_{0} \right)}} \mathbb{I}
  \left\{ q_{a} \left( y, z; Q_{0} \right) > 0 \right\}.
  \label{eqn--score-tangent-set}
\end{equation}
Consider the collection of the above functions indexed by any regular parametric
submodel, i.e.
\begin{equation*}
  \mathcal{H} = \left\{ h \left( \cdot; \mathbf{P}_{0} \right) : \mathbf{P}_{0}
  \text{ is a regular parametric submodel of } \mathbf{P} \right\}.
\end{equation*}
\(\mathcal{H}\) is the tangent set for \(\mathbf{P}\) at \(P_{0}\) and its
closed linear span, denoted by \(\mathrm{clspan} (\mathcal{H})\) is the
tangent space. Our task is then to prove that both \(\mathrm{clspan}
(\mathcal{H}) \subseteq \dot{\mathbf{P}}\) and \(\dot{\mathbf{P}} \subseteq
\mathrm{clspan} (\mathcal{H})\).

We start with proving \(\mathrm{clspan} (\mathcal{H}) \subseteq
\dot{\mathbf{P}}\). Note that \(\dot{\mathbf{P}}\) is a closed linear subspace
of \(L_{2} \left( P_{0} \right)\) and so it suffices to prove that \(\mathcal{H}
\subseteq \dot{\mathbf{P}}\). Take any regular parametric submodel
\(\mathbf{P}_{0}\) and set \(D_{a} \equiv D_{a} \left( \cdot; \mathbf{P}_{0}
\right)\) and \(h \equiv h \left( \cdot; \mathbf{P}_{0} \right)\). Clearly, \(h
\in L_{2} \left( P_{0} \right)\) since by \eqref{eqn--score-tangent-set}, for
fixed \(a \in \{0, 1\}\),
\begin{align*}
  \int |h (y, a, z)|^{2} q_{a} \left( y, z; Q_{0} \right) \nu_{a} (\mathrm{d} y,
  \mathrm{d} z) =
  & \ 4 \int |D_{a} (y, z)|^{2} \mathbb{I} \left\{ q_{a} \left( y,
  z; Q_{0} \right) > 0 \right\} \nu_{a} (\mathrm{d} y, \mathrm{d} z) \\
  \leq
  & \ 4 \int |D_{a} (y, z)|^{2} \nu_{a} (\mathrm{d} y, \mathrm{d} z) < \infty
\end{align*}
where finiteness follows since \(D_{a} \in L_{2} \left( \nu_{a} \right)\) by
definition. Next, denote \(\langle f, g \rangle_{\nu_{a}} = \int f \cdot g \;
\mathrm{d} \nu_{a}\). Set \(\xi_{a} (\cdot; \theta) = \sqrt{q_{a} (\cdot;
\theta)}\) and note that \(\left\langle \xi_{a} (\cdot; \theta), \xi_{a} (\cdot;
\theta) \right\rangle_{\nu_{a}} = 1\) for every \(\theta \in
\Theta\). Applying the chain rule and differentiating both sides at \(\theta =
\theta_{0}\) yields \(2 \left\langle D_{a}, \xi \left( \cdot; \theta_{0} \right)
\right\rangle_{\nu_{a}} = 0\), which can be rewritten as
\begin{equation*}
  0 = \int D_{a} (y, z) \sqrt{q_{a} \left( y, z; Q_{0} \right)} \; \nu_{a}
  \left( \mathrm{d} y, \mathrm{d} z \right) = \frac{1}{2} \int h (y, a, z) q_{a}
  \left( y, a; Q_{0} \right) \nu_{a} \left( \mathrm{d} y, \mathrm{d} z \right).
\end{equation*}
Since the choice of \(a \in \{0, 1\}\) in the above is arbitrary, this proves
that \(h\) satisfies condition \ref{lem--tangent-space-P-marginal-mean-zero} in
the definition of \(\dot{\mathbf{P}}\). Consider the coordinate projection \(T :
\mathbb{R}^{1 + k} \to \mathbb{R}^{k}\) defined by \(T (y, z) \equiv z\). The
marginal density of \(Z\) under \(q_{a} (\cdot; \theta)\) is equal to \(g (z;
\theta) = \int q_{a} (y, z; \theta) \mu_{a} (\mathrm{d} y)\) regardless of \(a
\in \{0, 1\}\) as required by \th\ref{def--parametric-submodel}. By Proposition
A.5.5 of \citet[p. 461]{1998bickelEfficientAdaptiveEstimation}, it
follows that \(\eta (z; \theta) = \sqrt{g (z; \theta)}\) is
differentiable in quadratic mean at \(\theta_{0}\) with derivative
\begin{equation*}
  \dot{\eta} \left( z \right) = \frac{1}{2} \sqrt{g \left( z;
  Q_{0} \right)} \int h (y, a, z) \frac{q_{a} \left( y, z; Q_{0} \right)}{g
  \left( z; Q_{0} \right)} \mu_{a} (\mathrm{d} y).
\end{equation*}
Note that \(\dot{\eta}\) must be common to any choice of \(a \in \{0, 1\}\)
since it is the quadratic mean derivative of \(\theta \mapsto \sqrt{g (\cdot;
\theta)}\) which does not depend on \(a\). Rearranging terms, we get that
\begin{equation*}
  \int h (y, a, z) q_{a} \left( y, z; Q_{0} \right) \mu_{1} (\mathrm{d} y)
  = 2 \dot{\eta} \left( z \right) \sqrt{g \left( z;
  Q_{0} \right)} = \int h (y, a, z) q_{0} \left( y, z; Q_{0} \right) \mu_{0}
  (\mathrm{d} y),
\end{equation*}
which shows that \(h\) satisfies \ref{lem--tangent-space-P-Z-marginal-equal}.
Thus, \(h \in \dot{\mathbf{P}}\) and hence \(\mathrm{clspan} (\mathcal{H})
\subseteq \dot{\mathbf{P}}\).

Next, we show that \(\dot{\mathbf{P}} \subseteq \mathrm{clspan} (\mathcal{H})\)
following the construction in
\citet[p. 52]{1998bickelEfficientAdaptiveEstimation}. Take any \(h \in
\dot{\mathbf{P}}\). If \(h = 0\) \(P_{0}\) almost everywhere, then it is
immediate that \(h \in \mathrm{clspan} (\mathcal{H})\) so assume that this is
not the case. Let \(\Psi : \mathbb{R} \to \mathbb{R}\) be a bounded
continuously differentiable positive-valued function with bounded derivative
\(\psi\), such that \(\psi / \Psi\) is bounded and \(\Psi (0) = \psi (0) =
1\). For example, we can take \(\Psi (t) = 2 / \left( 1 + e^{- 2 t}
\right)\). Let \(p_{0} = p \left( \cdot; Q_{0} \right)\) be the density of
\(P_{0} = P \left( \cdot; Q_{0} \right)\). Then for \(\varepsilon > 0\) and
\(\theta \in \Theta = (- \varepsilon, \varepsilon)\), let
\begin{align*}
  q_{a} (y, z; \theta) =
  & \ q_{a} \left( y, z; Q_{0} \right) \frac{\Psi (\theta
    \cdot h (y, a, z))}{\int \Psi \left( \theta \cdot h \left( \widetilde{y}, a,
    \widetilde{z} \right) \right) q_{a} \left( \widetilde{y}, \widetilde{z};
    Q_{0} \right) \nu_{a} \left( \mathrm{d} \widetilde{y}, \mathrm{d}
    \widetilde{z} \right)}
  \\
  p (y, a, z; \theta) =
  & \ \left[ q_{1} (y, z; \theta) \cdot \pi (\mathbb{S} (z)) \right]^{a} \left[
    q_{0} (y, z; \theta) \cdot (1 - \pi (\mathbb{S} (z))) \right]^{1 - a} \\
  =
  & \ \pi (\mathbb{S} (z))^{a} (1 - \pi (\mathbb{S} (z)))^{1 - a} q_{a} (y, z;
    \theta)
\end{align*}
Clearly, each \(p (\cdot; \theta)\) is a \(\nu\)-density and the corresponding
probability measures \(P_{\theta}\) all meet the conditions of
\th\ref{asm--car-pop-ssra} and \th\ref{def--parametric-submodel}. In addition,
\(p (\cdot; 0) = p_{0}\) uniquely so that \(\mathbf{P}_{0}\) is a parametric
submodel of \(\mathbf{P}\) indexed by \(\Theta\) that identifies
\(P_{0}\). Taking the logarithm of the density, we get
\begin{equation*}
  \log p \left( y, a, z; \theta \right) = a \log \pi (\mathbb{S} (z)) + (1 - a)
  \log (1 - \pi (\mathbb{S} (z))) + \log q_{a} (y, z; \theta),
\end{equation*}
so that
\begin{align*}
  \frac{\partial}{\partial \theta} \log p (y, a, z; \theta) =
  & \ \frac{\partial}{\partial \theta} \log q_{a} (y, z; \theta) \\
  =
  & \ \frac{h (y, a, z) \psi (\theta \cdot h (y, a, z))}{\Psi (\theta \cdot h
    (y, a, z))} - \frac{\int h \left( \widetilde{y}, a, \widetilde{z} \right)
    \psi \left( \theta \cdot h \left( \widetilde{y}, a, \widetilde{z} \right)
    \right) q_{a} \left( \widetilde{y}, \widetilde{z}; Q_{0} \right) \nu_{a}
    \left( \mathrm{d} \widetilde{y}, \mathrm{d} \widetilde{z} \right)}{\int \Psi
    \left( \theta \cdot h \left( \widetilde{y}, a, \widetilde{z} \right) \right)
    q_{a} \left( \widetilde{y}, \widetilde{z}; Q_{0} \right) \nu_{a} \left(
    \mathrm{d} \widetilde{y}, \mathrm{d} \widetilde{z} \right)}
\end{align*}
Since \(\int h (y, a, z) q_{a} (y, z) (\mathrm{d} y, \mathrm{d} z) = 0\) by
condition \ref{lem--tangent-space-P-marginal-mean-zero}, we have by substituting
\(\theta = 0\) into the above that
\begin{equation*}
  \left. \frac{\partial}{\partial \theta} \log p (y, a, z; \theta)
  \right|_{\theta = 0} = h (y, a, z).
\end{equation*}
Therefore, \(h\) is the score function for \(\mathbf{P}_{0}\) in the classic
sense (i.e. log-derivatives of the density). In addition, the
map \(\theta \mapsto p (\cdot; \theta)\) is continuously differentiable in
\(\theta\) for \(\nu\) almost all \((y, a, z)\) -- set \((\partial/\partial
\theta) p (\cdot; \theta) = p (\cdot; \theta) \cdot (\partial/\partial
\theta) \log p (\cdot; \theta)\) and note that the right hand side is a product
of continuous functions of \(\theta\) for \(\nu\) almost all \((y, a,
z)\). \(h\) is square integrable by assumption and is not the zero function in
\(L_{2} \left( P_{0} \right)\) and so \(\int h^{2} \mathrm{d} P_{0} >
0\). Applying Proposition 2.1.1 in
\citet[p. 13]{1998bickelEfficientAdaptiveEstimation}, we get that
\(\mathbf{P}_{0}\) is a parametric submodel of \(\mathbf{P}\) that is regular at
\(P_{0}\) (i.e. it is regular at \(\theta = \theta_{0}\)). This proves that
\(h \in \mathcal{H}\) and hence \(\dot{\mathbf{P}_{0}} \subseteq \mathrm{clspan}
(\mathcal{H})\).
\end{proof}

\begin{lemma}
\th\label{lem--pathwise-differentiability-ate}
The ATE parameter \(\beta (\cdot)\) in \eqref{eqn--ate-P} is pathwise
differentiable in the sense of \citet[Definition 3.3.1,
p. 57]{1998bickelEfficientAdaptiveEstimation}. In particular, let
\(\mathbf{P}_{0} = \left\{ P_{\theta} : \theta \in \Theta \right\}\) be a
regular parametric submodel of \(\mathbf{P}\) parameterized by an open and
bounded set \(\Theta \subseteq \mathbb{R}^{d}\) such that \(P_{0} =
P_{\theta_{0}}\) for a unique \(\theta_{0} \in \Theta\). Assume
\(\mathbf{P}_{0}\) is chosen to satisfy
\begin{equation}
  \max_{a \in \{0, 1\}} \sup_{\theta \in \Theta} \int y^{2} \cdot q_{a} (y,
  z; \theta) \; \nu_{a} (\mathrm{d} y, \mathrm{d} z) < \infty.
  \label{eqn--parametric-submodel-ate-square-integrablity}
\end{equation}
Let the map \(\theta \mapsto \gamma (\theta)\) be defined by \(\gamma (\theta) =
\beta \left( P_{\theta} \right)\). Then the map \(\theta \mapsto \gamma
(\theta)\) is differentiable at \(\theta_{0}\) with gradient
\begin{equation}
  \nabla \gamma \left( \theta_{0} \right) = \mathbb{E}_{P_{0}} \left[
  \varphi_{1} (Y, A, Z) \cdot \dot{\ell} \left( Y, A, Z; \theta_{0} \right)
  \right].
  \label{eqn--path-derivative-gamma}
\end{equation}
where
\begin{equation}
  \varphi_{1} (y, a, z) = \frac{y \cdot a}{\pi (\mathbb{S} (z))} - \frac{y \cdot
  (1 - a)}{1 - \pi (\mathbb{S} (z))}.
  \label{eqn--eif-precursor}
\end{equation}
Furthermore, the \(L_{2} \left( P_{0} \right)\) projection of \(\varphi_{1}\)
onto \(\dot{\mathbf{P}}\) is the efficient influence function \(\varphi_{0}\)
defined in \eqref{eqn--eif}.
\end{lemma}

\begin{proof}
Inspecting \eqref{eqn--ate-P}, it follows that
\begin{equation*}
  \beta \left( P \right) = \beta_{1} (P) - \beta_{0} (P) \text{ where }
  \beta_{1} (P) = \mathbb{E}_{P} \left[ \frac{Y A}{\pi (\mathbb{S} (Z))}
  \right] \text{ and } \beta_{0} (P) = \mathbb{E}_{P} \left[ \frac{Y
  (1 - A)}{1 - \pi (\mathbb{S} (Z))} \right].
\end{equation*}
Passing to the submodel \(\mathbf{P}_{0}\), we get \(\gamma (\theta) =
\gamma_{1} (\theta) - \gamma_{0} (\theta)\) with \(\gamma_{a} (\theta) =
\beta_{a} \left( P_{\theta} \right)\). Note then that for \(c \in \{0, 1\}\)
\begin{align*}
  \gamma_{c} (\theta) =
  & \ \mathbb{E}_{P_{\theta}} \left[ \frac{Y A}{\pi (\mathbb{S} (Z))} \cdot c
    + \frac{Y (1 - A)}{1 - \pi (\mathbb{S} (Z))} (1 - c)
    \right] \\
  =
  & \ c \int \frac{y a}{\pi (\mathbb{S} (z))} \left[ q_{1} (y, z; \theta) \cdot
    \pi (\mathbb{S} (z)) \right]^{a} \left[ q_{0} (y, z; \theta) \cdot (1 - \pi
    (\mathbb{S} (z))) \right]^{1 - a} \nu (\mathrm{d} y, \mathrm{d} a,
    \mathrm{d} z) \\
  & + (1 - c) \int \frac{y (1 - a)}{1 - \pi (\mathbb{S} (z))} \left[ q_{1} (y,
    z; \theta) \cdot \pi (\mathbb{S} (z)) \right]^{a} \left[ q_{0} (y, z;
    \theta) \cdot (1 - \pi (\mathbb{S} (z))) \right]^{1 - a} \nu (\mathrm{d} y,
    \mathrm{d} a, \mathrm{d} z) \\
  =
  & \ c \int y q_{1} (y, z) \nu_{1} (\mathrm{d} y, \mathrm{d} z) + (1 - c) \int
    y q_{0} (y, z) \nu_{0} (\mathrm{d} y, \mathrm{d} z)
\end{align*}
Therefore,
\begin{equation}
  \gamma_{c} (\theta) = \int y q_{c} (y, z) \nu_{c} (\mathrm{d} y, \mathrm{d}
  z) \quad \forall c \in \{0, 1\}.
  \label{eqn--parametric-submodel-individual-components}
\end{equation}
Using \eqref{eqn--parametric-submodel-ate-square-integrablity} and
\eqref{eqn--parametric-submodel-individual-components} in conjunction
with Lemma 7.2 of \citet[p. 67]{1981ibragimovStatisticalEstimation}, we get
differentiability of the maps \(\theta \mapsto \gamma_{c} (\theta)\) at
\(\theta_{0}\) for each \(c \in \{0, 1\}\) with the gradients given by
\begin{align*}
  \nabla \gamma_{c} \left( \theta_{0} \right) =
  & \ 2 \int y D_{c} (y, z) \sqrt{q_{c} \left( y, z; Q_{0} \right)}
    \; \nu_{c} (\mathrm{d} y, \mathrm{d} z) \\
  =
  & \ \int y \dot{\ell}_{c} (y, z) q_{c} \left( y, z; Q_{0} \right) \; \nu_{c}
    (\mathrm{d} y, \mathrm{d} z) \\
  =
  & \ \int \frac{y}{\pi (\mathbb{S} (z))^{c} (1 - \pi (\mathbb{S} (z)))^{1 - c}}
    \dot{\ell}_{c} (y, z) \left[ q_{c} \left( y, z; Q_{0} \right) \pi
    (\mathbb{S} (z))^{c} (1 - \pi (\mathbb{S} (z)))^{1 - c} \right] \; \nu_{c}
    (\mathrm{d} y, \mathrm{d} z)
\end{align*}
Using the fact that \(\dot{\ell} (y, a, z) = a \dot{\ell}_{1} (y, z) + (1 - a)
\dot{\ell}_{0} (y, z)\) and the definition of \(\nu\), we can rewrite the above
as
\begin{align*}
  \nabla \gamma_{c} \left( \theta_{0} \right) =
  & \ c \int \frac{y a}{\pi (\mathbb{S} (z))} \dot{\ell} (y, z) \left[ q_{1}
    \left( y, z; Q_{0} \right) \pi (\mathbb{S} (z)) \right]^{a} \left[ q_{0}
    \left( y, z; Q_{0} \right) (1 - \pi (\mathbb{S} (z))) \right]^{1 - a} \; \nu
    (\mathrm{d} y, \mathrm{d} z) \\
  & + (1 - c) \int \frac{y (1 - a)}{1 - \pi (\mathbb{S} (z))} \dot{\ell}
    (y, z) \left[ q_{1} \left( y, z; Q_{0} \right) \pi (\mathbb{S} (z))
    \right]^{a} \left[ q_{0} \left( y, z; Q_{0} \right) (1 - \pi (\mathbb{S}
    (z))) \right]^{1 - a} \; \nu (\mathrm{d} y, \mathrm{d} z).
\end{align*}
Hence,
\begin{equation*}
  \nabla \gamma_{c} \left( \theta_{0} \right) = \mathbb{E}_{P_{0}} \left[ \left(
  \frac{Y A}{\pi (\mathbb{S} (Z))} c + \frac{Y (1 - A)}{1 - \pi
  (\mathbb{S} (Z))} (1 - c) \right) \dot{\ell} (Y, A, Z) \right]
\end{equation*}
Leveraging the fact that \(\gamma (\theta) = \gamma_{1} (\theta) - \gamma_{0}
(\theta)\), we get that
\begin{equation}
  \nabla \gamma \left( \theta_{0} \right) = \nabla \gamma_{1} \left( \theta_{0}
  \right) - \nabla \gamma_{0} \left( \theta_{0} \right) = \mathbb{E}_{P_{0}}
  \left[ \left( \frac{Y A}{\pi (\mathbb{S} (Z))} - \frac{Y (1 - A)}{1 - \pi
  (\mathbb{S} (Z))} \right) \dot{\ell} (Y, A, Z) \right].
  \label{eqn--gamma-derivatives}
\end{equation}
which proves \eqref{eqn--path-derivative-gamma}.

We now show that the projection of \(\varphi_{1}\) in \(L_{2} \left( P_{0}
\right)\) onto \(\dot{\mathbf{P}}\) is \(\varphi_{0}\). To that end, it is
sufficient to show that for any given \(h \in \dot{\mathbf{P}}\),
\begin{equation}
  \mathbb{E}_{P_{0}} \left[ h (Y, A, Z) \cdot \left\{ \varphi_{1} (Y, A, Z) -
  \varphi_{0} (Y, A, Z) \right\} \right] = 0.
  \label{eqn--orthogonality}
\end{equation}
We first note that comparing \eqref{eqn--eif-precursor} and \eqref{eqn--eif},
\begin{equation}
  \varphi_{1} (y, a, z) - \varphi_{0} (y, a, z) = \ \frac{a \cdot m_{\ast} (1,
  z)}{\pi (\mathbb{S} (Z))} - \frac{(1 - a) \cdot m_{\ast} (0, z)}{1 - \pi
  (\mathbb{S} (z))} - \left\{ m_{\ast} (1, z) - m_{\ast} (0, z) - \beta_{0}
  \right\} =: \widetilde{\varphi} (a, z).
  \label{eqn--phitilde-def}
\end{equation}
By writing \(h_{a} (y, z) = h (y, a, z)\) for brevity, we get
\begin{equation*}
  \mathbb{E}_{P_{0}} \left[ h (Y, A, Z) \middle| A = a, Z = z \right]
  = \mathbb{E}_{P_{0}} \left[ h_{a} (Y, Z) \middle| A = a, Z = z \right] =
  \frac{\int h_{a} (y, z) q_{a} \left( y, z; Q_{0} \right) \nu_{a} (\mathrm{d}
  y)}{g \left( z; Q_{0} \right)}
\end{equation*}
Additionally, the above displayed equation and
\th\ref{lem--tangent-space-P} \ref{lem--tangent-space-P-Z-marginal-equal}
together imply that there is together imply that there is some measurable
function \(\widetilde{h} (z)\) square-integrable against \(P_{0}\) such that for
any \(a \in \{0, 1\}\),
\begin{equation}
  \mathbb{E}_{P_{0}} \left[ h (Y, A, Z) \middle| A = a, Z = z \right]
  = \frac{\int h_{a} (y, z) q_{a} \left( y, z; Q_{0} \right) \nu_{a} (\mathrm{d}
  y)}{g \left( z; Q_{0} \right)} = \widetilde{h} (z).
  \label{eqn--g-partial-out-a}
\end{equation}
Furthermore, by the Law of Iterated Expectations (LIE) and
\th\ref{lem--tangent-space-P} \ref{lem--tangent-space-P-marginal-mean-zero},
it is necessary that \(\mathbb{E} \left[\widetilde{h} (Z) \right] = 0\). Hence,
\begin{align}
  \mathbb{E}_{P_{0}} \left[ h (Y, A, Z) \cdot \widetilde{\varphi} (A, Z) \right]
  =
  & \ \mathbb{E}_{P_{0}} \left[ \mathbb{E} [h (Y, A, Z) | A, Z] \cdot
    \widetilde{\varphi} (A, Z) \right]
  & \text{by LIE}
    \nonumber \\
  =
  & \ \mathbb{E}_{P_{0}} \left[ \widetilde{h} (Z) \cdot \widetilde{\varphi} (A,
    Z) \right] \nonumber \\
  =
  & \ \mathbb{E}_{P_{0}} \left[ \widetilde{h} (Z) \cdot \left(
    \begin{array}{l}
      \frac{A \cdot m_{\ast} (1, Z)}{\pi (\mathbb{S} (Z))} - \frac{(1 - A) \cdot
      m_{\ast} (0, Z)}{1 - \pi (\mathbb{S} (Z))} \\
      - \left\{ m_{\ast} (1, Z) - m_{\ast} (0, Z) - \beta_{0} \right\}
    \end{array} \right) \right]
    \label{eqn--hphitilde-expanded}
\end{align}
Note that for any measurable function \(\delta (z)\) that is
\(P_{0}\)-integrable, since \(\mathbb{E}_{P_{0}} [A
| Z] = \pi (\mathbb{S} (Z))\), it is necessarily true that
\begin{equation}
  \mathbb{E}_{P_{0}} \left[ A \delta (Z) \right] = \mathbb{E}_{P_{0}} \left[ \pi
  (\mathbb{S} (Z)) \delta (Z) \right].
  \label{eqn--adelta-lie}
\end{equation}
Combining \eqref{eqn--hphitilde-expanded} and \eqref{eqn--adelta-lie},
we get
\begin{equation}
  \begin{split}
    \mathbb{E}_{P_{0}} \left[ h (Y, A, Z) \cdot \widetilde{\varphi} (A, Z)
    \right]
    =
    & \ \mathbb{E}_{P_{0}} \left[ \widetilde{h} (Z) \cdot \left( m_{\ast} (1, Z)
      - m_{\ast} (0, Z) - \left\{ m_{\ast} (1, Z) - m_{\ast} (0, Z) - \beta_{0}
      \right\} \right) \right] \\
    =
    & \ \mathbb{E}_{P_{0}} \left[ \widetilde{h} (Z) \right] \beta = 0 \cdot
      \beta = 0
  \end{split}
  \label{eqn--hphitilde-0}
\end{equation}
Combining \eqref{eqn--phitilde-def} and \eqref{eqn--hphitilde-0} proves
\eqref{eqn--orthogonality}.
\end{proof}

\subsection{A quotient rule result for differentiability in quadratic mean}

In the following, we show a version of the quotient rule for differentiability
in quadratic mean. This result is required for the proof of
\th\ref{lem--parametric-submodel-lan}. Before that we do some algebra that will
be useful later on. Let \(u (\theta), v (\theta)\) be real-valued functions,
both differentiable at an interior point of their common domain, \(\theta_{0}\),
with derivatives \(\dot{u}, \dot{v}\) respectively at \(\theta_{0}\). If
\(\theta \in \Theta \subseteq \mathbb{R}^{d}\), define \(\dot{u}\) and
\(\dot{v}\) as row vectors. Differentiability of \(v\) at \(\theta_{0}\) means
that \(v\) is continuous at \(\theta_{0}\). Thus if \(v \left( \theta_{0}
\right) > 0\), it follows that \(v (\theta) > 0\) for a neighborhood around
\(\theta_{0}\). Suppose also that for all \(\theta\) in a neighborhood of
\(\theta_{0}\), \(v (\theta) > 0\). For brevity, define \(u_{0} = u
\left(\theta_{0} \right)\), \(v_{0} = v \left( \theta_{0} \right)\) and \(u_{h}
= u \left( \theta_{0} + h \right)\), \(v_{h} = v \left( \theta_{0} + h
\right)\). The proof of the quotient rule follows from a decomposition of the
remainder in linear approximation as follows.
\begin{align}
  \mathrm{Rem}_{h} =
  & \ \frac{u_{h}}{v_{h}} - \frac{u_{0}}{v_{0}} - \left[\frac{\dot{u}}{v_{0}} -
    \frac{u_{0} \dot{v}}{v_{0}^{2}} \right] h \nonumber \\
  =
  & \ v_{h}^{- 1} \left[ u_{h} - u_{0} - \dot{u} h \right] - v_{h}^{- 1}
    v_{0}^{- 1} u_{0} \left[ v_{h} - v_{0} - \dot{v} h \right] + \left[ v_{h}^{-
    1} - v_{0}^{- 1} \right] \left[ \dot{u} - v_{0}^{- 1} u_{0} \dot{v} \right]
    h
    \label{eqn--quotient-rule-remainder-expansion}
\end{align}
Since differentiability of \(v\) at \(\theta_{0}\) implies continuity at
\(\theta_{0}\), \(v_{h} \to v_{0}\) as \(\|h\| \to 0\). Since \(v_{0} > 0\) by
assumption, we must have \(v_{h}^{- 1} \to v_{0}^{- 1}\) as \(\|h\| \to 0\). The
quotient rule follows from division of the above by \(\|h\|\) for \(\|h\| > 0\)
and then taking limits as \(\|h\| \to 0\).

\begin{lemma}
\th\label{lem--qmd-strata-cond-density}
If \(\theta \mapsto \sqrt{q_{a} \left( \cdot; \theta \right)}\) differentiable
in quadratic mean at \(\theta_{0}\) in the sense of \th\ref{def--qmd}
\ref{def--qmd-differentiability}, then the maps \(\theta \mapsto
\sqrt{\widetilde{q}_{a} \left( \cdot | s; \theta \right)}\) in
\eqref{eqn--S-marginal-q-cond-S} are all differentiable in the quadratic mean at
\(\theta_{0}\) with score \(\lambda_{a} (\cdot | s)\) as in
\eqref{eqn--conditional-score-in-strata} and information matrix \(\mathcal{J}
(a, s)\) as in \eqref{eqn--conditional-info-in-strata} across all \(a \in \{0,
1\}\) and \(s \in \mathbb{N}_{\mathcal{S}}\).
\end{lemma}

\begin{proof}
Let \(s \in \mathbb{N}_{\mathcal{S}}\) satisfy \(\phi \left( s; Q_{0} \right) >
0\) and \(a \in \{0, 1\}\) be given. Define \(\zeta_{a} (\cdot | s) =
\frac{1}{2} \lambda_{a} (\cdot | s) \cdot \sqrt{\widetilde{q}_{a} \left( \cdot |
s; Q_{0} \right)}\). Since the form of the information matrix \(\mathcal{J} (a,
s)\) follows readily from \(\lambda_{a} (\cdot | s)\) being the score, we need
to show that \(\zeta_{a} (\cdot | s)\) is the derivative in quadratic mean of
\(\theta \mapsto \sqrt{\widetilde{q}_{a} \left( \cdot | s; \theta
\right)}\). Up to taking integrals, we will assume that all \((y, z)\) pairs in
the subsequent arguments satisfy \(\mathbb{S} (z) = s\). Given any \(h \in
\mathbb{R}^{d}\) such that \(\theta_{0} + h \in \Theta\), define the remainder
in a first order linear approximation (Taylor expansion) around \(\theta_{0}\)
term by
\begin{align}
  \rho_{a, s} \left( \theta_{0}, h \right) =
  & \ \sqrt{\widetilde{q}_{a} \left( y, z | s; \theta_{0} + h \right)} -
    \sqrt{\widetilde{q}_{a} \left( y, z | s; \theta_{0} \right)} - \zeta_{a} (y,
    z | s)^{\prime} h \nonumber \\
  =
  & \ \sqrt{\frac{q_{a} \left( y, z; \theta_{0} + h \right)}{\phi \left( s;
    \theta_{0} + h \right)}} - \sqrt{\frac{q_{a} \left( y, z ;
    \theta_{0} \right)}{\phi \left( s; \theta_{0} \right)}}
    - \zeta_{a} (y, z | s)^{\prime} h
    \label{eqn--1st-order-rem}
\end{align}
Let \(D_{a}\) be as in \th\ref{def--qmd} \ref{def--qmd-differentiability} and
define the derivative in quadratic mean of \(\phi (\cdot; \theta)\) by \(\Delta
\equiv \frac{1}{2} r \sqrt{\phi}\). Using \eqref{eqn--sample-score},
\eqref{eqn--score-in-strata} and \eqref{eqn--conditional-score-in-strata}, it
follows that
\begin{align*}
  \zeta_{a} (y, z | s) =
  & \ \frac{1}{2} \left\{ \dot{\ell}_{a} (y, z) - r (s) \right\}
    \sqrt{\frac{q_{a} \left( y, z ; Q_{0} \right)}{\phi \left( s; Q_{0}
    \right)}} \\
  =
  & \ \left\{ \phi \left( s; Q_{0} \right)^{- \frac{1}{2}}
    D_{a} (y, z) - \frac{1}{2} r (s) \sqrt{\frac{q_{a} \left( y, z ; Q_{0}
    \right)}{\phi \left( s; Q_{0} \right)}} \right\} \\
  =
  & \ \left\{ \phi \left( s; Q_{0} \right)^{- \frac{1}{2}}
    D_{a} (y, z) - \frac{\Delta (s) \sqrt{q_{a} \left( y, z; Q_{0}
    \right)}}{\phi \left( s; Q_{0} \right)} \right\}.
\end{align*}
Note that the last displayed equation is akin to the usual \(\frac{\mathrm{d}
u}{v} - \frac{u \mathrm{d} v}{v^{2}}\) form of the derivative for the quotient
rule applied to \(u / v\) with \(u = \sqrt{q_{a}}\) and \(v = \sqrt{\phi}\).
Substituting the above into
\eqref{eqn--1st-order-rem}, and applying
\eqref{eqn--quotient-rule-remainder-expansion},
\begin{align*}
  \rho_{a, s} \left( y, z; \theta_{0}, h \right) =
  & \ \sqrt{\frac{q_{a} \left( y, z; \theta_{0} + h \right)}{\phi \left( s;
    \theta_{0} + h \right)}} - \sqrt{\frac{q_{a} \left( y, z ;
    \theta_{0} \right)}{\phi \left( s; \theta_{0} \right)}}
    - \zeta_{a} (y, z | s)^{\prime} h \\
  =
  & \ \phi \left( s; \theta_{0} + h \right)^{- \frac{1}{2}}
    \left[ \sqrt{q_{a} \left( y, z; \theta_{0} + h \right)} - \sqrt{q_{a} \left(
    y, z; \theta_{0} \right)} - D_{a} (y, z)^{\prime} h \right] \\
  & - \sqrt{\frac{q_{a} \left( y, z; \theta_{0} \right)}{\phi \left( s;
    \theta_{0} + h \right) \phi \left( s; \theta_{0} \right)}} \left[ \sqrt{\phi
    \left( s; \theta_{0} + h \right)} - \sqrt{\phi \left( s; \theta_{0} \right)}
    - \Delta (s)^{\prime} h \right] \\
  & + \left[ \phi \left( s; \theta_{0} + h \right)^{- \frac{1}{2}} - \phi \left(
    s; \theta_{0} \right)^{- \frac{1}{2}} \right] \left[ D_{a} (y, z) -
    \Delta (s) \sqrt{\frac{q_{a} \left( y, z; \theta_{0} \right)}{\phi \left( s;
    \theta_{0} \right)}} \right]^{\prime} h
\end{align*}
Squaring the above and using the \(C_{r}\)-inequality with parameter \(2\), we
get
\begin{align*}
  \rho_{a, s} \left( y, z; \theta_{0}, h \right)^{2} \leq
  & \ 3 \phi \left( s; \theta_{0} + h \right)^{- 1} \left[ \sqrt{q_{a} \left( y,
    z; \theta_{0} + h \right)} - \sqrt{q_{a} \left( y, z; \theta_{0} \right)} -
    D_{a} (y, z)^{\prime} h \right]^{2} \\
  & + 3 \frac{q_{a} \left( y, z; \theta_{0} \right)}{\phi \left( s;
    \theta_{0} + h \right) \phi \left( s; \theta_{0} \right)} \left[ \sqrt{\phi
    \left( s; \theta_{0} + h \right)} - \sqrt{\phi \left( s; \theta_{0} \right)}
    - \Delta (s)^{\prime} h \right]^{2} \\
  & + 3 \left[ \phi \left( s; \theta_{0} + h \right)^{- \frac{1}{2}} - \phi
    \left( s; \theta_{0} \right)^{- \frac{1}{2}} \right]^{2} \left( \left[ D_{a}
    (y, z) - \Delta (s) \sqrt{\frac{q_{a} \left( y, z; \theta_{0} \right)}{\phi
    \left( s; \theta_{0} \right)}} \right]^{\prime} h \right)^{2} \\
  \leq
  & \ 3 \phi \left( s; \theta_{0} + h \right)^{- 1} \left[ \sqrt{q_{a} \left( y,
    z; \theta_{0} + h \right)} - \sqrt{q_{a} \left( y, z; \theta_{0} \right)} -
    D_{a} (y, z)^{\prime} h \right]^{2} \\
  & + 3 \frac{q_{a} \left( y, z; \theta_{0} \right)}{\phi \left( s;
    \theta_{0} + h \right) \phi \left( s; \theta_{0} \right)} \left[ \sqrt{\phi
    \left( s; \theta_{0} + h \right)} - \sqrt{\phi \left( s; \theta_{0} \right)}
    - \Delta (s)^{\prime} h \right]^{2} \\
  & + 3 \left[ \phi \left( s; \theta_{0} + h \right)^{- \frac{1}{2}} - \phi
    \left( s; \theta_{0} \right)^{- \frac{1}{2}} \right]^{2} \left\| D_{a}
    (y, z) - \Delta (s) \sqrt{\widetilde{q}_{a} \left( y, z | s; \theta_{0}
    \right)} \right\|^{2} \|h\|^{2},
\end{align*}
where the last inequality follows from Cauchy-Schwarz. Dividing by \(\|h\|^{2}\)
and integrating over the region \(\left\{ y, z : \mathbb{S} (z) = s \right\}\),
\begin{align*}
  & \|h\|^{- 2} \int \rho_{a, s} \left( y, z; \theta_{0}, h \right)^{2}
    \mathbb{I} \{ \mathbb{S} (Z) = s \} \nu_{a} (\mathrm{d} y, \mathrm{d} z) \\
  \leq
  & \|h\|^{- 2} \int \rho_{a, s} \left( y, z; \theta_{0}, h \right)^{2} \nu_{a}
  (\mathrm{d} y, \mathrm{d} z) \\
  \leq
  & \ \frac{3}{\phi \left( s; \theta_{0} + h \right) \|h\|^{2}} \int \left[
    \sqrt{q_{a} \left( y, z; \theta_{0} + h \right)} - \sqrt{q_{a} \left( y, z;
    \theta_{0} \right)} - D_{a} (y, z)^{\prime} h \right]^{2} \nu_{a}
    (\mathrm{d} y, \mathrm{d} z) \\
  & + \frac{3 \int q_{a} \left( y, z; \theta_{0} \right) \nu_{a}
    (\mathrm{d} y, \mathrm{d} z)}{\phi \left( s; \theta_{0} + h \right) \phi
    \left( s; \theta_{0} \right)} \frac{1}{\|h\|^{2}} \left[ \sqrt{\phi \left(
    s; \theta_{0} + h \right)} - \sqrt{\phi \left( s; \theta_{0} \right)} -
    \Delta (s)^{\prime} h \right]^{2} \\
  & + 3 \left[ \phi \left( s; \theta_{0} + h \right)^{- \frac{1}{2}} - \phi
    \left( s; \theta_{0} \right)^{- \frac{1}{2}} \right]^{2} \int \left\| D_{a}
    (y, z) - \Delta (s) \sqrt{\widetilde{q}_{a} \left( y, z | s; \theta_{0}
    \right)} \right\|^{2} \nu_{a} (\mathrm{d} y, \mathrm{d} z) \\
  \to & \ 0.
\end{align*}
To see the limit claim at the end, the first term tends to \(0\) by
differentiability in quadratic mean of \(\theta \mapsto \sqrt{q_{a} (\cdot;
\theta)}\). The second term tends to zero by the differentiability in quadratic
mean of \(\theta \mapsto \sqrt{\phi (\cdot; \theta)}\) and the fact that \(\phi
\left( s; \theta_{0} + h \right) \to \phi \left( s; \theta_{0} \right) > 0\) as
\(\|h\| \to 0\), which is also implied by the former observation. The final term
also tends to zero by this last argument. It follows from the Sandwich Theorem
that \(\theta \mapsto \sqrt{\widetilde{q}_{a} (\cdot | s; \theta)}\) is
differentiable in quadratic mean at \(\theta_{0}\) with derivative \(\zeta_{a}
(\cdot | s)\). Then, by definition, the score function is \(\lambda_{a} (\cdot |
s) = 2 \widetilde{q}_{a} \left( \cdot | s; \theta_{0} \right)^{- \frac{1}{2}}
\zeta_{a} (\cdot | s)\).
\end{proof}

\subsection{Useful coupling and limit theorems for covariate adaptive
randomization}

\begin{lemma}
\th\label{lem--coupling-lemma}
Let \(\mathbf{W}^{\ast} = \left\{ W_{i} (s) : i \in \mathbb{N}, s \in
\mathbb{N}_{\mathcal{S}} \right\}\) be as in
\th\ref{asm--coupling-construct}. For \(s \in \mathbb{N}_{\mathcal{S}}\), let
\(\sigma_{n, s} : \mathbb{N}_{n} \to \mathbb{N}_{n}\) be any set of permutations
(bijection). Define
\begin{align}
  Y_{n i}^{\ast} =
  & \ \sum_{s = 1}^{\mathcal{S}} \sum_{a \in \{0, 1\}} Y_{\sigma_{n, s} (i)} (a,
    s) \mathbb{I} \left\{ A_{n i} = a, \mathbb{S} \left( Z_{i} \right) = s
    \right\}
  \\
  Z_{n i}^{\ast} =
  & \ \sum_{s = 1}^{\mathcal{S}} Z_{\sigma_{n, s} (i)} (s) \mathbb{I} \left\{
    \mathbb{S} \left( Z_{i} \right) = s \right\} \\
  X_{n i}^{\ast} =
  & \ \left( Y_{n i}^{\ast}, A_{n i}, Z_{n i}^{\ast} \right) \qquad
    \mathbf{X}_{n}^{\ast} = \left( X_{n 1}, \dots, X_{n n} \right)^{\prime}.
\end{align}
As before, let sample strata be \(\mathbf{S}_{n}^{\prime} = \left(
\mathbb{S} \left( Z_{1} \right), \dots, \mathbb{S} \left( Z_{n} \right)
\right)\). Then the joint distribution of \(\left( \mathbf{X}_{n}^{\ast},
\mathbf{S}_{n} \right)\) is equal to that of \(\left( \mathbf{X}_{n},
\mathbf{S}_{n} \right)\). In addition, \(\left[ \mathbf{X}_{n} \indep
\mathbf{X}_{n}^{\ast} \middle] \right| \mathbf{A}_{n}, \mathbf{S}_{n}\).
\end{lemma}

\begin{proof}
The independence claim at the end is an immediate consequence of
\th\ref{asm--coupling-construct}. The fact that the conclusions hold across all
permutations \(\sigma_{n, s} : \mathbb{N}_{n} \to \mathbb{N}_{n}\) is a
consequence of \(\mathbf{W}^{\ast}\) being independent to
\(\mathbf{A}_{n}, \mathbf{W}_{n}\) for all \(n \in \mathbb{N}\) and having
i.i.d. components across all \((i, s) \in \mathbb{N} \times
\mathbb{N}_{\mathcal{S}}\). As such, it suffices to prove the claims for the
identity permutation, \(\sigma_{n, s} (i) = i\). As before, let
\begin{equation}
  \phi (s; Q) = Q (\mathbb{S} (Z) = s) \quad \text{and} \quad \widetilde{q}_{a}
  (y, z | s; Q) = \frac{q_{a} \left( y, z; Q \right)}{\phi (s; Q)} \mathbb{I}
  \left\{ \mathbb{S} (z) = s \right\}
  \label{eqn--S-marginal-q-cond-S-coupling}
\end{equation}
for each \(s \in \mathbb{N}_{\mathcal{S}}\). \(\phi (\cdot; Q)\) is the marginal
probability mass function of \(\mathbb{S} (Z)\) under \(Q\) and
\(\widetilde{q}_{a} (\cdot | s; Q)\) is the conditional density of \((Y (a),
Z)\) given \(\mathbb{S} (Z) = s\) under \(Q\). Let \(\mathbf{s}_{n}^{\prime} =
\left( s_{1}, \dots, s_{n} \right) \in \mathbb{N}_{\mathcal{S}}^{n}\) be
given. In view of \th\ref{lem--prod-struct}, the joint density of
\(\mathbf{X}_{n}, \mathbf{S}_{n}\) in \eqref{eqn--data-joint-density-prod}
against the \(n\)-fold product measure \(\nu^{n}\) can be factored using
\eqref{eqn--S-marginal-q-cond-S-coupling} as
\begin{align}
  p_{n} \left( \mathbf{y}_{n}, \mathbf{a}_{n}, \mathbf{z}_{n}; Q \right) =
  & \ \alpha_{n} \left( \mathbf{a}_{n} \middle| \mathbf{s}_{n} \right) \prod_{i
    = 1}^{n} q_{1} \left( y_{i}, z_{i}; Q \right)^{a_{i}} q_{0} \left( y_{i},
    z_{i}; Q \right)^{1 - a_{i}} \mathbb{I} \left( \mathbf{s}_{n} = \left(
    \mathbb{S} \left( z_{1} \right), \dots, \mathbb{S} \left( z_{n} \right)
    \right)^{\prime} \right) \nonumber \\
  =
  & \
    \alpha_{n} \left( \mathbf{a}_{n} \middle| \mathbf{s}_{n} \right) \prod_{i =
    1}^{n} \phi \left( s_{i}; Q \right) \prod_{i = 1}^{n} \widetilde{q}_{1}
    \left( y_{i}, z_{i} \middle| s_{i}; Q \right)^{a_{i}} \widetilde{q}_{0}
    \left(y_{i}, z_{i} \middle| s_{i}; Q \right)^{1 - a_{i}}.
    \label{eqn--prod-struct-equi-dist-main}
\end{align}
Note \eqref{eqn--prod-struct-equi-dist-main} is the joint density of \(\left(
\mathbf{X}_{n}, \mathbf{S}_{n} \right)\) against product measure formed from
\(\nu^{n}\) and the \(n\)-fold product counting measure on
\(\mathbb{N}_{\mathcal{S}}^{n}\). Let \(\mathbf{a}_{n}^{\prime} = \left( a_{1},
\dots, a_{n} \right) \in \{0, 1\}^{n}\) and \(\mathbf{s}_{n}^{\prime} = \left(
s_{1}, \dots, s_{n} \right) \in \mathbb{N}_{\mathcal{S}}^{n}\) be
given. Following the same arguments as those in \th\ref{lem--prod-struct}, the
joint density of \(\mathbf{X}_{n}^{\ast} | \left( \mathbf{A}_{n} =
\mathbf{a}_{n}, \mathbf{S}_{n} = \mathbf{s}_{n} \right)\) against \(\nu^{n}\) is
\begin{equation}
  \prod_{i = 1}^{n} \widetilde{q}_{1} \left( y_{i}, z_{i}; Q \right)^{a_{i}}
  \widetilde{q}_{0} \left(y_{i}, z_{i}; Q \right)^{1 - a_{i}}
  \label{eqn--prod-struct-equi-dist-coupling-step}
\end{equation}
By independence of \(\left\{ W_{i} (s) : i \in \mathbb{N}_{n}, s \in
\mathbb{N}_{\mathcal{S}} \right\}\) and \(\mathbf{S}_{n}, \mathbf{A}_{n}\), the
joint density of \(\left( \mathbf{X}_{n}^{\ast}, \mathbf{S}_{n} \right)\)
is the product of \eqref{eqn--prod-struct-equi-dist-coupling-step} and
\(\alpha_{n} \left( \mathbf{a}_{n} \middle| \mathbf{s}_{n} \right) \prod_{i =
1}^{n} \phi \left( s_{i}; Q \right)\) which yields
\eqref{eqn--prod-struct-equi-dist-main}.
\end{proof}

\begin{lemma}
\th\label{lem--treat-strat-sample-prop-conv}
Let \(\pi_{a} (s) = \pi (s)^{a} (1 - \pi (s))^{1 - a}\). Under
\th\ref{asm--iid-Q}, \th\ref{asm--treat-strat} \ref{asm--treat-exog} and
\th\ref{asm--treat-strat} \ref{asm--treat-prop},
\begin{equation}
  \frac{N_{n} (a, s)}{n} \overset{\mathrm{p}}{\to} \pi_{a} (s) Q_{0} (\mathbb{S}
  (Z) = s).
  \label{eqn--treat-strat-sample-prop-conv}
\end{equation}
\end{lemma}

\begin{proof}
Note that to show \eqref{eqn--treat-strat-sample-prop-conv}, it suffices to
prove it for the case \(a = 1\). Let \(s \in \mathbb{N}_{\mathcal{S}}\) be
given. Notice that \(N_{n} (1, s) = 0\) whenever \(N_{n} (s) = 0\). Therefore,
using the convention \(N_{n} (s)^{- 1} N_{n} (1, s) = 0\) if \(N_{n} (s) = 0\),
we can write \(N_{n} (1, s) / n = \left( N_{n} (1, s) / N_{n} (s) \right) \cdot
\left( N_{n} (s) / n \right)\) so that
\begin{align*}
  \frac{N_{n} (1, s)}{n} - \pi (s) Q_{0} (\mathbb{S} (Z) = s) =
  & \ \frac{N_{n} (1, s)}{N_{n} (s)} \cdot \frac{N_{n} (s)}{n} - \pi (s)
    Q_{0} (\mathbb{S} (Z) = s) \\
  =
  & \ \left[ \frac{N_{n} (1, s)}{N_{n} (s)} - \pi (s) \right] Q_{0}
    (\mathbb{S} = s) \\
  & + \frac{N_{n} (1, s)}{N_{n} (s)} \left[ \frac{N_{n} (s)}{n} - Q_{0}
    (\mathbb{S} (Z) = s) \right].
\end{align*}
The first tends to zero in probability by \th\ref{asm--treat-strat}
\ref{asm--treat-prop}. The second tends to zero almost surely, since \(0
\leq N_{n} (1, s) / N_{n} (s) \leq 1\) and the subsequent term in the product
converges almost surely to zero by Kolmogorov's Second Strong Law of Large
Numbers.
\end{proof}

\begin{lemma}
\th\label{lem--car-vector-clt}
Let \(J \in \mathbb{N}\), \(\xi : \mathbb{N}_{\mathcal{S}} \to
\mathbb{R}^{J}\) satisfy
\begin{equation*}
  \mathbb{E} [\xi (\mathbb{S} (Z))] = \mathbf{0}_{J}
  \quad \text{and} \quad \mathbb{E} \left[ \|\xi (\mathbb{S} (Z))\|^{2} \right]
  < \infty.
\end{equation*}
In addition, let \(h : \mathbb{R} \times \{0, 1\} \times \mathbb{R}^{k} \times
\mathbb{N}_{\mathcal{S}} \to \mathbb{R}^{J}\) be a Borel measurable function
such that for all \(s \in \mathbb{N}_{\mathcal{S}}\) and \(a \in \{0, 1\}\),
\begin{equation*}
  \mathbb{E} [h (Y (a), a, Z) | \mathbb{S} (Z) = s] =
  \mathbf{0}_{J} \quad \text{and} \quad \mathbb{E} \left[ \| h (Y (a), a, Z
  \|_{2}^{2} \middle| \mathbb{S} (Z) = s \right] < \infty
\end{equation*}
Define \(\eta_{n i} = h \left( Y_{n i}, A_{n i}, Z_{i}, \mathbb{S} \left( Z_{i}
\right) \right)\), \(\xi_{i} = \xi \left( \mathbb{S} \left( Z_{i} \right)
\right)\), \(T_{n} = \frac{1}{\sqrt{n}} \sum_{i = 1}^{n} \eta_{n i}\),
\(\zeta_{n} = \frac{1}{\sqrt{n}} \sum_{i = 1}^{n} \xi_{i}\). Under
\th\ref{asm--Q,asm--iid-Q,asm--treat-strat},
\begin{equation}
  \begin{split}
    T_{n} + \zeta_{n} \overset{d}{\to}
    & \ \mathcal{N} \left( \mathbf{0}_{J}, \mathbb{V}_{h} (\pi) +
      \mathbb{V}_{\xi} \right) \\
    \text{where} \quad \mathbb{V}_{h} (\pi) =
    & \ \sum_{s = 1}^{\mathcal{S}} Q_{0} (\mathbb{S} (Z) = s) \left\{
      \begin{array}{l}
        \pi (s) \mathrm{Var} [h (Y (1), 1, Z) | \mathbb{S} (Z) = s] \\
        + (1 - \pi (s)) \mathrm{Var} [h (Y (0), 0, Z) | \mathbb{S} (Z) = s]
      \end{array}
      \right\} \\
    =
    & \ \mathbb{E} \left[ \pi (\mathbb{S} (Z)) \mathrm{Var} [h (Y (1), 1, Z) |
      \mathbb{S} (Z)] \right] \\
    & + \mathbb{E} \left[ (1 - \pi (\mathbb{S} (Z))) \mathrm{Var} [h (Y (0), 0,
      Z) | \mathbb{S} (Z)] \right] \\
    \text{and} \quad \mathbb{V}_{\xi} =
    & \ \mathrm{Var} [\xi (\mathbb{S} (Z))].
  \end{split}
  \label{eqn--car-vector-clt}
\end{equation}
\end{lemma}

\begin{proof}
Note that \(T_{n}\) is numerically invariant to permutations of the sample
indices \(i \in \mathbb{N}_{n}\). Hence, we can take any permutation that places
observations in increasing order with respect to stratum labels and for each \(s
\in \mathbb{N}_{\mathcal{S}}\) places \(\left\{ i \in \mathbb{N}_{n} : A_{n i} =
1, \mathbb{S} \left( Z_{i} \right) = s \right\}\) before \(\left\{ i \in
\mathbb{N}_{n} : A_{n i} = 0, \mathbb{S} \left( Z_{i} \right) = s
\right\}\). Denote \(\eta_{i} (a, s) = h \left( Y_{i} (a, s), a, Z_{i} (s)
\right)\). Then, defining
\begin{equation}
  \widetilde{T}_{n} = \sum_{s = 1}^{\mathcal{S}} \left\{ \frac{1}{\sqrt{n}}
  \sum_{i = \widetilde{N}_{n} (s) + 1}^{\widetilde{N}_{n} (s) + N_{n} (1, s)}
  \eta_{i} (1, s) + \frac{1}{\sqrt{n}} \sum_{i = \widetilde{N}_{n} (s) + N_{n}
  (1, s) + 1}^{\widetilde{N}_{n} (s) + N_{n} (s)} \eta_{i} (0, s) \right\},
  \label{eqn--car-indep-rep}
\end{equation}
it follows from \th\ref{lem--coupling-lemma} that \(T_{n} + \zeta_{n} |
\mathbf{A}_{n}, \mathbf{S}_{n} \overset{\mathrm{d}}{=} \widetilde{T}_{n} +
\zeta_{n} | \mathbf{A}_{n}, \mathbf{S}_{n}\) which further implies \(T_{n} +
\zeta_{n} \overset{\mathrm{d}}{=} \widetilde{T}_{n} + \zeta_{n}\). Thus, finding
a limit distribution for \(T_{n} + \zeta_{n}\) is equivalent to finding one for
\(\widetilde{T}_{n} + \zeta_{n}\).

\(\widetilde{T}_{n}\) can be expressed in terms of a partial sums process
which allows us to derive its limit distribution from weak convergence of the
partial sums process. Indeed, for \(u \in [0, 1]\), define
\begin{equation}
  \begin{split}
    T_{n}^{\ast} (u, a, s) =
    & \ \frac{1}{\sqrt{n}} \sum_{i = 1}^{\lfloor n u \rfloor} \eta_{i} (a, s),
    \\
    \mathbf{T}_{n}^{\ast} (u)^{\prime} =
    & \ \left( T_{n}^{\ast} (u, 0, 1)^{\prime}, T_{n}^{\ast} (u, 1, 1)^{\prime},
      \dots, T_{n}^{\ast} (u, 0, \mathcal{S})^{\prime}, T_{n}^{\ast} (u, 1,
      \mathcal{S})^{\prime} \right).
  \end{split}
  \label{eqn--car-partial-sum}
\end{equation}
Thus, \(\mathbf{T}_{n}^{\ast} (\cdot)\) is a \(2 \mathcal{S} k\)-vector valued
process with sample paths in the space \(\ell^{\infty} ([0, 1])^{2 \mathcal{S}
k}\). Furthermore, it follows from \eqref{eqn--car-indep-rep} and
\eqref{eqn--car-partial-sum} that we can write \(\widetilde{T}_{n}\) as
\begin{equation}
  \widetilde{T}_{n} = \sum_{s = 1}^{\mathcal{S}} \left\{
  \begin{array}{l}
    T_{n}^{\ast} \left( \frac{\widetilde{N}_{n} (s)}{n} + \frac{N_{n} (1,
    s)}{n}, 1, s \right) - T_{n}^{\ast} \left( \frac{\widetilde{N}_{n} (s)}{n},
    1, s \right) \\
    + T_{n}^{\ast} \left( \frac{\widetilde{N}_{n} (s)}{n} + \frac{N_{n} (s)}{n},
    0, s \right) - T_{n}^{\ast} \left( \frac{\widetilde{N}_{n} (s)}{n} +
    \frac{N_{n} (1, s)}{n}, 0, s \right)
  \end{array}
  \right\}.
  \label{eqn--car-partial-sum-rep}
\end{equation}
Thus, \(\widetilde{T}_{n}\) is a linear combination of ``random
point-evaluations'' of the partial-sums process \(\mathbf{T}_{n}^{\ast}
(\cdot)\). We proceed by first deriving a weak limit for \(\left(
\mathbf{T}_{n}^{\ast \prime} (\cdot), \zeta_{n}^{\prime}
\right)^{\prime}\). Note that by \th\ref{asm--coupling-construct},
\(\mathbf{T}_{n}^{\ast}\) and \(\zeta_{n}\) are independent for every \(n \in
\mathbb{N}\). So, we can argue their weak limits separately (see for instance
Example 1.4.6 of \citet[p. 31]{1996vandervaartWeakConvergenceEmpirical}). Then
we combine this weak limit with the asymptotic behavior of the evaluation points
defining \(\widetilde{T}_{n}\) via the Continuous Mapping Theorem to derive the
limit distribution of \(\widetilde{T}_{n} + \zeta_{n}\).

The limit distribution of \(\zeta_{n}\) is straightforward from the
Lindeberg-L\'evy Central Limit Theorem:
\begin{equation}
  \zeta_{n} \overset{\mathrm{d}}{\to} \zeta \sim \mathcal{N} \left(
  \mathbf{0}_{J}, \mathbb{V}_{\xi} \right).
  \label{eqn--car-vector-clt-zeta}
\end{equation}
Next, we consider \(\mathbf{T}_{n}^{\ast} (\cdot)\). By hypothesis, \(\left(
\eta_{i} (0, s), \eta_{i} (1, s) \right)\) are mean zero, finite variance and
independent across \(s \in \mathbb{N}_{\mathcal{S}}\). By a \((2 \mathcal{S}
k)\)-dimensional variant of Donsker's Functional Central Limit Theorem (see for
instance \citet[Theorem 4.3.5, p. 106]{2002whittStochasticProcessLimits}),
\(\mathbf{T}_{n}^{\ast} (\cdot)\) converges weakly in \(\ell^{\infty} ([0,
1])^{2 \mathcal{S} k}\) to a \((2 \mathcal{S} k)\)-vector valued scaled Brownian
motion \(\mathbf{T}^{\ast} (\cdot)\)
\begin{align*}
  \mathbf{T}^{\ast} (\cdot)^{\prime} =
  & \ \left( T^{\ast} (\cdot, 0, 1)^{\prime}, T^{\ast} (\cdot, 1, 1)^{\prime},
    \dots, T^{\ast} (\cdot, 0, \mathcal{S})^{\prime}, T^{\ast} (\cdot, 1,
    \mathcal{S})^{\prime} \right) \\
  \left(
  \begin{array}{c}
    T^{\ast} (\cdot, 0, s) \\
    T^{\ast} (\cdot, 1, s)
  \end{array} \right) =
  & \ \Sigma_{h, s}^{\frac{1}{2}} B_{s} (\cdot)
\end{align*}
where \(\left( B_{1} (\cdot)^{\prime}, \dots, B_{\mathcal{S}} (\cdot)^{\prime}
\right)^{\prime}\) is a \((2 \mathcal{S} k)\)-vector valued standard Brownian
motion and
\begin{equation*}
  \Sigma_{h, s} = \ \left(
  \begin{array}{cc}
    \Sigma_{h, s} (0, 0) & \Sigma_{h, s} (0, 1) \\
    \Sigma_{h, s} (1, 0) & \Sigma_{h, s} (1, 1)
  \end{array}
  \right) = \left(
  \begin{array}{cc}
    \mathbb{E} \left[ \eta_{i} (0, s) \eta_{i} (0, s)^{\prime} \right]
    & \mathbb{E} \left[ \eta_{i} (0, s) \eta_{i} (1, s)^{\prime} \right] \\
    \mathbb{E} \left[ \eta_{i} (1, s) \eta_{i} (0, s)^{\prime} \right]
    & \mathbb{E} \left[ \eta_{i} (1, s) \eta_{i} (1, s)^{\prime} \right]
  \end{array}
  \right).
\end{equation*}
Denoting weak convergence in the product space \(\ell^{\infty} ([0, 1])^{2
\mathcal{S} J} \times \mathbb{R}^{J}\) by \(\rightsquigarrow\), we have
\begin{equation*}
  \left(
  \begin{array}{c}
    \mathbf{T}_{n}^{\ast} (\cdot) \\
    \zeta_{n}
  \end{array}
  \right) \rightsquigarrow \left(
  \begin{array}{c}
    \mathbf{T}^{\ast} \\
    \zeta
  \end{array}
  \right).
\end{equation*}

Next, we deal with asymptotic behavior of the evaluation points defining
\(\widetilde{T}_{n}\) in terms of \(\mathbf{T}^{\ast}_{n} (\cdot)\) in
\eqref{eqn--car-partial-sum-rep}. By Kolmogorov's Second Strong Law of Large
Numbers, \(n^{- 1} N_{n} (s) \overset{\mathrm{a.s.}}{\to} Q_{0} (\mathbb{S} (Z)
= s)\) and hence also in probability for each \(s \in
\mathbb{N}_{\mathcal{S}}\). So by the Continuous Mapping Theorem, it is
straightforward that \(n^{- 1} \widetilde{N}_{n} (s)
\overset{\mathrm{a.s.}}{\to} F (s)\) and hence also in probability. For each
\((a, s) \in \{0, 1\} \times \mathbb{N}_{\mathcal{S}}\), letting \(\pi_{a} (s) =
\pi (s)^{a} \cdot (1 - \pi (s))^{1 - a} \), by
\th\ref{lem--treat-strat-sample-prop-conv},
\begin{equation*}
  \frac{N_{n} (a, s)}{n} = \frac{N_{n} (a, s)}{N_{n} (s)} \cdot \frac{N_{n}
  (s)}{n} \overset{\mathrm{p}}{\to} \pi_{a} (s) \cdot Q_{0} (\mathbb{S} (Z) = s)
  := \zeta (a, s) > 0.
\end{equation*}
Thus the evaluation points defining \(\widetilde{T}_{n}\) in terms of
\(\mathbf{T}^{\ast}_{n} (\cdot)\) all converge in probability to non-stochastic
limits. A standard property of Brownian motions on \([0, 1]\) is that their
sample paths are almost-surely uniformly continuous. Then,
denoting weak convergence in the product space \(\ell^{\infty} ([0, 1])^{2
\mathcal{S} J} \times \mathbb{R}^{J + 2 \mathcal{S}}\) again by
\(\rightsquigarrow\), (v) of Theorem 18.10 in
\citet[p. 259]{1998vandervaartAsymptoticStatistics} implies that
\begin{equation*}
  \left(
  \begin{array}{c}
    \mathbf{T}_{n}^{\ast} (\cdot) \\
    \zeta_{n} \\
    \left( \frac{N_{n} (a, s)}{n} : s \in \mathbb{N}_{\mathcal{S}}, a \in \{0,
    1\} \right)
  \end{array}
  \right) \rightsquigarrow \left(
  \begin{array}{c}
    \mathbf{T}^{\ast} \\
    \zeta \\
    \left( \pi_{a} (s) Q_{0} (\mathbb{S} (Z) = s) : s \in
    \mathbb{N}_{\mathcal{S}}, a \in \{0, 1\} \right)
  \end{array}
  \right)
  \label{eqn--car-vector-clt-zeta-T}
\end{equation*}
with \(\mathbf{T}^{\ast}\) and \(\zeta\) independent, and the remaining
components all non-stochastic. Therefore, by Theorem 18.11 in
\citet[p. 259]{1998vandervaartAsymptoticStatistics},
\begin{align}
  \widetilde{T}_{n} + \zeta_{n} =
  & \ \sum_{s = 1}^{\mathcal{S}} \left\{
    \begin{array}{l}
      T_{n}^{\ast} \left( \frac{\widetilde{N}_{n} (s)}{n} + \frac{N_{n} (1,
      s)}{n}, 1, s \right) - T_{n}^{\ast} \left( \frac{\widetilde{N}_{n}
      (s)}{n}, 1, s \right) \\
      + T_{n}^{\ast} \left( \frac{\widetilde{N}_{n} (s)}{n} + \frac{N_{n}
      (s)}{n}, 0, s \right) - T_{n}^{\ast} \left( \frac{\widetilde{N}_{n}
      (s)}{n} + \frac{N_{n} (1, s)}{n}, 0, s \right)
    \end{array}
    \right\} + \zeta_{n} \nonumber \\
  \overset{\mathrm{d}}{\to}
  & \ \sum_{s = 1}^{\mathcal{S}} \left\{
    \begin{array}{l}
      T^{\ast} (F (s) + \pi (s) Q_{0} (\mathbb{S} (Z) = s), 1, s) - T^{\ast} (F
      (s), 1, s) \\
      + T^{\ast} (F (s) + Q_{0} (\mathbb{S} (Z) = s), 0, s) - T^{\ast} (F (s) +
      \pi (s) Q_{0} (\mathbb{S} (Z) = s), 0, s)
    \end{array}
    \right\} + \zeta \nonumber \\
  \equiv
  & \ \mathcal{T} + \zeta.
    \label{eqn--car-vector-clt-limit}
\end{align}

Note that the limit random vector \(\mathcal{T}\) in
\eqref{eqn--car-vector-clt-limit} is a linear combination of random vectors that
are all jointly Gaussian. The limit in \eqref{eqn--car-vector-clt-limit} is mean
zero since \(\mathbf{T}^{\ast}\) is mean zero. The final step is then showing
that the variance matrix of \(\mathcal{T}\) is equal to
\eqref{eqn--car-vector-clt}. For that, we can use the independent increments
property of Brownian motions. Independently of each other,
\begin{align*}
  & T^{\ast} (F (s) + \pi (s) Q_{0} (\mathbb{S} (Z) = s), 1, s) - T^{\ast} (F
    (s), 1, s) \sim \mathcal{N} \left( \mathbf{0}_{J}, \mathbb{V}_{1 h} (s)
    \right), \\
  & \text{where } \mathbb{V}_{1 h} (s) = \pi (s) Q_{0} (\mathbb{S} (Z) = s)
    \Sigma_{h, s} (1, 1), \\
  & T^{\ast} (F (s) + Q_{0} (\mathbb{S} (Z) = s), 0, s) - T^{\ast} (F (s) + \pi
    (s) Q_{0} (\mathbb{S} (Z) = s), 0, s) \sim \mathcal{N} \left(
    \mathbf{0}_{J}, \mathbb{V}_{2 h} (s) \right), \\
  & \text{where } \mathbb{V}_{2 h} (s) = (1 - \pi (s)) Q_{0} (\mathbb{S} (Z) =
    s) \Sigma_{h, s} (0, 0).
\end{align*}
By the independent increments property of Brownian motions, the variance of
\(\mathcal{T}\) is the sum of the above variances. Since \(\zeta\) is
independent to \(\mathbf{T}^{\ast} (\cdot)\), it is necessarily independent to
\(\mathcal{T}\). Therefore, putting everything together,
\begin{equation*}
  T_{n} + \zeta_{n} \overset{\mathrm{d}}{=} \widetilde{T}_{n} + \zeta_{n}
  \overset{\mathrm{d}}{\to} \mathcal{T} \sim \mathcal{N} \left( \mathbf{0}_{J},
  \mathbb{V}_{h} (\pi) + \mathbb{V}_{\xi} \right).
\end{equation*}
\end{proof}

\begin{lemma}
\th\label{lem--car-vector-lln}
Let \(J \in \mathbb{N}\) and \(h : \mathbb{R} \times \{0, 1\} \times
\mathbb{R}^{k} \to \mathbb{R}^{J}\) be a Borel measurable function such that for
any \(a \in \{0, 1\}\), \(\mathbb{E} \left[ \Vert h (Y (a), a, Z) \Vert_{1}
\right] < \infty\). Define \(\eta_{n i} = h \left( Y_{n i}, A_{n i}, Z_{i}
\right)\) and \(\overline{\eta}_{n} = \frac{1}{n} \sum_{i = 1}^{n} \eta_{n
i}\). Then under \th\ref{asm--Q,asm--iid-Q,asm--treat-strat},
\(\overline{\eta}_{n} \overset{\mathrm{p}}{\to} \mathbb{E} \left[ \pi
(\mathbb{S} (Z)) h (Y (1), 1, Z) + (1 - \pi (\mathbb{S} (Z))) h (Y (0), 0, Z)
\right]\).
\end{lemma}

\begin{proof}
Notice that by \th\ref{lem--coupling-lemma}
\begin{equation*}
  \overline{\eta}_{n} \overset{\mathrm{d}}{=} \sum_{s = 1}^{\mathcal{S}} \sum_{a
  = 0}^{1} \frac{1}{n} \sum_{i = 1}^{N_{n} (a, s)} \eta_{i} (a, s) = \sum_{s =
  1}^{\mathcal{S}} \sum_{a = 0}^{1} \frac{N_{n} (a, s)}{n} \frac{1}{N_{n} (a,
  s)} \sum_{i = 1}^{N_{n} (a, s)} \eta_{i} (a, s).
\end{equation*}
By Kolmogorov's Second Strong Law of Large Numbers, it follows that for
each \((a, s) \in \{0, 1\} \times \mathbb{N}_{\mathcal{S}}\),
\begin{equation*}
  \frac{1}{m} \sum_{i = 1}^{m} \eta_{i} (a, s) \overset{\mathrm{a.s.}}{\to}
  \mathbb{E} [h (Y (a), a, Z) | \mathbb{S} (Z) = s] \qquad \text{as } m \to
  \infty
\end{equation*}
and hence also in probability. By Kolmogorov's Second Strong Law of Large
Numbers, for each \(s \in \mathbb{N}_{\mathcal{S}}\), \(n^{- 1} N_{n} (s)
\overset{\mathrm{a.s.}}{\to} Q_{0} (\mathbb{S} (Z) = s)\) and hence also in
probability. For each \((a, s) \in \{0, 1\} \times \mathbb{N}_{\mathcal{S}}\),
by \th\ref{lem--treat-strat-sample-prop-conv},
\begin{equation*}
  \frac{N_{n} (a, s)}{n} = \frac{N_{n} (a, s)}{N_{n} (s)} \cdot \frac{N_{n}
  (s)}{n} \overset{\mathrm{p}}{\to} \pi (s)^{a} \cdot (1 - \pi (s))^{1 - a}
  \cdot Q_{0} (\mathbb{S} (Z) = s) := \zeta (a, s) > 0.
\end{equation*}
By \th\ref{lem--ratio-conv-implies-div}, \(N_{n} (a, s)
\overset{\mathrm{p}}{\to} \infty\) for every \((a, s) \in \{0, 1\} \times
\mathbb{N}_{\mathcal{S}}\) in the sense of
\th\ref{def--stochastic-divergence}. By \th\ref{lem--random-subseq-conv}, we
have
\begin{equation*}
  \frac{1}{N_{n} (a, s)} \sum_{i = 1}^{N_{n} (a, s)} \eta_{i} (a, s)
  \overset{\mathrm{p}}{\to} \mathbb{E} [h (Y (a), a, Z) | \mathbb{S} (Z) = s].
\end{equation*}
Applying the Continuous Mapping Theorem, we have
\begin{align*}
  \overline{\eta}_{n} \overset{\mathrm{p}}{\to}
  & \ \sum_{s = 1}^{\mathcal{S}} Q_{0} (\mathbb{S} (Z) = s) \sum_{a = 0}^{1} \pi
    (s)^{a} \cdot (1 - \pi (s))^{1 - a} \cdot \mathbb{E} [h (Y (a), a, Z) |
    \mathbb{S} (Z) = s] \\
  =
  & \ \mathbb{E} [\pi (\mathbb{S} (Z)) h (Y (1), 1, Z) + (1 - \pi (\mathbb{S}
    (Z))) h (Y (0), 0, Z)]
\end{align*}
where the last equality follows by the Law of Iterated Expectations.
\end{proof}

\begin{lemma}
\th\label{lem--sample-split-tends-correct}
Let \th\ref{asm--Q,asm--iid-Q,asm--treat-strat} hold. Under
\th\ref{asm--sample-split}, for any given \(s \in \mathbb{N}_{\mathcal{S}}\)
\begin{align}
  \frac{\widehat{N}_{n} (a, s, j)}{n} \overset{\mathrm{p}}{\to}
  & \ \frac{\pi_{a} (s)}{J} Q_{0} (\mathbb{S} (Z) = s)
    \label{eqn--strat-treat-fold-size-correct}
  \\
  \frac{\widehat{N}_{n} (s, j)}{n} \overset{\mathrm{p}}{\to}
  & \ \frac{1}{J} Q_{0} (\mathbb{S} (Z) = s)
    \label{eqn--strat-fold-size-correct}
\end{align}
with \(\widehat{N}_{n} (a, s, j)\) and \(\widehat{N}_{n} (s, j)\) defined in
\eqref{eqn--strat-treat-fold-size} and \eqref{eqn--strat-fold-size}
respectively. If in addition, \th\ref{asm--weak-balance} holds, then
\begin{equation}
  \left| \frac{\widehat{N}_{n} (a, s, j)}{n} - \frac{\pi_{a} (s)}{J} Q_{0}
  (\mathbb{S} (Z) = s) \right| = O_{\mathrm{p}} \left( 1 / \sqrt{n} \right).
  \label{eqn--strat-treat-fold-weak-balance}
\end{equation}
Furthermore, under \th\ref{asm--weak-balance}, for each \(s \in
\mathbb{N}_{\mathcal{S}}\) such that \(Q_{0} (\mathbb{S} (Z) = s) > 0\), the
above implies that
\begin{equation}
  \left| \frac{\widehat{N}_{n} (a, s, j)}{\widehat{N}_{n} (s, j)} - \pi_{a}
  (s) \right| = O_{\mathrm{p}} \left( 1 / \sqrt{n} \right).
  \label{eqn--strat-treat-within-fold-weak-balance}
\end{equation}
\end{lemma}

\begin{proof}[Proof of \th\ref{lem--sample-split-tends-correct}]
We start by proving \eqref{eqn--strat-treat-fold-size-correct}. First,
\begin{equation}
  \begin{split}
    \frac{1}{n} \left\lfloor \frac{N_{n} (a, s)}{J} \right\rfloor -
    \frac{\pi_{a} (s) Q_{0} (\mathbb{S} (Z) = s)}{J} =
    & \ \frac{1}{J} \left[ \frac{N_{n} (a, s)}{n} - \pi_{a} (s) Q_{0}
      (\mathbb{S} (Z) = s) \right] \\
    & + \frac{N_{n} (a, s)}{n} \left\{ \frac{1}{N_{n} (a, s)} \left\lfloor
      \frac{N_{n} (a, s)}{J} \right\rfloor - \frac{1}{J} \right\}
  \end{split}
  \label{eqn--sample-split-tends-correct-1}
\end{equation}
By definition of the integer floor function \(\lfloor \cdot \rfloor\), we have
\begin{equation*}
  \left\lfloor \frac{N_{n} (a, s)}{J} \right\rfloor \leq \frac{N_{n} (a, s)}{J}
  < \left\lfloor \frac{N_{n} (a, s)}{J} \right\rfloor + 1.
\end{equation*}
and so, taking absolute values on both sides of
\eqref{eqn--sample-split-tends-correct-1} and applying the triangle inequality,
and using the above,
\begin{equation}
  \left| \frac{1}{n} \left\lfloor \frac{N_{n} (a, s)}{J} \right\rfloor -
  \frac{\pi_{a} (s) Q_{0} (\mathbb{S} (Z) = s)}{J} \right| \leq
  \frac{1}{J} \left| \frac{N_{n} (a, s)}{n} - \pi_{a} (s) Q_{0} (\mathbb{S} (Z)
  = s) \right| + \frac{1}{n}
  \label{eqn--sample-split-treat-prop-breakdown}
\end{equation}
Applying \th\ref{lem--treat-strat-sample-prop-conv}, we have
\begin{equation}
  \frac{1}{n} \left\lfloor \frac{N_{n} (a, s)}{J} \right\rfloor
  \overset{\mathrm{p}}{\to} \frac{\pi_{a} (s) Q_{0} (\mathbb{S} (Z) = s)}{J}.
  \label{eqn--sample-split-tends-correct-2}
\end{equation}
Next, by \th\ref{asm--sample-split} \ref{asm--sample-split-size},
\eqref{eqn--strat-treat-fold-size-correct} follows for \(j \in \mathbb{N}_{J -
1}\) since
\begin{equation}
  \widehat{N}_{n} (a, s, j) = \left\lfloor \frac{N_{n} (a, s)}{J} \right\rfloor
\end{equation}
For \(j = J\), we have
\begin{align*}
  \frac{\widehat{N}_{n} (a, s, J)}{n} =
  & \ \frac{N_{n} (a, s)}{n} - \frac{(J - 1)}{n} \cdot \left\lfloor \frac{N_{n}
    (a, s)}{J} \right\rfloor \\
  =
  & \ \frac{N_{n} (a, s)}{J n} + (J - 1) \left[ \frac{N_{n} (a, s)}{J n} -
    \frac{1}{n} \left\lfloor \frac{N_{n} (a, s)}{J} \right\rfloor \right].
\end{align*}
The full conclusion of \eqref{eqn--strat-treat-fold-size-correct} then follows
from \eqref{eqn--sample-split-tends-correct-2}, the above and the Continuous
Mapping Theorem. The same reasoning yields \eqref{eqn--strat-fold-size-correct}
since \(\pi (s) + (1 - \pi (s)) = 1\) and
\begin{equation*}
  \widehat{N}_{n} (s, j) = \widehat{N}_{n} (1, s, j) + \widehat{N}_{n} (0, s,
  j).
\end{equation*}

Next, we prove \eqref{eqn--strat-treat-fold-weak-balance}. Scaling by
\(\sqrt{n}\) in \eqref{eqn--sample-split-treat-prop-breakdown} yields the following
\begin{align*}
  \sqrt{n} \left| \frac{1}{n} \left\lfloor \frac{N_{n} (a, s)}{J} \right\rfloor
  - \frac{\pi_{a} (s) Q_{0} (\mathbb{S} (Z) = s)}{J} \right| \leq
  & \ \frac{\sqrt{n}}{J} \left| \frac{N_{n} (a, s)}{n} - \pi_{a} (s) Q_{0}
    (\mathbb{S} (Z) = s) \right| + \frac{1}{\sqrt{n}} \\
  \leq
  & \ \frac{\sqrt{n}}{J} \left| \frac{N_{n} (a, s)}{N_{n} (s)} - \pi_{a} (s)
    \right| \frac{N_{n} (s)}{n} \\
  & + \frac{\pi_{a} (s) \sqrt{n}}{J} \left|
    \frac{N_{n} (a)}{n} - Q_{0} (\mathbb{S} (Z) = s) \right| +
    \frac{1}{\sqrt{n}} \\
  = & \ O_{\mathrm{p}} (1) \cdot O_{\mathrm{p}} (1) + O_{\mathrm{p}} (1) +
      o_{\mathrm{p}} (1) = O_{p} (1)
\end{align*}
In the above, the first term being a product of \(O_{\mathrm{p}} (1)\)'s follows
from \th\ref{asm--weak-balance} and Kolmogorov's Second Strong Law of Large
Numbers. The second term being a \(O_{p} (1)\) term follows from the
Lindeberg-L\'evy Central Limit Theorem. Then,
\eqref{eqn--strat-treat-fold-weak-balance} follows for \(j \in \mathbb{N}_{J -
1}\) immediately as noted before. For \(j = J\), as we have already argued
before,
\begin{equation*}
  \frac{\widehat{N}_{n} (a, s, J)}{n} = \frac{N_{n} (a, s)}{J n} + (J - 1)
  \left[ \frac{N_{n} (a, s)}{J n} - \frac{1}{n} \left\lfloor N_{n} (a, s) / J
  \right\rfloor \right].
\end{equation*}
Therefore,
\begin{align*}
  \sqrt{n} \left| \frac{\widehat{N}_{n} (a, s, J)}{n} - \frac{\pi_{a} (s)}{J}
  Q_{0} (\mathbb{S} (Z) = s) \right| \leq
  & \ \frac{\sqrt{n}}{J} \left| \frac{N_{n} (a, s)}{n} - \pi_{a} (s) Q_{0}
    (\mathbb{S} (Z) = s) \right| \\
  & + \frac{(J - 1)}{\sqrt{n}} \left| \frac{N_{n} (a, s)}{J} -
    \left\lfloor \frac{N_{n} (a, s)}{J} \right\rfloor \right| \\
  \leq
  & \frac{\sqrt{n}}{J} \left| \frac{N_{n} (a, s)}{n} - \pi_{a} (s) Q_{0}
    (\mathbb{S} (Z) = s) \right| + \frac{(J - 1)}{\sqrt{n}} \\
  = & \ O_{\mathrm{p}} (1) + o (1) = O_{\mathrm{p}} (1).
\end{align*}
In the above, the fact that the first term is \(O_{\mathrm{p}} (1)\) follows
from our previous arguments. Therefore,
\eqref{eqn--strat-treat-fold-weak-balance} follows for \(j = J\) as well.

For \eqref{eqn--strat-treat-within-fold-weak-balance}, first note that
\begin{align*}
  \frac{\widehat{N}_{n} (s, j)}{n} - \frac{Q_{0} (\mathbb{S} (Z) = s)}{J} =
  & \ \frac{\widehat{N}_{n} (a, s, j)}{n} - \frac{\pi_{a} (s) Q_{0}
    (\mathbb{S} (Z) = s)}{J} \\
  & + \frac{\widehat{N}_{n} (1 - a, s, j)}{n} - \frac{\pi_{1 - a} (s) Q_{0}
    (\mathbb{S} (Z) = s)}{J} \\
  = & \ O_{\mathrm{p}} \left( \frac{1}{\sqrt{n}} \right) + O_{\mathrm{p}} \left(
      \frac{1}{\sqrt{n}} \right) \\
  = & \ O_{\mathrm{p}} \left( \frac{1}{\sqrt{n}} \right).
\end{align*}
so that
\begin{equation}
  \frac{\widehat{N}_{n} (s, j)}{n} - \frac{Q_{0} (\mathbb{S} (Z) = s)}{J} =
  O_{\mathrm{p}} \left( \frac{1}{\sqrt{n}} \right)
\end{equation}
Next,
\begin{align*}
  \frac{\widehat{N}_{n} (a, s, j)}{\widehat{N}_{n} (s, j)} - \pi_{a} (s) =
  & \ \frac{n}{\widehat{N}_{n} (s, j)} \left[ \frac{\widehat{N}_{n} (a, s,
    j)}{n} - \pi_{a} (s) \frac{\widehat{N}_{n} (s, j)}{n} \right] \\
  =
  & \ \frac{n}{\widehat{N}_{n} (s, j)} \left[ \frac{\widehat{N}_{n} (a, s,
    j)}{n} - \frac{\pi_{a} (s) Q_{0} (\mathbb{S} (Z) = s)}{J} \right] \\
  & + \pi_{a} (s) \frac{n}{\widehat{N}_{n} (s, j)} \left[ \frac{Q_{0}
    (\mathbb{S} (Z) = s)}{J} - \frac{\widehat{N}_{n} (s, j)}{n} \right] \\
  =
  & \ \frac{1}{\frac{Q_{0} (\mathbb{S} (Z) = s)}{J} + O_{\mathrm{p}} \left(
    \frac{1}{\sqrt{n}} \right)} \cdot O_{\mathrm{p}} \left( \frac{1}{\sqrt{n}}
    \right)
    \\
  & + \frac{\pi_{a} (s)}{\frac{Q_{0} (\mathbb{S} (Z) = s)}{J} + O_{\mathrm{p}}
    \left( \frac{1}{\sqrt{n}} \right)} \cdot O_{\mathrm{p}} \left(
    \frac{1}{\sqrt{n}} \right)
  \\
  = & \ O_{\mathrm{p}} \left( \frac{1}{\sqrt{n}} \right)
\end{align*}
when \(Q_{0} (\mathbb{S} (Z) = s) > 0\).
\end{proof}

\subsection{Convergence along randomly indexed subsequences}

\begin{definition}[Deterministic Divergence]
\th\label{def--deterministic-divergence}
Let \(\left\{ m_{n} \right\}_{n \in \mathbb{N}}\) be a sequence of non-negative
numbers. We say \(\left\{ m_{n} \right\}_{n \in \mathbb{N}}\) diverges to
\(\infty\), denoted \(m_{n} \to \infty\) if for every \(M > 0\), there is
\(N_{M} \in \mathbb{N}\) such that \(m_{n} \geq M\) for all \(n \geq N_{M}\).
\end{definition}

\begin{definition}[Stochastic Divergence]
\th\label{def--stochastic-divergence}
Let \(\left\{ M_{n} \right\}_{n \in \mathbb{N}}\) be a sequence of non-negative
random variables all defined on a probability space \(\left( \Omega,
\mathcal{F}, \Pr \right)\). We say \(M_{n}\) diverges to \(\infty\) almost
surely, denoted \(M_{n} \overset{\mathrm{a.s.}}{\to} \infty\) if
\(\Pr \left( \left\{ \omega \in \Omega : M_{n} (\omega) \to \infty \right\}
\right) = 1\). We say \(M_{n}\) diverges to \(\infty\) in probability, denoted
\(M_{n} \overset{\mathrm{p}}{\to} \infty\) if for any \(G > 0\), \(\lim_{n \to
\infty} \mathbb{P} \left( M_{n} < G \right) = 0\).
\end{definition}

\begin{lemma}[Almost sure divergence implies divergence in probability]
\th\label{lem--div-as-implies-p}
Let \(\left\{ M_{n} \right\}_{n \in \mathbb{N}}\) be a sequence of non-negative
random variables all defined on a probability space \(\left( \Omega,
\mathcal{F}, \Pr \right)\). If \(M_{n} \overset{\mathrm{a.s.}}{\to} \infty\),
then \(M_{n} \overset{\mathrm{p}}{\to} \infty\).
\end{lemma}

\begin{proof}
Let \(E = \left\{ \omega \in \Omega : M_{n} (\omega) \to \infty \right\}\) and
for \(G > 0\), define the event
 \begin{equation*}
  E_{G} = \bigcup_{N \in \mathbb{N}} \bigcap_{n = N}^{\infty} \left\{ \omega \in
  \Omega : M_{n} (\omega) \geq G \right\} = \bigcup_{N \in \mathbb{N}} \left\{
  \omega \in \Omega : \left[ \inf_{n \geq N} M_{n} (\omega) \right] \geq G
  \right\}.
\end{equation*}
By \th\ref{def--deterministic-divergence}, it follows that \(E = \cap_{G > 0}
E_{G}\). By \(M_{n} \overset{\mathrm{a.s.}}{\to} \infty\), we have \(\Pr (E) =
1\). This implies that given any \(G > 0\), \(\Pr \left( E_{G} \right) =
1\). Hence, using continuity of \(\Pr\) under monotone sequences of events,
\begin{equation*}
  \lim_{N \to \infty} \Pr \left( \left\{ \omega \in
  \Omega : \left[ \inf_{n \geq N} M_{n} (\omega) \right] \geq G \right\} \right)
  = \Pr \left( \bigcup_{N \in \mathbb{N}} \left\{ \omega \in \Omega : \left[
  \inf_{n \geq N} M_{n} (\omega) \right] \geq G \right\} \right) =
  \Pr \left( E_{G} \right) = 1
\end{equation*}
Now, given any \(N \in \mathbb{N}\),
\begin{equation*}
  \Pr \left( \left\{ \omega \in \Omega : M_{N} (\omega) \geq G \right\} \right)
  \geq \Pr \left( \left\{ \omega \in \Omega : \left[ \inf_{n \geq N} M_{n}
  (\omega) \right] \geq G \right\} \right).
\end{equation*}
The Sandwich Theorem thus implies that \(\lim_{N \to \infty} \Pr \left( M_{N}
\geq G \right) = 1\).
\end{proof}

\begin{lemma}
\th\label{lem--ratio-conv-implies-div}
Let \(\left\{ M_{n} \right\}_{n \in \mathbb{N}}\) and \(\left\{ N_{n}
\right\}_{n \in \mathbb{N}}\) be sequences of non-negative random variables
defined on a probability space \((\Omega, \mathcal{F}, \Pr)\). Suppose also that
\(N_{n}\) is supported in the set of non-negative integers.
\begin{enumerate}[label=(\alph*)]
\item If \(N_{n} \overset{\mathrm{a.s.}}{\to} \infty\) and for some \(0 < \pi
  < \infty\), \(M_{n} / N_{n} \overset{\mathrm{a.s.}}{\to} \pi > 0\), then
  \(M_{n} \overset{\mathrm{a.s.}}{\to} \infty\).
\item If \(N_{n} \overset{\mathrm{p}}{\to} \infty\) and for some \(0 < \pi
  < \infty\), \(M_{n} / N_{n} \overset{\mathrm{p}}{\to} \pi > 0\), then
  \(M_{n} \overset{\mathrm{p}}{\to} \infty\).
\end{enumerate}
\end{lemma}

\begin{proof}
First assume that \(M_{n} / N_{n} \overset{\mathrm{a.s.}}{\to} \pi\) and \(N_{n}
\overset{\mathrm{a.s.}}{\to} \infty\) as \(n \to \infty\) and denote
\begin{align*}
  A =
  & \ \left\{ \omega \in \Omega : \lim_{n \to \infty} \left[ M_{n} (\omega) /
    N_{n} (\omega) \right] = \pi \right\}, \\
  B =
  & \ \left\{ \omega \in \Omega : \lim_{n \to \infty} N_{n} (\omega) = \infty
    \right\}, \\
  C =
  & \ \left\{ \omega \in \Omega : \lim_{n \to \infty} M_{n} (\omega) = \infty
    \right\}.
\end{align*}
By hypothesis, \(\Pr (A) = 1\) and \(\Pr (B) = 1\), which then implies that
\(\Pr (A \cap B) = 1\). Since \(A \cap B \subseteq C\), \(\Pr (C) = 1\) and
hence \(M_{n} \overset{\mathrm{a.s.}}{\to} \infty\).

Next, assume that \(M_{n} / N_{n} \overset{\mathrm{p}}{\to} \pi\) and \(N_{n}
\overset{\mathrm{p}}{\to} \infty\) as \(n \to \infty\). For each \(n \in
\mathbb{N}\) and \(G > 0\),
\begin{align*}
  \Pr \left( M_{n} < G \right) =
  & \ \Pr \left( M_{n} < G, \left| \frac{M_{n}}{N_{n}} - \pi \right| >
    \frac{\pi}{2} \right) + \Pr \left( M_{n} < G, \left| \frac{M_{n}}{N_{n}} -
    \pi \right| \leq \frac{\pi}{2} \right) \\
  \leq
  & \ \Pr \left( \left| \frac{M_{n}}{N_{n}} - \pi \right| > \frac{\pi}{2}
    \right) + \Pr \left( M_{n} < G, N_{n} \leq \frac{2}{\pi} M_{n} \right) \\
  \leq
  & \ \Pr \left( \left| \frac{M_{n}}{N_{n}} - \pi \right| > \frac{\pi}{2}
  \right) + \Pr \left( N_{n} < \frac{2}{\pi} G \right).
\end{align*}
The first term in the final upper bound tends to zero by \(M_{n} / N_{n}
\overset{\mathrm{p}}{\to} \pi\). The second term in the final upper bound tends
to zero by \(N_{n} \overset{\mathrm{p}}{\to} \infty\). Conclude using the
Sandwich Theorem.
\end{proof}

\begin{lemma}[Convergence for Randomly Indexed Subsequences]
\th\label{lem--random-subseq-conv}
Let \((\Omega, \mathcal{F}, \Pr)\) be a probability space. Let \(\left\{X_{n}
\right\}_{n \in \mathbb{N}}\) be a sequence of random \(\mathbb{R}^{k}\)-vectors
and \(\left\{ M_{n} \right\}_{n \in \mathbb{N}}\) be a sequence of positive
integer-valued random variables all defined on \((\Omega, \mathcal{F},
\Pr)\). Furthermore, let \(X\) be a random \(\mathbb{R}^{k}\)-vector defined on
\((\Omega, \mathcal{F}, \Pr)\). For each \(n \in \mathbb{N}\), define \(Z_{n} :
\Omega \to \mathbb{R}^{k}\) by \(Z_{n} (\omega) = X_{M_{n} (\omega)}
(\omega)\). Then the following hold.
\begin{enumerate}
\item For every \(n \in \mathbb{N}\), \(Z_{n}\) is a measurable function
  mapping \(\Omega\) into \(\mathbb{R}^{k}\) (i.e. \(Z_{n}\) is a random
  \(\mathbb{R}^{k}\)-vector).
\item If \(X_{n} \overset{\mathrm{a.s.}}{\to} X\) and \(M_{n}
  \overset{\mathrm{a.s.}}{\to} \infty\) as \(n \to \infty\), then \(Z_{n}
  \overset{\mathrm{a.s.}}{\to} X\) as \(n \to \infty\).
\item If \(X_{n} \overset{\mathrm{p}}{\to} X\) and \(M_{n}
  \overset{\mathrm{p}}{\to} \infty\) as \(n \to \infty\), then \(Z_{n}
  \overset{\mathrm{p}}{\to} X\) as \(n \to \infty\).
\item If \(X_{n} \overset{L_{p}}{\to} X\) and \(M_{n}
  \overset{\mathrm{p}}{\to} \infty\) as \(n \to \infty\), then \(Z_{n}
  \overset{\mathrm{p}}{\to} X\) as \(n \to \infty\).
\end{enumerate}
\end{lemma}

\begin{proof}
To show measurability, note that for any \(x \in \mathbb{R}^{k}\),
\begin{align*}
  \left\{ \omega \in \Omega : Z_{n} (\omega) \leq x \right\}
  =
  & \ \bigcup_{k \in \mathbb{N}} \left\{ \omega \in \Omega : M_{n} (\omega) =
    k, X_{k} (\omega) \leq x \right\} \\
  =
  & \ \bigcup_{k \in \mathbb{N}} \left[ \left\{ \omega \in \Omega : M_{n}
    (\omega) = k \right\} \cap \left\{\omega \in \Omega : X_{k} (\omega) \leq x
    \right\} \right].
\end{align*}
Measurability of \(Z_{n}\) is then immediate from measurability of \(M_{n}\) and
of \(X_{k}\) for every \(k \in \mathbb{N}\).

Next, assume that \(X_{n} \overset{\mathrm{a.s.}}{\to} X\) and \(M_{n}
\overset{\mathrm{a.s.}}{\to} \infty\) as \(n \to \infty\) and denote
\begin{align*}
  A =
  & \ \left\{ \omega \in \Omega : \lim_{n \to \infty} X_{n} (\omega) = X
    (\omega) \right\}, \\
  B =
  & \ \left\{ \omega \in \Omega : \lim_{n \to \infty} M_{n} (\omega) = \infty
    \right\}, \\
  C =
  & \ \left\{ \omega \in \Omega : \lim_{n \to \infty} Z_{n} (\omega) = X
    (\omega) \right\}.
\end{align*}
By hypothesis, \(\Pr (A) = 1\) and \(\Pr (B) = 1\), which then implies
that \(\Pr (A \cap B) = 1\). Note that \(A \cap B \subseteq C\) so that \(\Pr
(C) = 1\) and hence \(Z_{n} \overset{\mathrm{a.s.}}{\to} X\).

Finally, assume that \(X_{n} \overset{\mathrm{p}}{\to} X\) and \(M_{n}
\overset{\mathrm{p}}{\to} \infty\) as \(n \to \infty\). Since \(X_{n}
\overset{\mathrm{p}}{\to} X\) as \(n \to \infty\), we know that for any
\(\varepsilon > 0\) and \(\delta > 0\), there is \(N (1, \varepsilon, \delta)
\in \mathbb{N}\) such that
\begin{equation*}
  \Pr \left( \left\Vert X_{n} - X \right\Vert > \varepsilon \right) < \delta
  \qquad \forall \ n \in \mathbb{N} \text{ such that } n \geq N (1, \varepsilon,
  \delta).
\end{equation*}
Note that this immediately means that
\begin{equation*}
  \sup_{m \geq n} \Pr \left( \left\Vert X_{m} - X \right\Vert > \varepsilon
  \right) \leq \delta \qquad \forall \ n \in \mathbb{N} \text{ such that } n
  \geq N (1, \varepsilon, \delta).
\end{equation*}
Similarly, since \(M_{n} \overset{\mathrm{p}}{\to} \infty\), for every \(G > 0\)
and \(\delta > 0\), there is \(N (2, G, \delta) \in \mathbb{N}\) such that
\begin{equation*}
  \Pr \left( M_{n} < G \right) < \delta \qquad \forall \ n \in \mathbb{N} \text{
  such that } n \geq N (2, G, \delta).
\end{equation*}
So, let \(\varepsilon, \delta > 0\) be given. Denote the support of \(M_{n}\) by
\(\mathbb{M}_{n}\). For each \(n \in \mathbb{N}\) and \(G > 0\),
\begin{align*}
  \Pr \left( \left\Vert Z_{n} - X \right\Vert > \varepsilon \right) =
  & \ \Pr \left( \left\Vert Z_{n} - X \right\Vert > \varepsilon, M_{n} < G
    \right) + \Pr \left( \left\Vert Z_{n} - X \right\Vert > \varepsilon, M_{n}
    \geq G \right) \\
  \leq
  & \ \Pr \left( M_{n} < G \right) + \Pr \left( \left\Vert Z_{n} - X \right\Vert
    > \varepsilon, M_{n} \geq G \right) \\
  =
  & \ \Pr \left( M_{n} < G \right) + \sum_{m \in \mathbb{M}_{n}, m \geq G} \Pr
    \left( \left\Vert X_{m} - X \right\Vert > \varepsilon \middle| M_{n} = m
    \right) \Pr \left( M_{n} = m \right).
\end{align*}
where the final equality follows by the Law of Total Probability. We normalize
sums over the empty set to be zero so that above is still well defined if there
is no \(m \in \mathbb{M}_{n}\) with \(m \geq G\) for a particular fixed \(n\),
though this should not be problematic asymptotically. Now, for any \(m \in
\mathbb{M}_{n}\) such that \(m \geq G\),
\begin{align*}
  \Pr \left( \left\Vert X_{m} - X \right\Vert > \varepsilon \middle| M_{n} = m
  \right) =
  & \ \frac{\Pr \left( \left\Vert X_{m} - X \right\Vert > \varepsilon, M_{n} = m
    \right)}{\Pr \left( M_{n} = m \right)} \\
  \leq
  & \ \frac{\Pr \left( \left\Vert X_{m} - X \right\Vert > \varepsilon
    \right)}{\Pr \left( M_{n} = m \right)} \\
  \leq
  & \ \frac{\sup_{\kappa \geq G} \Pr \left( \left\Vert X_{\kappa} - X
    \right\Vert > \varepsilon \right)}{\Pr \left( M_{n} = m \right)}.
\end{align*}
Thus, we can conclude that
\begin{align*}
  \Pr \left( \left\Vert Z_{n} - X \right\Vert > \varepsilon \right) \leq
  & \ \Pr \left( M_{n} < G \right) + \sum_{m \in \mathbb{M}_{n}, m \geq G} \Pr
    \left( \left\Vert X_{m} - X \right\Vert > \varepsilon \middle| M_{n} = m
    \right) \Pr \left( M_{n} = m \right) \\
  \leq
  & \ \Pr \left( M_{n} < G \right) + \left[ \sup_{\kappa \geq G} \Pr \left(
    \left\Vert X_{\kappa} - X \right\Vert > \varepsilon \right) \right] \cdot
    \Pr \left( M_{n} \geq G \right) \\
  \leq
  & \ \Pr \left( M_{n} < G \right) + \sup_{\kappa \geq G} \Pr \left( \left\Vert
    X_{\kappa} - X \right\Vert > \varepsilon \right).
\end{align*}
Set \(G_{\varepsilon, \delta} = N \left( 1, \varepsilon, \frac{1}{2} \delta
\right)\) and \(N (\varepsilon, \delta) = N \left( 2, G_{\varepsilon, \delta},
\frac{1}{2} \delta \right)\). Then for all \(n \geq N (\varepsilon, \delta)\),
\begin{align*}
  \Pr \left( \left\Vert Z_{n} - X \right\Vert > \varepsilon \right) \leq
  & \ \Pr \left( M_{n} < G_{\varepsilon, \delta} \right) + \sup_{\kappa \geq
    G_{\varepsilon, \delta}} \Pr \left( \left\Vert X_{\kappa} - X \right\Vert >
    \varepsilon \right) \leq \frac{1}{2} \delta + \frac{1}{2} \delta = \delta.
\end{align*}
This implies \(Z_{n} \overset{\mathrm{p}}{\to} X\).
\end{proof}

\begin{lemma}[Convergence in \(L_{r}\) for Randomly Indexed Subsequences]
\th\label{lem--random-subseq-conv-Lr}
Let \((\Omega, \mathcal{F}, \Pr)\) be a probability space and \(r \in [1,
\infty)\). Let \(\left\{X_{n} \right\}_{n \in \mathbb{N}}\) be a sequence of
random \(\mathbb{R}^{k}\)-vectors and \(\left\{ M_{n} \right\}_{n \in
\mathbb{N}}\) be a sequence of positive integer-valued random variables all
defined on \((\Omega, \mathcal{F}, \Pr)\). Define \(Z_{n} : \Omega \to
\mathbb{R}^{k}\) by \(Z_{n} (\omega) = X_{M_{n} (\omega)} (\omega)\). Define
\begin{equation*}
  \xi_{n} = \mathbb{E} \left[ Z_{n} \middle | M_{n} \right].
\end{equation*}
If \(\lim_{n \to \infty} \mathbb{E} \left[ \left\| X_{n} \right\|^{r} \right] =
0\) (i.e. if \(X_{n} \overset{\mathrm{L_{r}}}{\to} \mathbf{0}_{k}\)), then
\(\xi_{n} \overset{\mathrm{p}}{\to} 0\).
\end{lemma}

\begin{proof}
Let \(\varepsilon > 0\) be given. Then by \(X_{n} \overset{\mathrm{L_{r}}}{\to}
\mathbf{0}_{k}\), there must exist \(N_{\varepsilon} \in \mathbb{N}\) such that
\begin{equation*}
  \mathbb{E} \left[ \left\| X_{n} \right\|^{r} \right] \leq \varepsilon \quad
  \forall n \geq N_{\varepsilon}.
\end{equation*}
Then
\begin{align*}
  0 \leq \Pr \left( \left| \zeta_{n} \right| > \varepsilon \right) =
  & \ \Pr \left( \left| \zeta_{n} \right| > \varepsilon, M_{n} < N_{\varepsilon}
    \right) + \Pr \left( \left| \zeta_{n} \right| > \varepsilon, M_{n} \geq
    N_{\varepsilon} \right) \\
  =
  & \ \Pr \left( \left| \zeta_{n} \right| > \varepsilon, M_{n} < N_{\varepsilon}
    \right) + 0 \\
  \leq
  & \ \Pr \left( M_{n} < N_{\varepsilon} \right) \\
  \to & \ 0.
\end{align*}
The conclusion follows the Sandwich Theorem.
\end{proof}

%%% Local Variables:
%%% mode: latex
%%% TeX-master: "../2023_semipar_eff_car"
%%% End:
% LocalWords:  SSRA SPBR ATT

\end{document}